\pdfoutput=1
\newif\ifSubmission
\Submissiontrue 

\documentclass[acmsmall,screen]{acmart} 

\usepackage{booktabs}   
\usepackage{subcaption} 
\usepackage{xfrac}
\usepackage{tensor}
\usepackage{enumerate}
\usepackage{tikz}
\usetikzlibrary{automata,positioning,arrows}
\usetikzlibrary{snakes,arrows,decorations.pathmorphing,positioning,fit,trees,shapes,shadows,automata,calc}

\usepackage[utf8]{inputenc}
\usepackage[T1]{fontenc}
\usepackage{microtype}

\newcommand{\lleq}{~{}\leq{}~}

\newcommand{\eeq}{~{}={}~}
\newcommand{\nneq}{~{}\neq{}~}

\newcommand{\SKIP}{\mathtt{skip}}
\newcommand{\ABORT}{\mathtt{abort}}
\newcommand{\ASSIGN}[2]{{#1}\mathrel{\coloneqq}{#2}}
\newcommand{\ITE}[3]{\mathtt{if} \left( {#1} \right) \left\{ {#2} \right\} \text{else} \left\{ {#3} \right\}}

\newcommand{\PCHOICE}[3]{\left\{ {#1} \right\} \mathrel{\PC{#2}} \left\{ {#3} \right\}}
\newcommand{\NDCHOICE}[2]{\left\{ {#1} \right\} \mathrel{\Box} \left\{ {#2} \right\}}
\newcommand{\COMPOSE}[2]{{#1};\:{#2}}

\newcommand{\WHILEDO}[2]{\mathtt{while} \left( {#1}\right) \left\{ {#2} \right\}}
\newcommand{\WHILE}[1]{\mathtt{while} \left( {#1}\right) \{}

\newcommand{\RECDO}[2]{\left\{\mathtt{rec}~#1\kern.1em\Spot#2\right\}}
\newcommand{\REC}[1]{\left\{\mathtt{rec}~#1\kern.1em\Spot\right.}

\newcommand{\IF}[1]{\mathtt{if}~{#1}\colon}
\newcommand{\IFELSE}[3]{\mathtt{if}~{#1}\colon~#2~\mathtt{else}~#3}

\newcommand{\To}{\rightarrow}

\newcommand{\iverson}[1]{\left[ {#1} \right]}
\newcommand{\wpsymbol}{\mathsf{wp}}
\renewcommand{\wp}[2]{\wpsymbol\,\textbf{.}\,{#1}\,\textbf{.}\,{#2}}
\newcommand{\wpCom}[1]{\wpsymbol\,\textbf{.}\,{#1}}
\newcommand{\awpsymbol}{\mathsf{awp}}
\newcommand{\awp}[2]{\awpsymbol\,\textbf{.}\,{#1}\,\textbf{.}\,{#2}}
\newcommand{\one}{\boldsymbol{1}}
\newcommand{\zero}{\boldsymbol{0}}
\newcommand{\charwp}[3]{\tensor*[^{#1}_{#2}]{F}{_{{#3}}}}
\newcommand{\charwpn}[4]{\tensor*[^{#1}_{#2}]{F}{_{{#3}}^{#4}}}
\newcommand{\subst}[2]{[{#1}/{#2}]}
\DeclareMathOperator\lfp{\textsf{lfp}}
\DeclareMathOperator\mymin{\textsf{\upshape{min}}}
\DeclareMathOperator\mymax{\textsf{\upshape{max}}}
\renewcommand{\min}{\mymin}
\renewcommand{\max}{\mymax}

\newcommand{\Max}[2]{\max \left\{ {#1},\, {#2} \right\}}

\newcommand{\Nats}{\mathbb{N}}
\newcommand{\Ints}{\mathbb{Z}}
\newcommand{\Rpos}{\mathbb{R}_{\geq 0}}
\newcommand{\Rspos}{\mathbb{R}_{> 0}}
\newcommand{\Rgen}{\mathbb{R}}
\newcommand{\Rposinf}{\mathbb{R}_{\geq 0}^{\infty}}
\newcommand{\PGCL}{\mathsf{pGCL}}
\newcommand{\GCL}{\mathsf{GCL}}
\newcommand{\E}{\mathbb{E}}
\newcommand{\Vars}{\textsf{Vars}}

\newcommand{\true}{\textsf{true}}
\newcommand{\false}{\textsf{false}}

\newcommand\AS			{\textit{AS}}
\newcommand\AST			{\textit{AST}}
\newcommand\AT			{@} 
\newcommand\Defs			{{:=}\,}
\newcommand\Expec[2]		{{\mathcal E}_{#1}{#2}} 
\newcommand\False			{\textit{false}}
\newcommand\Implies		{\mathbin{\Rightarrow}} 
\newcommand\LHS			{\textit{lhs}	}
\newcommand\Mod			{\mathbin{\rm mod}}
\newcommand\NF[2]			{\nicefrac{#1}{#2}}
\newcommand\PC[1]			{\mathbin{{}_{#1}\kern-.05em\oplus}}
\newcommand\PCF			{\PC{\NF{1}{2}}}
\newcommand\pGCL		{\textit{pGCL}}
\newcommand\RHS			{\textit{rhs}}
\DeclareMathSymbol\RRightarrow{\mathrel}{AMSa}{"56}
\newcommand\True			{\textit{true}}
\newcommand\Wide[1] 		{\hspace*{2em}#1\hspace*{2em}} 
\newcommand\WIDERM[1]	{\Wide{\WideRm{#1}}} 
\newcommand\WideRm[1]	{\hspace*{2em}\textrm{#1}\hspace*{2em}} 
\newcommand\Spot		{\mathbin{\boldmath{\cdot}}} 

\usepackage{bbding}				
\usepackage{listings}			
\usepackage{units}				
\usepackage{wasysym}			

\def\Vhrulefill{\leavevmode\leaders\hrule height 0.7ex depth \dimexpr0.4pt-0.7ex\hfill\kern0pt}

\newcommand\ImageInText[4]%
{\makebox[0pt][l]{%
{\def\Up{#1}\def\Right{#2}
\hfill%
\raisebox{\Up}[0pt][0pt]{
\makebox[0pt][l]{\hspace{\Right}
\includegraphics[scale=#3]{#4}%
}}}}}

\newcommand\ImageInTextBlock[7]%
{\def\Height{#1}
 \def\Depth{#2}
 \def\Width{#3}
 \def\Up{#4}
 \def\Right{#5}
\raisebox{\Up}[\Height][\Depth]{\makebox[\Width][l]{\hspace{\Right}\includegraphics[scale=#6]{#7}
}}\\} 

\newcommand\FirstLineOfTheorem[1][\empty] {\ifx#1\empty\quad\else~\hrulefill~#1\\\fi}
\newcommand\MakeFact[2]{\newenvironment{#1}[2]{\begin{#2}\label{##2}\textit{##1}\rm\noexpand\noexpand\FirstLineOfTheorem}{\hfill$\Box$\end{#2}}}

\MakeFact{Definition}{definition}
\MakeFact{Corollary}{corollary}
\MakeFact{Lemma}{lemma}
\MakeFact{Theorem}{theorem}
\MakeFact{Example}{example}

\newcommand\Cor[1] {Cor.~\ref{#1}}

\newcommand\Def[1] {Def.~\ref{#1}}

\newcommand\Eqn[1] {(\ref{#1})}
\newcommand\Fig[2][\empty] {\ifx#1\empty Fig.~\ref{#2}\else Fig.~\ref{#2}(#1)\fi}
\newcommand\Itm[1] {(\ref{#1})}

\newcommand\Lem[1] {Lem.~\ref{#1}}
\newcommand\Note[1] {\FnSym\ref{#1}}

\newcommand\Sec[1] {\S\ref{#1}}
\newcommand\App[1] {App.~\ref{#1}}
\newcommand\Thm[1] {Thm.~\ref{#1}}

\lstnewenvironment{Prog} {
	\lstset{
		language=csh,
		basicstyle=\ttfamily,	
  		breaklines=false,	
  		showspaces=false,
  		keywords={vis,hid,var,while,if,end,begin,to,then,fi,skip},
  		columns=flexible,
  		escapechar=\%,	
  		mathescape=true
	}
}{}

\lstnewenvironment{ASCII} {
	\lstset{
		language=csh,
		basicstyle=\ttfamily,
  		showspaces=false,
  		breaklines=true,	
  		breakindent=0pt,
  		keywords={},
  		columns=flexible,
  		escapechar=\%,	
  		mathescape=true
	}
}{}

\newenvironment{Reason}{\begin{tabbing}\hspace{4em}\= \hspace{1cm} \= \kill}
    {\end{tabbing}\vspace{-1em}}
\newcommand\Step[2] {#1 \> $\begin{array}[t]{@{}llll}#2\end{array}$ \\}
\newcommand\StepR[3] {#1 \> $\begin{array}[t]{@{}llll}#3\end{array}$
    \` {\RF \makebox[0pt][r]{\begin{tabular}[t]{r}``#2''\end{tabular}}} \\}
\newcommand\WideStepR[3] {#1 \>
    $\begin{array}[t]{@{}ll}~\\#3\end{array}$ \`
    {\RF \makebox[0pt][r]{\begin{tabular}[t]{r}``#2''\end{tabular}}} \\}

\newcommand\RF {\small}

\usepackage{amssymb}			
\usepackage[normalem]{ulem}		
\usepackage{bbding}				
\usepackage{endnotes}			

\definecolor{PurplePlum}{rgb}{0.1,0,0.55} 
\definecolor{Brown}{rgb}{0.5,.25,0}
\definecolor{Orange}{rgb}{1,.6,0}
\definecolor{Gray}{rgb}{.7,.7,.7}
\definecolor{DarkGreen}{rgb}{0,.6,0}

\newif\ifBleck
\newcommand\Bleck {\Blecktrue} 

\newcommand\Colour[1] {\color{#1}}

\newcommand\ToCLinks {
 \ifBleck\else 
   \expandafter\def\csname@oddfoot\endcsname{
    {\Colour{blue}\mbox{\hyperlink{w1619}{\sf$\rightarrow$~top}\quad
    				    \hyperlink{w1031}{\sf$\rightarrow$~toc}\quad
    				    \hyperlink{w1148}{\sf$\rightarrow$~lof}\quad
                                      \hyperlink{GreenRoom}{\sf$\rightarrow$~gr}\quad
                                      \hyperlink{EndNotes}{\sf$\rightarrow$~en}\quad
                                      \hyperlink{Sargasso}{\sf$\rightarrow$~sg}\quad
                                      \hyperlink{GlossaryOfMacros}{\sf$\rightarrow$~gm}}\hfill
   }}
   \expandafter\def\csname@evenfoot\endcsname{
    {\Colour{blue}\mbox{\hyperlink{w1619}{\sf$\rightarrow$~top}\quad
                                      \hyperlink{w1031}{\sf$\rightarrow$~toc}\quad
    				    \hyperlink{w1148}{\sf$\rightarrow$~lof}\quad
                                      \hyperlink{GreenRoom}{\sf$\rightarrow$~gr}\quad
                                      \hyperlink{EndNotes}{\sf$\rightarrow$~en}\quad
                                      \hyperlink{Sargasso}{\sf$\rightarrow$~sg}\quad
                                      \hyperlink{GlossaryOfMacros}{\sf$\rightarrow$~gm}}\hfill
   }}
 \fi
}

\newif\ifEndNotes 

\newcommand\FnSym{{\scriptsize\PencilLeftDown\kern.1em}} 
\newcommand\EnSym {{$\bigtriangledown$}}

\newif\ifMakeMarkupsCalled 
\newif\ifSuppress 
\newcommand\MakeMarkups[3][.]{
 \Suppressfalse 
 \ifBleck\Suppresstrue\fi
 \ifx0#1\Suppresstrue\fi
 \ifx1#1\Suppressfalse\fi
 \expandafter\providecommand\csname#2x\endcsname {} 
 \ifSuppress\expandafter\renewcommand\csname#2x\endcsname{\relax}\else
                   \expandafter\renewcommand\csname#2x\endcsname{#3}\fi 
 \expandafter\providecommand\csname#2\endcsname {} 
 \ifSuppress\expandafter\renewcommand\csname#2\endcsname[1]{##1}\else
                   \expandafter\renewcommand\csname#2\endcsname[1]{{\csname#2x\endcsname##1}}\fi 
 \expandafter\providecommand\csname#2d\endcsname {} 
 \ifSuppress\expandafter\renewcommand\csname#2d\endcsname[1]{\relax}\else
                   \expandafter\renewcommand\csname#2d\endcsname[1]{{\csname#2x\endcsname\sout{##1}}}\fi 
 \expandafter\providecommand\csname#2r\endcsname {} 
 \ifSuppress\expandafter\renewcommand\csname#2r\endcsname[2]{{##2}}\else
                   \expandafter\renewcommand\csname#2r\endcsname[2]{\csname#2d\endcsname{##1} \csname#2\endcsname{##2}}\fi 
 \expandafter\providecommand\csname#2t\endcsname {} 
 \ifSuppress\expandafter\renewcommand\csname#2t\endcsname[1]{\relax}\else
                   \expandafter\renewcommand\csname#2t\endcsname[1]{{\csname#2x\endcsname{$\langle\!\langle$##1$\rangle\!\rangle$}}}\fi 
 \expandafter\providecommand\csname#2b\endcsname {} 
 \ifSuppress\expandafter\renewcommand\csname#2b\endcsname[1][empty]{\relax}\else 
                   \expandafter\renewcommand\csname#2b\endcsname[1][\empty]{\ifx\empty##1\empty
                   	\label{#2-bookmark} 
                          \marginpar [\raggedleft\csname#2\endcsname{{\footnotesize\fbox{#2 working here}}~$\Rightarrow$}]
                                            {\csname#2\endcsname{$\Leftarrow$~{\footnotesize\fbox{#2 working here}}}}
                   \else 
                   	\marginpar [\raggedleft\csname#2\endcsname{\ifx\empty##1\empty\else\fbox{\tiny\parbox{6em}{\raggedright##1}}~\fi$\Rightarrow$}]
                                            {\csname#2\endcsname{$\Leftarrow$\ifx\empty##1\empty\else~{\tiny\fbox{\parbox{6em}{\raggedright##1}}}\fi}}\fi}\fi
 \expandafter\providecommand\csname#2TD\endcsname {} 
 \ifSuppress\expandafter\renewcommand\csname#2TD\endcsname{\relax}\else
                   \expandafter\renewcommand\csname#2TD\endcsname{\csname#2\endcsname{\fbox{#2 to do}}}\fi 
 \expandafter\providecommand\csname#2Bar\endcsname {} 
 \ifSuppress\expandafter\renewcommand\csname#2Bar\endcsname{\relax}\else
                   \expandafter\renewcommand\csname#2Bar\endcsname{\csname#2\endcsname{\scriptsize\XSolidBrush}}\fi 
 \expandafter\providecommand\csname#2f\endcsname {} 
 \ifSuppress\expandafter\renewcommand\csname#2f\endcsname[2][]{\relax}\else
  \expandafter\renewcommand\csname#2f\endcsname[2][\empty]{ 
    {\mbox{\csname#2x\endcsname\tiny$\boxtimes$}\marginpar{\csname#2x\endcsname\fbox{\FnSym\footnotemark}}\relax 
    \footnotetext{\csname#2x\endcsname##2}}}\fi
 \expandafter\providecommand\csname#2e\endcsname {} 
 \ifSuppress\expandafter\renewcommand\csname#2e\endcsname[1]{\relax}\else
  \expandafter\renewcommand\csname#2e\endcsname[1]{
   \global\EndNotestrue 
   \mbox{\scriptsize\csname#2x\endcsname$\boxtimes$}\relax
   \marginpar{\csname#2x\endcsname\fbox{\EnSym\endnotemark
                      \hypertarget{ENmark\thepage-\theendnote}{}~\hyperlink{ENtext\thepage-\theendnote}{{\Colour{blue}$\downarrow$}}}
   }
   { 
    \def\zz{\noexpand#3}
    \edef\z{~{\zz[Endnote \theendnote\ on p.\noexpand\hypertarget{ENtext\thepage-\theendnote}{}\thepage
                ~\noexpand\hyperlink{ENmark\thepage-\theendnote}{{\noexpand\Colour{blue}$\uparrow$}}]}
    }
    \expandafter\endnotetext\expandafter{\z\vspace{2ex}\\ ##1\newpage}
   } 
  }\fi
 \expandafter\providecommand\csname#2fe\endcsname {} 
 \ifSuppress\expandafter\renewcommand\csname#2fe\endcsname[2][]{\relax}\else 
  \expandafter\renewcommand\csname#2fe\endcsname[2][]{ 
   \def\File{##1}\relax
   \ifx\File\empty\csname#2f\endcsname{##2}\else 
    \global\EndNotestrue 
    \mbox{\scriptsize\csname#2x\endcsname$\boxtimes$}
    \marginpar{\csname#2x\endcsname\fbox{\FnSym\footnotemark}}\relax
    \footnotetext{~\csname#2x\endcsname##2\
                         --- See [\EnSym\endnotemark\hypertarget{ENmark\thepage-\theendnote}{}
                         \kern-.2em\hyperlink{ENtext\thepage-\theendnote}{{\Colour{blue}$\downarrow$}}].}\relax
   { 
     \def\zz{\noexpand#3}
     \edef\z{~{\zz[Footnote~\thefootnote~on~p.\noexpand\hypertarget{ENtext\thepage-\theendnote}{}\thepage
                 ~\noexpand\hyperlink{ENmark\thepage-\theendnote}
                 {{\noexpand\Colour{blue}\kern-0.1em$\uparrow$}]}}
                 {\noexpand\footnotesize\noexpand\newline\noexpand\hspace*{2em} (~from file {\noexpand\tt\File.tex}~)}
     }    
     \expandafter\endnotetext\expandafter{\z~\par\input{##1}\newpage}
    } 
   \fi 
  } 
 \fi 
 \ifSuppress\relax\else
  \csname#2f\endcsname{
   $\backslash$\texttt{#2}$\cdots$\ markups are in \textbf{this} colour\ifx#1..\else\ifx1#1.\else, for #1.\fi\fi
   \ifMakeMarkupsCalled\relax\else 
    \begin{quote}\begin{tabular}{l@{\hspace{2em}}p{.8\linewidth}}
     \multicolumn{2}{l}{\texttt{$\backslash$MakeMarkups\ifx#1.\relax\else[#1]\fi\{#2\}\{{\it$\langle$colour command\/$\rangle$}\}}
     				 --- Defines the macros below:}\\
         & see comments at \texttt{$\backslash$MakeMarkups} definition. \\[1ex]
     \texttt{$\backslash$#2\{$\langle$text$\rangle$\}} & Sets \texttt{$\langle$text$\rangle$} in \texttt{#2}'s colour. \\
     \texttt{$\backslash$#2x} & Changes to \texttt{#2}'s colour (until end of context). \\
     \texttt{$\backslash$#2d\{$\langle$text$\rangle$\}} & Sets \texttt{$\langle$text$\rangle$} in \texttt{#2}'s colour with a strikethrough (i.e.\ delete). \\
     \texttt{$\backslash$#2r\{$\langle$this$\rangle$\}\{$\langle$that$\rangle$\}} &
      Strikes through \texttt{$\langle$this$\rangle$} and inserts \texttt{$\langle$that$\rangle$} (i.e.\ replace). \\
     \texttt{$\backslash$#2f\{$\langle$text$\rangle$\}} & Meta-comment: puts \texttt{$\langle$text$\rangle$} in a \texttt{#2}-footnote with a {\tiny$\boxtimes$} in the main text. \\
     \texttt{$\backslash$#2t\{$\langle$text$\rangle$\}} & Use for meta when  \texttt{$\backslash$#2f} isn't allowed (``Not in outer-par mode.'') \\
     \texttt{$\backslash$#2b[$\langle$optional$\rangle$]} & Marginal pointer, with label for hyper-linking directly there. \\
     \texttt{$\backslash$#2e\{$\langle$text$\rangle$\}} & Puts \texttt{$\langle$text$\rangle$} in a \texttt{#2}-endnote with a (big) $\boxtimes$ in the main text. \\[.5ex]
     \texttt{$\backslash$#2fe[$\langle$this$\rangle$]\{$\langle$that$\rangle$\}} & Makes a  \texttt{$\backslash$#2f\{$\langle$that$\rangle$\}} that refers to a \\
       & \texttt{$\backslash$#2e\{$\langle$contents of file this.tex$\rangle$\}}. \\ 
       & Without the optional argument, acts as \texttt{$\backslash$#2f\{$\langle$that$\rangle$\}}. \\[.5ex]
     \texttt{$\backslash$#2Bar} & Inserts ``burn after reading'' symbol \csname#2Bar\endcsname, meaning
      \begin{quote}\begin{itemize}\setlength\itemsep{0pt}
       \item If yours is the only \csname#2Bar\endcsname\ in this (presumably someone else's) footnote, and you are happy that the footnote has been addressed,
       go ahead and comment-out the whole footnote. (The \csname#2Bar\endcsname\ is their request for you to ``approve and remove''.)
       \item If you are not happy, delete only your \csname#2Bar\endcsname\ and follow-on in the footnote
        (in your colour, i.e.\ with \texttt{$\backslash$#2x}) saying why you are not happy.
       \item If you are happy, but there are others' burn-after-reading symbols as well as yours, just delete yours; the other people have not yet responded.
      \end{itemize}
      \end{quote}
      The idea is that when everyone's happy, the last person will comment-out the meta-text. \\[0.5ex]
     \texttt{$\backslash$#2TD} & Inserts {\csname#2TD\endcsname}\ . \\
    \end{tabular}\end{quote}
   \fi
  }
  \MakeMarkupsCalledtrue 
 \fi
}

\newcommand\Divider {~\vfill~\bigskip\hrule~\\\hrule~\bigskip} 

\newif\ifNoGreenRoom
\newcommand\MakeGreenRoom {\ifBleck\relax\else\ifNoGreenRoom\relax\else
 \hrule
 ~\\\begin{center}\Huge \hypertarget{GreenRoom}{Green Room}
 \end{center}~\\
 \hrule
\fi\fi}

\newif\ifNoEndNotes
\newcommand\MakeEndNotes {\ifBleck\relax\else\ifNoEndNotes\relax\else
 \hrule~\\\begin{center}\ifEndNotes\Huge Endnotes\else\textit{No endnote macros were called in this run.}\fi\end{center}~\\\hrule
 {\parindent 0pt \parskip 2ex \def\enotesize{\normalsize} \def\notesname{} 
 \ifEndNotes\theendnotes\fi} 
\fi\fi}

\newif\ifNoSargasso\newcommand\NoSargasso{\NoSargassotrue}
\newcommand\MakeSargasso {
 \hypertarget{Sargasso}{}
 \newcommand\NewLabel[1] {\OldLabel{Sargasso-##1}} 
 \newcommand\NewRef[1] 
 {\expandafter\ifx\csname r@Sargasso-##1\endcsname\relax\OldRef{##1}\else\OldRef{Sargasso-##1}\fi}
 \let\OldLabel\label \let\label\NewLabel
 \let\OldRef\ref \let\ref\NewRef
\ifBleck\end{document}\else\ifNoSargasso
\begin{center}\huge \textit{Sargasso suppressed in this run}\end{center}
\else
  \hrule
  ~\\\begin{center}\Huge Sargasso
  \end{center}~\\
  \hrule
 \fi\fi
}

\newcommand\EndDocument {\ifBleck\end{document}\fi} 

\newcommand\Cite[2][\empty] {{\color{red}\ifx#1\empty[#2]\else[#2,~#1]\fi}}

\NoSargasso 
\MakeMarkupsCalledtrue 
\ifSubmission\Bleck\fi

\acmPrice{}
\acmDOI{10.1145/3158121}
\acmYear{2018}
\copyrightyear{2018}
\acmJournal{PACMPL}
\acmVolume{2}
\acmNumber{POPL}
\acmArticle{33}
\acmMonth{1}

\begin{document}

\MakeMarkups[Carroll]{C}{\color{blue}}
\MakeMarkups[Final]{F}{\color{blue!60!black}} 
\MakeMarkups[Annabelle]{A}{\color{red}}
\MakeMarkups[Benjamin]{B}{\color{green!60!black}}
\MakeMarkups[Joost-Pieter]{J}{\color{orange}}
\MakeMarkups[Grey]{G}{\color{gray}}


\newif\ifShowLabels 
\newcommand\ShowLabels {\ShowLabelstrue}
\newcommand\NoShowLabels {\ShowLabelsfalse}

\newcommand\SL[1] {\mbox{~\tt\scriptsize[#1]~}} 
\newcommand\LabelWithShow[1] {\ifShowLabels\SL{#1}\fi\Label{#1}} 
\let\Label\label\let\label\LabelWithShow 

\newif\ifShowLabelsWas

\newenvironment{Equation}{\ShowLabelsWasfalse\ifShowLabels\ShowLabelsWastrue\fi\equation
\ifShowLabels\let\LabelWas\Label\let\Label\label\let\label\LabelWithShow\fi}{\endequation
\ifShowLabelsWas\let\label\Label\let\Label\LabelWas\fi\ShowLabelsWasfalse}

\newenvironment{Align}{\ShowLabelsWasfalse\ifShowLabels\ShowLabelsWastrue\fi\align
\ifShowLabels\let\LabelWas\Label\let\Label\label\let\label\LabelWithShow\fi}{\endalign
\ifShowLabelsWas\let\label\Label\let\Label\LabelWas\fi\ShowLabelsWasfalse}

\ifSubmission\relax\else\ShowLabels\fi
\newpage

\title[A New Proof Rule for Almost-Sure Termination]{A New Proof Rule for Almost-Sure Termination}


\author{Annabelle McIver}
\affiliation{
  \institution{Macquarie University}
  \country{Australia}
}
\email{annabelle.mciver@mq.edu.au}

\author{Carroll Morgan}
\affiliation{
  \institution{University of New South Wales}
  \country{Australia}
}
\affiliation{
  \institution{Data61, CSIRO}
  \country{Australia}
}
\email{carroll.morgan@unsw.edu.au}

\author{Benjamin Lucien Kaminski}
\affiliation{
 \institution{RWTH Aachen University}
 \country{Germany}
}
\affiliation{
  \institution{UCL}
  \country{UK}
}
\email{benjamin.kaminski@informatik.rwth-aachen.de}

\author{Joost-Pieter Katoen}
\affiliation{
 \institution{RWTH Aachen University}
 \country{Germany}
}
\affiliation{
 \institution{IST}
 \country{Austria}
}
\email{katoen@cs.rwth-aachen.de}

\begin{abstract}\F{
An important question for a probabilistic program is whether the probability mass of all its diverging runs is zero, that is that it terminates ``almost surely''. Proving that can be hard, and this paper presents a new method for doing so; it is expressed in a program logic, and so applies directly to source code. The programs may contain both probabilistic- and demonic choice, and the probabilistic choices may depend on the current state.\par
As do other researchers, we use variant functions (a.k.a.\ ``super-martingales'') that are real-valued and probabilistically might decrease on each loop iteration; but our key innovation is that the amount as well as the probability of the decrease are \emph{parametric}.\par
We prove the soundness of the new rule, indicate where its applicability goes beyond existing rules, and explain its connection to classical results on denumerable (non-demonic) Markov chains.}
\end{abstract}

%



\setcopyright{acmlicensed}
\acmJournal{PACMPL}
\acmYear{2018} \acmVolume{2} \acmNumber{POPL} \acmArticle{33} \acmMonth{1} \acmPrice{}\acmDOI{10.1145/3158121}

\maketitle

\section{Introduction}\label{s0906}
This paper concerns termination proofs for sequential, imperative \emph{probabilistic} programs, i.e.\ those that, in addition to the usual constructs, include a binary operator for probabilistic choice. Writing ``standard'' to mean ``non-probabilistic'', we recall that the standard technique for loop termination is to find an \emph{integer-valued} function over the program's state space, a ``variant'', that satisfies the ``progress'' condition that each iteration is guaranteed to decrease the variant strictly and further that the loop guard and invariant imply that the variant is bounded below by a constant (typically zero). Thus it cannot continually decrease without eventually making the guard false; and so existence of such a variant implies the loop's termination.

For probabilistic programs, the definition of loop termination is often weakened to ``almost-sure termination'', or ``termination with probability one'', by which is meant that (only) the probability of the loop's iterating forever is zero. For example if you flip a fair coin repeatedly until you get heads, \F{it is almost sure that you will eventually stop --- for the probability of flipping tails forever is $\NF{1}{2}{\cdot}\NF{1}{2}{\cdots}$, i.e.\ zero}. We will write \AS\ for ``almost sure'' and \AST\ for ``almost-sure termination'' or ``almost-surely terminating''.

But the standard variant rule we mentioned above is too weak for \AST\ in general. Write $\mathit{Com}\PC{p}\mathit{Com}'$ for choice of $\mathit{Com},\mathit{Com}'$ with probability $p,1{-}p$ resp.\ and consider the \AST\ program
\begin{Equation}\label{e1029}
	\COMPOSE{\ASSIGN{x}{1}}{\quad\WHILEDO{x{\neq}0}{~\ASSIGN{x}{(x{+}1)\Mod3}
	~\PCF~ \ASSIGN{x}{(x{-}1)\Mod3}~}}~.
\end{Equation}
It has no standard variant, because that variant would have to be decreased strictly by both updates to $x$. Also the simple \AST\ program
\begin{Equation}\label{e1030}
	\textrm{1dSRW:}\hspace{3em}
	\WHILEDO{x{\neq}0}{~\ASSIGN{x}{x{+}1} ~\PCF~ \ASSIGN{x}{x{-}1}~} ~,
\end{Equation}
the \emph{symmetric random walk} over integers $x$, is beyond the reach of the standard rule.

Thus we need \AST-rules for properly probabilistic programs, and indeed many exist already. One such, designed to be as close as possible to the standard rule, is that an integer-valued variant must be bounded \emph{above} as well as below, and its strict decrease need only occur with \emph{non-zero probability} on each iteration, i.e.\ not necessarily every time \citep[Lem.2.7.1]{McIver:05a}.
\footnote{Over an infinite state space, the second condition becomes ``with some probability bounded away from zero''.}
That rule suffices for Program \Eqn{e1029} above, with variant $x$ and upper bound 2; but still it does not suffice for Program \Eqn{e1030}.

The 1dSRW is however an elementary Markov process, and it is frustrating that a simple termination rule like the above (and some others' rules too) cannot deal with its \AST. This (and other examples) has led to many variations in the design of \AST-rules, a competition in which the rules' assumptions are weakened as much as one dares, to increase their applicability beyond what one's colleagues can do; and yet of course the assumptions must not be weakened so much that the rule becomes unsound. This is our first principal \textbf{Theme (A)} --- the power of \AST-rules.
 
A second \textbf{Theme (B)} in the design of \AST-rules is their applicability at the source level (of program texts), i.e.\ whether they are expressible and provable in a (probabilistic) program logic without ``descending into the model''.
We discuss that practical issue in \Sec{s1314} and \App{s1031} --- it is important e.g.\ for enabling theorem proving.

Finally, a third \textbf{Theme (C)} is the characterisation of the kinds of iteration for which a given rule is guaranteed to work, i.e.\ a completeness result stating for which \AST\ programs a variant is guaranteed to exist, even if it is hard to find. Typical characterisations are ``over a finite state space''  \citep{Hart:83},\citep[Lem.\,7.6.1]{McIver:05a} or ``with finite expected time to termination'' \citep{Fioriti:2015}.
\footnote{\Fx The difficult feature of the 1dSRW is that its expected time to termination is infinite.}

\textbf{The contribution of this paper} is to cover those three themes. We give a novel rule for \AST, one that:
(A) proves almost-sure termination in some cases that lie beyond what some other rules can do; (B) is applicable directly at the source level to probabilistic programs \emph{even if they include demonic choice}, for which we give examples; and (C) is supported by mathematical results from pre- computer-science days that even give some limited completeness criteria. In particular, one of those classical works shows that our new rule  must work for the \emph{two-dimensional} random walk: a variant is guaranteed to exist, and to satisfy all our criteria. That guarantee notwithstanding, we have yet to find a 2dSRW-variant in closed form.

\section{Overview}\label{s1314}
Expressed very informally, the new rule is this:
\begin{quote}
Find a non-negative \emph{real-valued} variant function $V$ of the state such that: (1) iteration cannot increase $V$'s expected value; (2) on each iteration the actual value $v$ of $V$ must decrease by at least $d(v)$ with probability at least $p(v)$ for some fixed non-increasing strictly positive real-valued functions $d,p$;
\footnote{As \Sec{s1643} explains, functions $d,p$ must have those properties for \emph{all} positive reals, not only the $v$'s that are reachable.}
and (3) iteration must cease if $v{=}0$.
\end{quote}
The formal statement of the rule, and a more detailed but still informal explanation, is given in \Sec{s0946}.

Section \ref{s1203} gives notation, and a brief summary of the programming logic we use. Section \ref{s1656} uses that logic to prove the new rule rigorously; thus we do not reason about transition systems directly in our proof. Instead we rely on the logic's \emph{being valid} for transition systems (e.g.\ valid for Markov decision processes), for the following two reasons:

\begin{description}
	\item[\rm Recall Theme (A) ---] The programming logic we use --its theorems to which we appeal-- are valid even for programs that contain demonic choice. And so our result is valid for demonic choice as well. (In \Sec{s1130} and \App{a1713} we discuss the degree of demonic choice that is permitted.)
	\item[\rm Recall Theme (B) ---] Expressing the termination rule in terms of a programming logic means that it can be applied to source code directly and that theorems can be (machine-) proved about it: there is no need to translate the program first into a transition system or any other formalism. The logic we use is a probabilistic generalisation of (standard) Hoare/Dijkstra logic \citep{Dijkstra:76}, due to \citet{Kozen:85} and later extended by \citet{Morgan:96d} and \citet{McIver:05a} to (re-)incorporate demonic choice.
\end{description}

Section \ref{s0918} carefully applies the rule to several small examples, illustrating its power and the logical manipulations it induces. Section \ref{s1530} explores the classical literature on \AST. Section \ref{s1521} examines other contemporary \AST\ rules. Section \ref{s1648} treats some theoretical aspects and limitations.

\section{Preliminaries}\label{s1203}
\subsection{Programming Language and Semantics}\label{s0940}

\F{$\PGCL$ is an imperative language based on Dijkstra's guarded command language $\GCL$ \citeyear{Dijkstra:76} but with an additional operator of binary probabilistic choice $\PC{p}$ introduced by \citet{Kozen:85} and extended by \citet{Morgan:96d} and \citet{McIver:05a} to restore demonic choice: the combination of the two allows one easily to write ``with probability no more than, or no less than, or between.''}
\footnote{\Fx Kozen's ground-breaking work \emph{replaced} demonic choice with probabilistic choice.}
Its forward, operational model is functions from states to sets of discrete distributions on states, \F{where any non-singleton sets represent demonic nondeterminism}: this is essentially Markov decision processes, but also probabilistic/demonic transition systems. (In \Sec{s1130} we describe some of the conditions imposed on the ``demonic'' sets.) Its backwards, logical model is functions from so-called ``post-expectations'' to ``pre-expectations'', non-negative real valued functions on the state that generalise the postconditions and preconditions of Hoare/Dijkstra \citep{Hoare:69} that are Boolean functions on the state: that innovation, and the original link between the forwards and backwards semantics, due to \citet{Kozen:85} but using our terminology here, is that $A=\wp{\mathit{Com}}{B}$, for $\PGCL$ program $\mathit{Com}$ and post-expectation $B$, means that pre-expectation $A$ is a function that gives for every initial state the expected value of $B$ in the final distribution reached by executing $\mathit{Com}$. The demonic generalisation of that \citep{Morgan:96d,McIver:05a} is that $A$ gives the \emph{infimum} over all possible final distributions of $B$'s expected value. Both of these generalise the ``standard'' Boolean interpretation exactly if false is interpreted as zero, true as one, implication as $(\leq)$ and therefore conjunction as infimum.

$\PGCL$'s weakest pre-expectation logic, like Dijkstra's weakest precondition logic, is designed to be applied at the source-code level of programs, as the case studies in \Sec{s0918} illustrate. Its theorems etc.\ are also expressed at the source-code level, but apply of course to whatever semantics into which the logic is (validly) interpreted.

We now set out more precisely the framework in which we operate. Let $\Sigma$ be the set of program states. We call a subset $G$ of $\Sigma$ a \textit{predicate}, equivalently a function from $\Sigma$ to the Booleans. If $\Sigma$ is the Cartesian product of named-variable types, we can describe functions on $\Sigma$ as expressions in which those variables appear free, and predicates are then Boolean-valued expressions.

We use Iverson bracket notation $\iverson{G}$ to denote the \textit{indicator function} of a predicate $G$, that is with value 1 on those states where $G$ holds and 0 otherwise.

An \textit{expectation} is a random variable that maps program states to non-negative reals:

\begin{definition}[Expectations~\textnormal{\citep{McIver:05a}}] \label{def:expectations}
	The set of expectations on $\Sigma$, denoted by $\E$, is defined as $\left\{f ~\middle|~ f\colon \Sigma \To \Rpos \cup \{\infty\}\right\}$ ~. We say that expectation $f$ is \textit{bounded} iff there exists a (non-negative) real $b$ such that $f(\sigma) \leq b$ for all states $\sigma$. The natural complete partial order $\leq$ on $\E$ is obtained by pointwise lifting, that is
\begin{align*}
	f_1\leq f_2 \WIDERM{iff} \forall \sigma\in\Sigma\colon\quad f_1(\sigma)\leq f_2(\sigma) ~. \tag*{$\triangle$}
\end{align*}
\end{definition}%
\noindent 
Thus Iverson brackets $\iverson{-}$ map predicates to expectations, and $(\Implies)$ to $(\leq)$ similarly --- that is, we have  $\iverson{A}{\leq}\iverson{B}$ just when $A{\Implies}B$.

Following \citet{Kozen:85}, here we are are based on Dijkstra's guarded-command language $\GCL$ \citep{Dijkstra:76}, but it is extended with a probabilistic-choice operator $\PC{p}$ between program (fragments) that chooses its left operand with probability $p$ (and its right complementarily). Beyond Kozen however, we use $\PGCL$ where demonic choice is \emph{retained} \citep{Morgan:96d,McIver:05a} --- i.e.\ $\PGCL$ contains \emph{both} probabilistic- and demonic choice. The syntax of $\PGCL$ is given in \autoref{table:wp}, and its semantics of \emph{expectation transformers}, the generalisation of predicate transformers, is defined as follows:

\begin{table}[t]
\renewcommand{\arraystretch}{1.5}
\caption{Rules for the expectation-transformer $\wpsymbol$.}
\begin{tabular}{@{\hspace{1em}}l@{\hspace{2em}}l@{\hspace{1em}}}
	\hline\hline
	$\boldsymbol{C}$			& $\boldsymbol{\textbf{\textsf{wp}}\,\textbf{.}\,C\,\textbf{.}\, f}$\\
	\hline\hline
	$\SKIP$					& $f$ \\
	$\ASSIGN{x}{e}$			& $f\subst{x}{e}$ \\
	$\ITE{G}{C_1}{C_2}$		& $\iverson{G} \cdot \wp{C_1}{f} + \iverson{\neg G} \cdot \wp{C_2}{f}$ \\
	$\PCHOICE{C_1}{p}{C_2}$	& $p \cdot \wp{C_1}{f} + (1 - p) \cdot \wp{C_2}{f}$ \\
	$\NDCHOICE{C_1}{C_2}$	& $\min\left\{ \wp{C_1}{f},\, \wp{C_2}{f} \right\}$ \\
	$\COMPOSE{C_1}{C_2}$		& $\wp{C_1}{\big(\wp{C_2}{f}\big)}$ \\
	 $\WHILEDO{G}{C'}$		& $\lfp X\textbf{.}~\iverson{\neg G} \cdot f + \iverson{G} \cdot \wp{C'}{X}$\\
	\hline
\end{tabular}

\begin{quote}
\par\bigskip In the table above $C$ is a $\PGCL$ program, and $f$ is an expectation. The notation $f\subst{x}{e}$ is function $f$ overridden at argument  $x$ by the value $e$.
A period ``$.$'' denotes (Curried) function application, so that for example $\wp{C_1}{f}$ is semantic-function $\wpsymbol$ applied to the syntax $C_1$;
the resulting transformer is then applied to the ``post-expectation'' $f$.\/\
A centred dot is multiplication, either of scalars or of an expectation by a scalar.

In $\PC{p}$ the probability $p$ can be an expression in the program variables (equivalently a $[0,1]$-valued function of $\Sigma$). Often however it is a constant.

The operator $\Box$ is demonic choice.
\end{quote}
\label{table:wp}
\end{table}

\begin{definition}[The $\wpsymbol$-Transformer~\textnormal{\citep{McIver:05a}}] \label{def:wp}
	The weakest pre-expectation transformer semantic function $\wpsymbol\colon \PGCL \To (\E \To \E)$ is defined in \autoref{table:wp} by induction on all $\PGCL$ programs.
\hfill$\triangle$
\end{definition}

If $f$ is an expectation on the \emph{final} state, then $\wp{\mathit{Com}}{f}$ is an expectation on the \emph{initial} state: thus $\wp{\mathit{Com}}{f}\textbf{.}\,\sigma$ is the infimum, over all distributions of final states that $\mathit{Com}$ can reach from $\sigma$\C{,} of the expected value of $f$ on each of them: there will be more than one just when $\mathit{Com}$ contains demonic choice. In the special case where $f$ is $\iverson{B}$ for predicate $B$, that value is thus the least guaranteed probability with which $\mathit{Com}$ from $\sigma$ will reach a final state satisfying $B$.

The natural connection between the standard world of predicate transformers (Dijkstra) and the probabilistic expectation transformers (Kozen/$\PGCL$) is the indicator function: for example $\iverson{\False}$ is $0$ and $\iverson{\True}$ is $1$,
\footnote{We will blur the distinction between Booleans and constant predicates, so that $\False$ is just as well the predicate that holds for no state. The same applies to reals and constant expectations.}
and the predicate implication $A\Implies B$ is equivalent to the expectation inequality $\iverson{A}\leq\iverson{B}$. The standard $A\Implies\wp{\mathit{Com}}{B}$, using standard $\wpsymbol$ and program $\mathit{Com}$ (i.e.\ without probabilistic choice in $\mathit{Com}$), becomes $\iverson{A}\leq\wp{\mathit{Com}}{\iverson{B}}$ when using the $\wpsymbol$ we adopt here. Finally, the idiom
\begin{align} \label{e0804}
	p\cdot\iverson{A} \Wide{\leq} \wp{\mathit{Com}}{\iverson{B}}\ ,
\end{align}
where ``$\cdot$'' is real-valued multiplication (pointwise lifted if necessary), means ``with probability at least $p$ the program $\mathit{Com}$ will take an initial state satisfying $A$ to a final state satisfying $B$'', where $p$ is a $[0,1]$-valued expression on (or equivalently a function of) the program state: in most cases however $p$ is constant. (See \App{a1221}.) This is because if the initial state $\sigma$ does not satisfy $A$, i.e.\ $A(\sigma)$ is $\False$, then the \LHS of \Eqn{e0804} is zero so that the inequality is trivially true; and if $\sigma$ does satisfy $A$ then the \LHS is $p\cdot1=p$ (or $p(\sigma)$ more generally) and the \RHS\ is the least guaranteed probability of reaching $B$, because the expected value of $\iverson{B}$ over a distribution is the probability that distribution assigns to $B$. (The ``least'' is, again, because of possible demonic nondeterminism.)

There are many properties of $\PGCL$'s probabilistic $\wpsymbol$ that are analogues of $\wpsymbol$ for standard programs; but one that is \emph{not} an analogue is ``scaling'' \citep[Def.~1.6.2]{McIver:05a}, an intrinsically numeric property whose justification rests ultimately on the distribution of multiplication through expected value from elementary probability theory. For us it is that for all commands $\mathit{Com}$, post-expectations $\mathit{Post}$ and non-negative reals $c$ we have
\begin{align}\label{e1752}
	\wp{\mathit{Com}}{(c\cdot\mathit{Post})} \Wide{=} c\cdot (\wp{\mathit{Com}}{\mathit{Post}})~.
\end{align}
We use it in the proof of \Thm{t1651} below. (See also \App{a0950}.)

\subsection{Probabilistic Invariants, Variants, and Termination with Probability 1}\label{s1559}
With the above correspondence, the following probabilistic analogues of standard termination and invariants are natural.

\begin{definition}[Probabilistic Invariants~\protect{~\textnormal{\citep[\textnormal{p.\ 39, Definition 2.2.1}]
{McIver:05a}}}]\label{d0923}
	Let $\mathit{Guard}$ be a predicate, a loop guard, and $\mathit{Com}$ be a $\PGCL$ program, a loop body. Then bounded expectation $\mathit{Inv}$ is a \textit{probabilistic invariant} of the loop\quad$\WHILEDO{\mathit{Guard}}{\mathit{Com}}$\quad just when
\begin{align} \label{e1209}
	\iverson{\mathit{Guard}} \cdot \mathit{Inv} \Wide{\leq} \wp{\mathit{Com}}{\mathit{Inv}}~. 
\end{align}
In this case we say that $\mathit{Inv}$ is \emph{preserved} by each iteration of $\WHILEDO{\mathit{Guard}}{\mathit{Com}}$.
\footnote{If (real valued) expectation $\mathit{Inv}$ were equal to $\iverson{\mathit{Inv}'}$ for some predicate $\mathit{Inv}'$, we'd have $\iverson{\mathit{Guard}\land \mathit{Inv}'} \leq \wp{\mathit{Com}}{\iverson{\mathit{Inv}'}}$, exactly the standard meaning of ``preserves $\mathit{Inv}'$''.}
\hfill$\triangle$
\end{definition}%
When some predicate $\mathit{Inv}'$ is such that $\mathit{Inv} = \iverson{\mathit{Inv'}}$ is a probabilistic invariant, we can equivalently say that $\mathit{Inv}'$ itself is a \emph{standard invariant} (predicate).
\footnote{\label{n1615}For any standard program $\mathit{Com}$, i.e.\ without probabilistic choice, Dijkstra's $\GCL$ judgement\quad $\mathit{Inv}\Implies\wp{\mathit{Com}}{\mathit{Inv}}$\quad is equivalent to our $\PGCL$ judgement\quad $\iverson{\mathit{Inv}}\leq\wp{\mathit{Com}}{\iverson{\mathit{Inv}}}$\quad for any predicate $\mathit{Inv}$.}

In \Sec{s0906} we recalled that the standard method of proving (standard) loop termination is to find an integer-valued variant function $\mathit{VInt}$ on the state such that the loop's guard (and the invariant, if one is given) imply  that $\mathit{VInt}{\geq}0$ and that $\mathit{VInt}$ strictly decreases on each iteration. A probabilistic analogue of loop termination is ``terminates with probability one'', i.e.\ terminates almost-surely, and one (of many) probabilistic analogue(s) of the standard loop-termination rule is the following:

\begin{theorem}[Variant rule for loops~\protect{(existing:~\textnormal{\citep[\textnormal{p.\ 55, Lemma 2.7.1}]{McIver:05a}})}]
\label{thm:lem-2-7-1}
	Let $\mathit{Inv}, \mathit{Guard} \subseteq \Sigma$ be predicates;
	let $\mathit{VInt}\colon\Sigma{\To}\Ints$ be an integer-valued function on the state space;
	let $\mathit{Low}, \mathit{High}$ be fixed integers;
	let $0{<}\varepsilon{\leq}1$ be a fixed strictly positive probability
		that bounds away from zero the probability that $\mathit{VInt}$ decreases; and
	let $\mathit{Com}$ be a $\PGCL$ program.
	Then the three conditions are
	\begin{enumerate}[(i)]
		\item\label{i1657-2}
			$\mathit{Inv}$ is a standard invariant (equiv.\ $\iverson{\mathit{Inv}}$ an invariant)
			of\quad$\WHILEDO{\mathit{Guard}}{\mathit{Com}}$\,, and
		\item\label{i1657-1}
			$\mathit{Guard}\land\mathit{Inv} \,\Implies\, \mathit{Low}{<}\mathit{VInt}{\leq}\mathit{High}$, and
			\footnote{The original rule \citep[Lem.~2.7.1]{McIver:05a} had $\mathit{Low}{\leq}\mathit{VInt}{<}\mathit{High}$. We make this inessential change for later neatness.}
		\item\label{i1657-3}
			for any constant integer $N$ we have\quad$\varepsilon \cdot \iverson{\mathit{Guard} \land Inv \land\mathit{VInt}{=}N} ~\leq~ \wp{\mathit{Com}}{\iverson{\mathit{VInt}{<}N}}$
	\end{enumerate}
	and, when taken all together, they imply
	\quad$\iverson{\mathit{Inv}} \leq \wp{\WHILEDO{\mathit{Guard}}{\mathit{Com}}}{\one}$~,
	that from any initial state satisfying $\mathit{Inv}$ the loop terminates \AS.
\end{theorem}%
The ``\emph{for any integer $N$}'' in \Eqn{i1657-3} above is the usual Hoare-logic technique for capturing an expression's initial value (in this case $\mathit{VInt}$'s) for use in the postcondition: we can write ``$\mathit{VInt}{<}N$'' there for ``the current value $\mathit{VInt}$, here in the final state, is strictly less than the value $N$ it had in the initial state.''
\footnote{In greater detail: if the universally quantified $N$ is instantiated to anything other than $\mathit{VInt}$'s initial value then the left-hand side of  \Itm{i1657-3} is zero, satisfying the inequality trivially since the right-hand side is non-negative by definition of expectations.}
Recalling \Eqn{e0804}, we see that assumption \Itm{i1657-3} thus reads
\begin{quote}
	On every iteration $\mathit{Com}$ of the loop the variant $\mathit{VInt}$
	is guaranteed to decrease strictly with probability at least some (fixed) strictly positive $\varepsilon$~.
\end{quote}

The probabilistic variant rule above differs from the standard rule in two essential respects: the probabilistic variant must be bounded \emph{above} as well as below (which tends to make the rule weaker); and the decrease need not be certain, rather only bounded away from zero (which tends to make the rule stronger). Although this rule does have wide applicability \citep[Chp.~3]{McIver:05a}, it nevertheless is not sufficient for example to show \AST\ of the symmetric random walk, Program \Eqn{e1030}.
\footnote{Any variant that works for \citep[\textnormal{p.\ 55, Lemma 2.7.1}]{McIver:05a} must be bounded above and -below, and integer-valued. And it must be able (with some non-zero probability) to decrease strictly on each step. If its bounds were say $L,H$, then it must therefore be able to terminate from \emph{anywhere} in no more than $H{-}L$ steps, a fixed and finite number. But \Eqn{e1030} does not have that property.}

The advance incorporated in our new rule, as explained in the next section, is to \emph{strengthen \Thm{thm:lem-2-7-1} in three ways}: (1) we remove the need for an upper bound on the variant; (2) we allow the probability $\varepsilon$ to vary; and (3) we allow the variant to be real-valued.  (\Thm{thm:lem-2-7-1} is itself used as a lemma in the proof of soundness of the new rule.)

We will need the following theorem, a probabilistic analogue of the standard technique that partial correctness plus termination gives total correctness, and with similar significance: proving ``only'' that a standard loop terminates certainly indeed does not necessarily give information about the loop's efficiency; but the termination proof is still an essential prerequisite for other proofs about the loop's functional correctness. The same applies in the probabilistic case.
\begin{theorem}[Almost-sure termination for probabilistic loops~
	\protect{(existing: ~\textnormal{\citep[\textnormal{p.\ 43, Lemma 2.4.1, Case 2.}]{McIver:05a}})}]
	\label{thm:lem-2-4-1}\Label{T0909}
	
	Let $\mathit{Term}$
satisfy
	\quad$\iverson{\mathit{Term}} \leq \wp{\WHILEDO{\mathit{Guard}}{\mathit{Com}}}{\one}$\,, that is that from any initial state satisfying $\mathit{Term}$ the 	loop terminates \AS\ (termination), and let \emph{\underline{bounded}} expectation $\mathit{Sub}$ be preserved by $\mathit{Com}$ whenever $\mathit{Guard}$ holds, i.e.\ it is a probabilistic invariant of \quad$\WHILEDO{\mathit{Guard}}{\mathit{Com}}$\quad (partial correctness).
	Then 
	\begin{align*}
		\iverson{\mathit{Term}} \cdot \mathit{Sub}
		\quad\leq\quad
		\wp{\WHILEDO{\mathit{Guard}}{\mathit{Com}}}{(\iverson{\neg\mathit{Guard}} \cdot \mathit{Sub})}~. 
		\tag{total correctness}
	\end{align*}
\end{theorem}%
The intuitive import of this theorem is that if bounded $\mathit{Sub}$ is a probabilistic invariant preserved by each iteration of the loop body, then also the whole loop ``preserves" $\mathit{Sub}$ from any state where the loop's termination is \AS. This holds even if $\mathit{Com}$ contains demonic choice.

Bounding $\mathit{Sub}$ is required by \citep{McIver:05a}, where \Thm{thm:lem-2-4-1} is found, and it is necessary here (\Sec{s1413}).

\section{A New Proof Rule for Almost-Sure Termination}\label{s1153}
\subsection{Martingales}
Important for us in extending the \AST\ rule is reasoning about ``sub- and super-martingales''.

A \emph{martingale} is a sequence of random variables for which the expected value of each random variable next in the sequence is equal to the current value (irrespective of any earlier values). A \emph{super}-martingale is more general: the current value may be larger than the expected subsequent value; and a \emph{sub}-martingale is the complementary generalisation. In probabilistic programs, as we treat them here, such a sequence of random variables is some expectation evaluated over the succession of program states as a loop executes,
and an exact/super/sub -martingale is an expectation whose exact value at the beginning of an iteration (a single state) is equal-to/no-less-than/no-more-than its expected value at the end of that iteration.

A trivial example of a sub-martingale is the invariant predicate of a loop in standard programming, provided we interpret $\False{\leq}\True$, for if the invariant is true at the beginning of the loop body it must be true at the end --- provided the loop guard is true. More generally in \Def{d0923} above we defined a probabilistic invariant, and at \Itm{e1209} there we see that it is a sub-martingale, again provided the loop guard holds. (If the loop guard does not hold, then $\iverson{G}$ is 0 and the inequality is trivial.) To take the loop guard $G$ into account, we say in that case that $\mathit{Inv}$ is a \emph{sub-martingale \underline{on $G$}}.

\subsection{Introduction, Informal Explanation and Example of the New Rule}\label{s0946}
The new rule is presented here, with an informal explanation; just below it we highlight the way in which it differs from the existing rule referred to in \Thm{thm:lem-2-7-1}; then we give an overview of the new rule's proof; and finally we give an informal example. The detailed proof follows in  Section \Sec{s1656}, and fully worked-out examples are given in \Sec{s0918}. To distinguish material in this section from the earlier rules above, here we use single-letter identifiers for predicates and expectations.

We say that a function is \emph{antitone} just when $x{\leq}y \Implies f(x){\geq}f(y)$ for all $x,y$.

\begin{theorem}[New Variant Rule for Loops]\label{t1651}\label{t1928}
\label{T1928} 
	Let $I,G \subseteq \Sigma$ be predicates;
	let $\mathit{V}\colon \Sigma{\To}\Rpos$ be a non-negative real-valued function
	\emph{\underline{not} necessarily bounded};
	let $p$ (for ``probability'') be a fixed function of type $\Rpos{\To}(0,1]$;
	let $d$ (for ``decrease") be a fixed function of type $\Rpos{\To}\Rspos$,
		both of them antitone on strictly positive arguments; and
	let $\mathit{Com}$ be a $\PGCL$ program.
	
	Suppose the following four conditions hold:
	\begin{enumerate}[(i)]
		\item\label{i1651-2}
			$I$ is a standard
			invariant of \quad$\WHILEDO{G}{\mathit{Com}}$~, and
		\item\label{i1651-1}
			$G\land I\Implies V{>}0$~, and
		\item\label{i1651-3}
			For any $R{\in}\Rspos$ we have
			$p(R) \cdot \iverson{G \land I \land V{=}R}
			~\leq~
			\wp{\mathit{Com}}{\iverson{\mathit{V} \leq R{-}d(R)}}$~, and
		\item\label{i1651-4}
			$V$ satisfies the ``super-martingale'' condition that
			\begin{align*}
				\textrm{for any constant $H$ in $\Rspos$ we have} \hspace{3em}
				\iverson{G\land I} \cdot (H{\ominus}V) \leq \wp{\mathit{Com}}{(H{\ominus}V)}~,
			\end{align*}
			where $H{\ominus}V$ is defined as $\Max{H{-}V}{0}$.
	\end{enumerate}
	Then we have\quad
	$\iverson{I} \leq \wp{\WHILEDO{G}{\mathit{Com}}}{1}$\quad, \F{i.e.\ \AST\ from any initial state satisfying $I$.}
\end{theorem}%
Note that our theorem is stated (and will be proved) in terms of $H{\ominus}V$. Our justification however for calling \textit{\Itm{i1651-4}} a ``super-martingale condition'' on $V$ is that decrease (in expectation) of\kern.3em$V$ is equivalent to increase of $H{\ominus}V$. (\App{a1216} gives more detail.) Further, in our coming appeal to \Thm{thm:lem-2-4-1} the expectation $\mathit{Sub}$ must be bounded --- and $V$ is not (necessarily). Thus we use $H{\ominus}V$ for arbitrary $H$ instead, each instance of which is bounded by $H$; and $V$ decreases when $H{\ominus}V$ increases.
			
The other reason for using the ``inverted'' formulation is that $\PGCL$ interprets demonic choice by \emph{minimising} over possible final distributions, and so the direction of the inequality in \Thm{thm:lem-2-4-1} means we must express the ``super-martingale property'' of\/$V$ in this complementary way.

As in \Thm{thm:lem-2-7-1}\Itm{i1657-3}, we have written in the Hoare style $V{=}R$ in the pre-expectation at \Itm{i1651-3} above to make $V$'s initial value available (as the real $R$\/) in the post-expectation. The overall effect is
	\begin{quote}
	If a predicate $I$ is a standard invariant,
	and there is a non-negative real-valued variant function $V$, on the state,
	that is a super-martingale on $I{\land}G$ with the progress condition
	that every iteration $\mathit{Com}$ of the loop
	decreases it by at least $d()$ of its initial value with probability at least $p()$ of its initial value,
	then the loop\quad$\WHILEDO{G}{\mathit{Com}}$\quad
	terminates \AS\ from any inital state satisfying $I$.
	\end{quote}

The differences from the earlier variant rule \Thm{thm:lem-2-7-1} are these:
\begin{enumerate}
	\item The variant $V$ is now real-valued, with no upper bound (but is bounded below by zero). We call $V$ a \emph{quasi-variant} to distinguish it from traditional integer-valued variants.
	\item Quasi-variants are \emph{not} required to decrease by a fixed non-zero amount with a fixed non-zero probability. Instead there are two functions $p,d$ that give for each variant-value how much $\textit{Com}$ must decrease it (at least) and with what probability (at least). The only restriction on those functions (aside from the obvious ones) is that they be antitone, i.e.\ that for larger arguments they must give equal-or-smaller (but never zero) values. The reason for requiring $p$ and $d$ to be antitone is to exclude Zeno-like behavior where the variant decreases less and less, and/or with less and less probability. Otherwise, each loop iteration could decrease the variant by a positive amount with positive probability --bringing it ever closer to zero-- but never actually reaching the zero that implies negation of the guard, and thus termination.
	\item Quasi-variants \emph{are} required to be super-martingales: that from every state satisfying $G{\land}I$ the expected value of the quasi-variant after $\mathit{Com}$ cannot increase.\par Note that \Thm{thm:lem-2-7-1} did not have a super-martingale assumption: although the probability that $\textit{VInt}$ decreased by at least 1 was required there to be at least $\varepsilon$, the change in expected value of $\textit{VInt}$ was unconstrained. For example, if with the remaining probability $1{-}\varepsilon$ it increased by a lot (but still not above $\textit{High}$), then its expected value could actually increase as well.
\end{enumerate}

A simple example of the power of \Thm{t1651} (Theme A in \Sec{s0906}) is in fact the symmetric random walk mentioned earlier. Let the state-space be the integers $x$, and let each loop iteration when $x{\neq}0$ either decrease $x$ by 1 or increase it by 1 with equal probability. \AST\ is out of reach of the earlier rule \Thm{thm:lem-2-7-1} because $x$ is not bounded above, and out of reach of some others' rules too, because the expected time to termination is infinite \citep{Fioriti:2015}. Yet termination at $x{=}0$ is shown immediately with \Thm{t1651} by taking $V{=}|x|$, trivially an exact martingale when $x{\neq}0$, and $p{=}\NF{1}{2}$ and $d{=}1$.

\subsection{Rigorous Proof of \Thm{t1651}}\label{s1656}
We begin with an informal description of the strategy of the proof that follows.
\begin{enumerate}[A.]
	\item\label{i1656-1} We choose an arbitrary real value $H{>}0$ and temporarily strengthen the loop's guard by conjoining $V{\leq}H$. From the antitone properties of $p,d$ we know that each execution of $\textit{Com}$ with that strengthened guard decreases quasi-variant $V$ by at least $d(H)$ with probability at least $p(H)$. Using that to ``discretise'' $V$, making it an integer bounded above and below, we can appeal to the earlier \Thm{thm:lem-2-7-1} to show that this guard-strengthened loop terminates \AS\ for any $H$.
	\item\label{i1656-2} Using the super-martingale property of $V$, we argue that the probability of ``bad'' escape to $V{>}H$ decreases to zero as $H$ increases: for escape from the strengthened loop to $V{>}H$ with some probability $e$ say implies a contribution of at least $e\cdot H$ to $V$'s expected value at that point. But that expected value cannot exceed $V$'s original value, because $V$ is a super-martingale. (For this we appeal to \Thm{thm:lem-2-4-1} after converting $V$ into a sub-martingale as required there.) Thus as $H$ gets larger $e$ must get smaller.
	\item\label{i1656-3} Since $e$ approaches 0 as $H$ increases indefinitely, we argue finally that, wherever we start, we can make the probability of escape to $V{>}H$ as small as we like by increasing $H$ sufficiently; complementarily we are making the only remaining escape probability, i.e.\ of ``good'' escape to $\neg G$, as close to 1 as we like. Thus it equals 1, since $H$ was arbitrary. Because this last argument depends essentially on increasing $H$ without bound, it means that $p,d$ must be defined, non-zero and antitone on \emph{all} positive reals, not only on those resulting from $V(\sigma$) on some state $\sigma$ the program happens to reach. This is particularly important when $V$ is bounded. (See \Sec{s1643}.)
\end{enumerate}

We now give the rigorous proof of \Thm{t1651}, following the strategy explained just above.
\begin{proof} (of \Thm{t1651})\\
	Let $V$ be a quasi-variant for\quad$\WHILEDO{G}{\mathit{Com}}$~, satisfying $p,d$ progress for some $p,d$ as defined in the statement of the theorem,
and recall that $I$ is a standard invariant for that loop.

\paragraph{\ref{i1656-1}} \hrulefill~\emph{For any $H$, the loop \Eqn{e1412} below terminates \AS\ from any initial state satisfying $I$.} \\
Fix arbitrary $H$ in $\Rspos$, and strengthen the loop guard $G$ of \quad$\WHILEDO{G}{\mathit{Com}}$\quad with the conjunct $V{\leq}H$. We show that
\begin{align}
	\iverson{I}\Wide{\leq}\wp{\WHILEDO{G\land V{\leq}H}{\mathit{Com}}}\one~, \label{e1412}
\end{align}
i.e.\ that standard invariant $I$ describes a set of states from which the loop \Eqn{e1412} terminates \AS.

We apply Thm.~\ref{thm:lem-2-7-1} to \Eqn{e1412}, after using ceiling $\lceil-\rceil$ to make an integer-valued variant $\mathit{VInt}$, and with other instantiations as follows:
\begin{Equation}\label{e1003}\begin{array}{c}
	\mathit{Inv}\Defs I \qquad
	\mathit{Guard}\Defs G\land V{\leq}H \\[1ex]
	\mathit{VInt}\Defs \left\lceil \frac{V}{d(H)} \right\rceil \qquad
	\mathit{Low}\Defs 0 \qquad 
	\mathit{High}\Defs \left\lceil \frac{H}{d(H)} \right\rceil \qquad 
	\varepsilon\Defs p(H) 
\end{array}\end{Equation}
The $\mathit{VInt}$ can be thought of as a \emph{discretised} version of $V$ ---
the original $V$ moves between $0$ and $H$ with down-steps of at least $d(H)$
while integer $\mathit{VInt}$ moves between $0$ and $\mathit{High}$ with down-steps of at least $1$.
In both cases, the down-steps occur with probability at least $p(H)$.

We now verify that our choices \Eqn{e1003} satisfy the assumptions of Thm.~\ref{thm:lem-2-7-1}:
\begin{enumerate}
	\item $\mathit{Inv}$ is a standard invariant of \Eqn{e1412} because  $I$ is by assumption a standard invariant of the loop\quad$\WHILEDO{G}{\mathit{Com}}$\,, and the only difference is that \Eqn{e1412} has a stronger guard.
	\item  Now note that $V{\leq}H$ implies $\left\lceil\sfrac{V}{a} \right\rceil \leq \left\lceil\sfrac{H}{a} \right\rceil$ for any strictly positive $a$. Then
		\begin{align*}
			& \mathit{Guard}\land\mathit{Inv} \\
			\Longleftrightarrow\quad & (G\land V{\leq}H)\land I \tag*{instantiations $\mathit{Guard},\mathit{Inv}$} \\
			\implies\quad & 0{<}V{\leq}H \tag*{$G\land I\Implies 0{<}V$ assumed at \Thm{t1651} (\ref{i1651-1})} \\
			\implies\quad & 0<\left\lceil \NF{V}{d(H)} \right\rceil\leq\left\lceil \NF{H}{d(H)} \right\rceil 
				\tag*{remark above and $d(H){>}0$} \\
			\implies\quad & \mathit{Low}<\mathit{VInt}\leq\mathit{High}~.
				\tag*{instantiations $\mathit{Low},\mathit{VInt},\mathit{High}$}
		\end{align*}
	\item In this final section of Step \Itm{i1656-1} we will write in an explicit style that relies less on Hoare-logic conventions and more on exposing clearly the types involved and the role of the initial- and final state. In this style, our assumption for appealing to \Thm{thm:lem-2-7-1} is that for all (initial) states $\sigma$ we have
\begin{align}
	& p(H)\cdot\iverson{G(\sigma)\land V(\sigma){\leq}H\land I(\sigma)} \label{e0945L} \\
	\leq\quad & \wp{\mathit{Com}}{(\lambda\sigma'.\iverson{\mathit{VInt}(\sigma')<\mathit{VInt}(\sigma)})}(\sigma)~. \label{e0945R}
\end{align}
Here both the \LHS and \RHS\ are real-valued expressions in which an arbitrary initial state $\sigma$ appears free. On the left $G,I$ are predicates on $\Sigma$, and $V$ is a non-negative real-valued function on $\Sigma$, and $p,H$ are constants of type $\Rspos{\to}\Rspos$ and $\Rspos$ respectively.

On the right\quad $\wp{\mathit{Com}}{(-)}$\quad is a (weakest pre-) expectation, a real-valued function on $\Sigma$; applying it to the initial state --the final $(\sigma)$ in \Eqn{e0945R} at \RHS-- produces a non-negative real scalar.

The second argument $(-)$ of\quad $\wp{\mathit{Com}}{(-)}$\quad is a post-expectation, again a function of type $\Sigma{\to}\Rpos$, but\quad $\wpCom{\mathit{Com}}$\quad takes that $(-)$'s expected value over the \emph{final} distribution(s) that $\mathit{Com}$ reaches from $\sigma$ --- for mnemonic advantage, we bind its states with $\sigma'$. And using $\sigma'$ also allows us to refer in $(-)$ to the initial state as $\sigma$, not captured by $(\lambda\sigma'.\cdots)$, so that we can compare the initial  $\mathit{VInt}(\sigma)$ and final $\mathit{VInt}(\sigma')$ values of $\mathit{VInt}$ as required.

What we have now is our assumption of progress for the original loop\quad$\WHILEDO{G}{\mathit{Com}}$\,, which was

\begin{Equation}\label{e0949}\begin{array}{rl}
	& p(V(\sigma))\cdot\iverson{G(\sigma)\land I(\sigma)} \\
	\leq\quad & \wp{\mathit{Com}}{(\lambda\sigma'.\iverson{V(\sigma')\leq V(\sigma){-}d\big(V(\sigma)\big)})}(\sigma)~,
\end{array}\end{Equation}%
and we must use \Eqn{e0949}, together with the antitone properties of $p,d$ to show $\Eqn{e0945L}{\leq}\Eqn{e0945R}$. We begin with \Eqn{e0945L} and reason

\begin{align*}
	& p(H)\cdot\iverson{G(\sigma)\land V(\sigma){\leq}H\land I(\sigma)} \tag*{\Eqn{e0945L} above}\\
	=\quad & p(H)\cdot\iverson{G(\sigma)\land 0{<}V(\sigma){\leq}H\land I(\sigma)} \tag*{$G\land I\Implies V{>}0$ by assumption \Thm{t1651}\Itm{i1651-1}} \\
	\leq\quad & p(V(\sigma))\cdot\iverson{G(\sigma)\land 0{<}V(\sigma){\leq}H\land I(\sigma)} \tag*{$V(\sigma){\leq}H$; $p$ antitone and defined on $V(\sigma)$ \footnotemark} \\
	\leq\quad & p(V(\sigma))\cdot\iverson{G(\sigma)\land I(\sigma)} \tag*{drop conjunct: $\iverson{A\land B\land C}\leq\iverson{A\land C}$} \\
	\leq\quad &  \wp{\mathit{Com}}{(\lambda\sigma'.\iverson{V(\sigma')\leq V(\sigma){-}d\big(V(\sigma)\big)})}(\sigma)\textrm{~.} \tag*{assumption \Eqn{e0949} above} \\
\end{align*}
\footnotetext{\label{n0752}Here potentially the value of $p(0)$ is used on the left, when $V(\sigma)$ is zero; but because $\iverson{\cdots0{<}V(\sigma)\cdots} = 0$ in that case, it makes no different what $p(0)$'s value is. The antitone property applies only for positive arguments.}
Now continuing only within the $\iverson{-}$ of the post-expectation we have
\footnote{\label{n1010}This reduces clutter, and in general $A{\Implies}B$ implies $\iverson{A}{\leq}\iverson{B}$, and $\wp{\mathit{Com}}{(-)}$ is itself monotonic for any $\mathit{Com}$.}

\begin{align*}
	& V(\sigma')\leq V(\sigma){-}d\big(V(\sigma)\big) \\
	\implies\quad & \big\lceil V(\sigma')/d(H)\big\rceil\leq \big\lceil V(\sigma)/d(H)-d\big(V(\sigma)\big)/d(H)\big\rceil \tag*{$d(H){>}0$, $\lceil-\rceil$ monotonic}\\
	\implies\quad & \big\lceil V(\sigma')/d(H)\big\rceil\leq\big\lceil V(\sigma)/d(H)\big\rceil-1 \tag*{$V(\sigma) \leq H$, $d$ antitone, \LHS \Eqn{e0945L}} \\
	\implies\quad & \big\lceil V(\sigma')/d(H)\big\rceil<\big\lceil V(\sigma)/d(H)\big\rceil \\
	\implies\quad & \textit{VInt}(\sigma')<\textit{VInt}(\sigma)~. \tag*{definition $\mathit{VInt}$}
\end{align*}%
Placing the last line back within\quad$\wp{\mathit{Com}}{(\lambda\sigma'.\iverson{-}})(\sigma)$\quad gives what was required at \Eqn{e0945R} and establishes \Eqn{e1412} --- that escape from $0{<}V{\leq}H$ occurs \AS\ from any initial state satisfying $I$.
\end{enumerate}

\paragraph{\ref{i1656-2}} \hrulefill~\emph{Loop \Eqn{e1412}'s probability of termination at $\neg G$ tends to 1 as $H{\rightarrow}\infty$.} \\
For the probabilistic invariant, i.e.\ sub-martingale $\mathit{Sub}$ in \autoref{thm:lem-2-4-1}, we choose $H{\ominus}V$. Note that, as required by \Thm{thm:lem-2-4-1}, expectation $\mathit{Sub}$ is bounded (by $H$\,). Let predicate $\mathit{Term}$ be $I$ which from \Eqn{e1412} we know ensures \AST\ of the modified loop. Thus the assumptions of \Thm{thm:lem-2-4-1} are satisfied: reasoning from its conclusion we have

\begin{align*}
	& \iverson{I}\cdot H{\ominus}V ~\leq~ \wp{\WHILEDO{G\land V{\leq}H}{\mathit{Com}}}{(\iverson{\neg (G\land V{\leq}H)}\cdot H{\ominus}V)} \\
	\Longleftrightarrow\quad & \iverson{I}\cdot H{\ominus}V ~\leq~ \wp{\WHILEDO{G\land V{\leq}H}{\mathit{Com}}}{(\iverson{\neg G}\cdot H{\ominus}V)} \tag*{$V{>}H\Implies H{\ominus}V{=}0$} \\
	\Longleftrightarrow\quad & \iverson{I}\cdot 1{\ominus}\NF{V}{H} ~\leq~ \wp{\WHILEDO{G\land V{\leq}H}{\mathit{Com}}}{(\iverson{\neg G}\cdot 1{\ominus}\NF{V}{H})} \tag*{scaling \Eqn{e1752} by $\NF{1}{H}$} \\
	\implies\quad & 1{\ominus}\NF{V}{H}\cdot\iverson{I}  ~\leq~ \wp{\WHILEDO{G\land V{\leq}H}{\mathit{Com}}}{\iverson{\neg G}}~, \tag*{monotonicity}
\end{align*}%
that is, recalling \Eqn{e0804},  that from any initial state satisfying $I$ the loop \Eqn{e1412} terminates in a state satisfying $\neg G$ with probability at least $1{\ominus}\NF{V}{H}$. As required, that probability (for fixed initial state) tends to 1 as $H$ tends to infinity.

\paragraph{\ref{i1656-3}} \hrulefill~\emph{The original loop terminates \AS\ from any initial state satisfying $I$.} \\
From \App{a1139}, instantiating $A\Defs G$ and $B\Defs V{\leq}H$,
we have for any $H$ that
\[
	\wp{\WHILEDO{G\land V{\leq}H}{\mathit{Com}}}{\iverson{\neg G}}
	\Wide{\leq}\wp{\WHILEDO{G}{\mathit{Com}}}{\iverson{\neg G}}
\] 
and, referring to the last line in \Itm{i1656-2} just above, we conclude $(1{\ominus}\NF{V}{H})\cdot\iverson{I} \leq\wp{\WHILEDO{G}{\mathit{Com}}}{\iverson{\neg G}}$\,. Since that holds for any $H$ no matter how large, we have finally that
\[
  \iverson{I} \Wide{\leq} \wp{\WHILEDO{G}{\mathit{Com}}}{\iverson{\neg G}} \Wide{\leq} \wp{\WHILEDO{G}{\mathit{Com}}}{\one}~,
\]
that is that from any initial state satisfying $I$ the loop\quad$\WHILEDO{G}{\mathit{Com}}$\quad terminates \AS.
\end{proof}

\section{Case Studies}\label{s0918}
In this section, we examine a few (mostly) non-trivial examples to show the effectiveness of \Thm{t1651}.
For all examples we provide a $p,d$ quasi-variant $V$ that proves \AST;
and we will always choose $p,d$ so that they are strictly positive and antitone.
We will not provide proofs of the $p,d$ properties, because they will be self-evident and are in any case ``external'' mathematical facts.
We do however carefully set-out any proofs that depend on the program text: that $V{=}0$ indicates termination, that $V$ satisfies the super-martingale property, and that $p$, $d$, and $V$ satisfy the progress condition.

For convenience in these examples, we define a derived expectation transformer $\awpsymbol$, over terminating straight-line programs only (as our loop bodies are, in this section), that ``factors out'' the $(H{\ominus})$; it has the same definition as of $\wpsymbol$ in \autoref{table:wp} except that nondeterminism is interpreted angelically rather than demonically: that is, we define
\begin{align*}
	\awp{\NDCHOICE{C_1}{C_2}}{f}	\Wide{=} \max\left\{ \awp{C_1}{f},\, \awp{C_2}{f} \right\} ~,
\end{align*}
and otherwise as for $\wpsymbol$ (except for loops, which we do not need here).
A straightforward structural induction then shows that for straight-line programs $\mathit{Com}$, constant $H$ and any expectation $V$  that
\begin{Align}\label{e1134}
	H\ominus \awp{\mathit{Com}}{V} \Wide{\leq} \wp{\mathit{Com}}{(H{\ominus}V)} ~.
\end{Align}
And from there we have immediately that
\begin{Align}
	V\geq \awp{\mathit{Com}}{V} \Wide{\implies} H{\ominus}V\leq \wp{\mathit{Com}}{(H{\ominus}V)} ~, \label{e0908}
\end{Align}
and finally therefore that
\begin{Align}
	V\geq \iverson{G\land I}\cdot \awp{\mathit{Com}}{V} \Wide{\implies} \iverson{G\land I}\cdot(H{\ominus}V)\leq \wp{\mathit{Com}}{(H{\ominus}V)} ~,\label{e0909}
\end{Align}
since if $G\land I$ holds then \Eqn{e0909} reduces to \Eqn{e0908} and, if it does not hold, both sides of \Eqn{e0909} are trivially true. Thus when the loop body is a straight-line program, by establishing \LHS \Eqn{e0909} we establish also \RHS\ \Eqn{e0909} as required by \Thm{t1651}\Itm{i1651-4}. We stress that $\awpsymbol$ is used here for concision and intuition only: applied only to finite, non-looping programs, it can always be replaced by $\wpsymbol$.

Thus \LHS \Eqn{e0909} expresses clearly and directly that $V$ is a super-martingale when $G\land I$ holds, and handles any nondeterminism correctly in that respect: because $\awpsymbol$ \emph{maximises} rather than minimises over nondeterministic outcomes (the opposite of $\wpsymbol$), the super-martingale inequality $(\geq)$ holds for every individual outcome, as required.

In \Sec{s1127} we discuss the reasons for not using $\awpsymbol$ in \Thm{t1651} directly, i.e.\ not eliminating ``$H\ominus$'' at the very start: in short, it is because our principal reference \citep{McIver:05a} does not support $\awpsymbol$.

\subsection{The Negative-Binomial Loop}\label{ss1118}

Our first example is also proved by other \AST\ rules, so we do not need the extra power of \Thm{t1651} for it; but we begin with this to illustrate Theme B with a familar example how \Thm{t1651} is used in formal reasoning over program texts.

\paragraph{Description of the loop.}
Consider the following while-loop over the real-valued variable $x$:
\begin{equation}\label{e1143}
	\WHILEDO{x{\neq}0}{~\ASSIGN{x}{x{-}1} ~\PCF~ \SKIP~} ~.
\end{equation}
An interpretation of this loop as a transition system is illustrated in \autoref{fig:geoloop}.
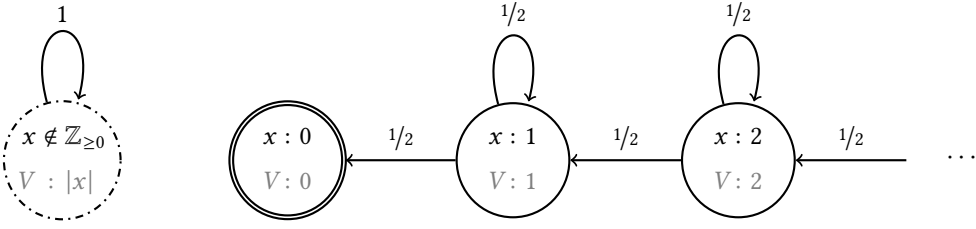
\begin{figure}[t]
	\begin{center}
		\let\oldarraycolsep\arraycolsep
		\arraycolsep=0pt
		\begin{tikzpicture}[every state/.append style={thick, inner sep=0pt, minimum size=1.5cm}]
			\draw[help lines, use as bounding box,white] (-3.75,-0.75) grid (9.75, 2);
			
			\node[state, dashdotted] (notZ) at (-3, 0) {$%
				\begin{array}{rcl}%
					x & {}\not\in{} & \mathbb{Z}_{\geq 0} \\[.5em]
					\textcolor{gray}{V} & \textcolor{gray}{:} & \textcolor{gray}{|x|}
				\end{array}%
			$};
			
			\node[state, accepting] (c0) at (0, 0) {$%
				\begin{array}{ll}%
					x & \colon 0 \\[.5em]
					\textcolor{gray}{V} & \textcolor{gray}{\colon 0}
				\end{array}%
			$};
			
			\node[state] (c1) at (3, 0) {$%
				\begin{array}{ll}%
					x & \colon 1 \\[.5em]
					\textcolor{gray}{V} & \textcolor{gray}{\colon 1}
				\end{array}%
			$};
			
			\node[state] (c2) at (6, 0) {$%
				\begin{array}{ll}%
					x & \colon 2 \\[.5em]
					\textcolor{gray}{V} & \textcolor{gray}{\colon 2}
				\end{array}%
			$};
			
			\node[state, draw=none] (dots) at (9, 0) {$\cdots$};

			\path[->,thick] (notZ) edge[loop above] node[above] {$1$} (c1);
			
			\path[->,thick] (c1) edge[loop above] node[above] {$\NF{1}{2}$} (c1);
			\path[->,thick] (c1) edge[left] node[above] {$\NF{1}{2}$} (c0);
			
			\path[->,thick] (c2) edge[loop above] node[above] {$\NF{1}{2}$} (c2);
			\path[->,thick] (c2) edge[left] node[above] {$\NF{1}{2}$} (c1);
			
			\path[->,thick] (dots) edge[left] node[above] {$\NF{1}{2}$} (c2);
		\end{tikzpicture}
		\arraycolsep=\oldarraycolsep
	\end{center}
	\caption{Execution of the negative binomial loop. 
	The solid nodes represent program states and moreover the doubly-circled node represents a state in which the loop has terminated.
	The leftmost dash-dotted node represents the \emph{collection} of all states in which the value of $x$ is not a non-negative integer (from where the random walk will indeed not terminate).
	Inside the nodes we give the variable valuations as well as the values of the variant $V = |x|$ in each state. 
	The values of $p$ and $d$ are constantly $\NF{1}{2}$ and $1$, respectively.}
	\label{fig:geoloop}
\end{figure}
%
Intuitively, this loop keeps flipping a coin until it flips, say, heads $x$ times (not necessarily in a row); every time it flips tails, the loop continues without changing the program state.

We call it the negative binomial loop because its runtime is distributed according to a negative binomial distribution (with parameters $x$ and $\NF {1}{2}$), and thus the expected runtime is linear (on average $2x$ loop iterations) even though it allows for infinite executions, namely those runs of the program that flip heads fewer than $x$ times and then keep flipping tails ad infinitum.

A subtle intricacy is that this loop will not terminate at all, if $x$ is initially not a \emph{non-negative integer}, because then the execution of the loop never reaches a state in which $x{=}0$.
This is where we use Theorem~\ref{t1651}'s ability of incorporating an invariant into the \AST~proof, as standard arguments over loop termination do.

\paragraph{Proof of almost-sure termination}
The guard is given by \hfill $G \eeq x{\neq}0$~,\\
and the loop body by \hfill $\mathit{Com} \eeq \PCHOICE{\ASSIGN{x}{x - 1}}{\NF{1}{2}}{\SKIP}$~.\\
And with the standard invariant \hfill $I \eeq x {\in} \mathbb{Z}_{\geq 0}$~,\\
we can now prove \AST\ of the loop with
an appropriate $p,d$ and quasi-variant $V$:
\begin{align*}
	V \eeq |x|, \qquad\textnormal{for } d \eeq 1 \quad\textnormal{and}\quad p \eeq1/2 ~.
\end{align*}
Notice that $d,p$ are strictly speaking constant functions mapping any positive real $v$ to $1,\NF{1}{2}$ respectively.
Intuitively, this choice of $I$, $V$, $p$, and $d$ tells us that if $x$ is a positive integer different from $0$, then after one iteration of the loop body (a) $x$ is still a non-negative integer (by invariance of $I$) and (b) the distance of $x$ from $0$ has decreased by at least $1$ with probability at least~$\NF{1}{2}$~(implied by the progress condition).

We first check that $I = x {\in} \mathbb{Z}_{\geq 0}$ is indeed an invariant:
\begin{align*}
	\iverson{G} \cdot \iverson{I} 
	\eeq \iverson{x \neq 0} \cdot \iverson{x \in \mathbb{Z}_{\geq 0}}
	\eeq &\iverson{x \in \mathbb{Z}_{> 0}} \\
	\lleq &\frac{1}{2} \big(\iverson{x \in \mathbb{Z}_{> 0}} + \iverson{x \in \mathbb{Z}_{\geq 0}} \big) \\
	\eeq &\frac{1}{2} \big(\iverson{x{-}1 \in \mathbb{Z}_{\geq 0}} + \iverson{x \in \mathbb{Z}_{\geq 0}} \big) \\
	\eeq &\wp{\PCHOICE{\ASSIGN{x}{x - 1}}{\NF{1}{2}}{\SKIP}}{\iverson{x \in \mathbb{Z}_{\geq 0}}} \\
	\eeq &\wp{\mathit{Com}}{\iverson{I}}~.
\end{align*}
Next, the second precondition of Theorem~\ref{t1651} is satisfied because of
\begin{align*}
	G \land I ~\iff~ x{\neq}0 \land x {\in} \mathbb{Z}_{\geq 0} ~\implies~ x{\neq}0 ~\implies~ |x|{>}0 ~\iff~ V{>}0~.
\end{align*}
Furthermore, $V$ satisfies the super-martingale property:
\begin{align*}
		\iverson{G\land I} \cdot \awp{\mathit{Com}}{V}  &\eeq \iverson{x{\neq}0 \land x {\in} \mathbb{Z}_{\geq 0}} \cdot \awp{\left( \PCHOICE{\ASSIGN{x}{x - 1}}{\NF{1}{2}}{\SKIP} \right)}{|x|} \\
		& \eeq \iverson{x \in \mathbb{Z}_{> 0}} \cdot \frac{1}{2} \cdot \big( |x - 1| + |x| \big) \\
		& \eeq \iverson{x \in \mathbb{Z}_{> 0}} \cdot  \left(|x| - \frac{1}{2} \right) \\
		& \lleq \iverson{x \in \mathbb{Z}_{> 0}} \cdot  |x| \\
		& \lleq |x| \\
		& \eeq V ~.
\end{align*}
Lastly, $V$, $p$, and $d$ satisfy the progress condition for all $R$:%
\begin{align*}
		&p(R) \cdot \iverson{G \land I \land V{=}R} \lleq \wp{\mathit{Com}}{\iverson{V \leq R - d(R)}} \\
		\Longleftrightarrow \quad &\frac{1}{2} \cdot \iverson{x {\neq} 0 \land x {\in} \mathbb{Z}_{\geq 0} \land |x|{=}R} \lleq \wp{\PCHOICE{\ASSIGN{x}{x - 1}}{\NF{1}{2}}{\SKIP}}{\iverson{|x| \leq R{-}1}} \\
		\Longleftrightarrow \quad &\frac{1}{2} \cdot \iverson{x {\in} \mathbb{Z}_{>0} \land |x|{=}R} \lleq \wp{\PCHOICE{\ASSIGN{x}{x - 1}}{\NF{1}{2}}{\SKIP}}{\iverson{|x| \leq R{-}1}} \\
		\Longleftrightarrow \quad &\frac{1}{2} \cdot \iverson{x {\in} \mathbb{Z}_{>0} \land |x|{=}R} \lleq \frac{1}{2} \cdot \big( \iverson{|x{-}1| \leq R{-}1} + \iverson{|x| \leq R{-}1} \big) \\
		\Longleftrightarrow \quad &\iverson{x {\in} \mathbb{Z}_{>0} \land |x|{=}R} \lleq \big( \iverson{|x{-}1| \leq R{-}1} + \iverson{|x| \leq R{-}1} \big) \\
		\Longleftrightarrow \quad & \iverson{x {\in} \mathbb{Z}_{>0} \land |x|{=}R} \lleq \iverson{x {\in} \mathbb{Z}_{>0} \land |x|{=}R} \cdot \big( \iverson{|x{-}1| \leq R{-}1} + \iverson{|x| \leq R{-}1} \big) \\
		\Longleftrightarrow \quad & \iverson{x {\in} \mathbb{Z}_{>0} \land |x|{=}R} \lleq \iverson{x {\in} \mathbb{Z}_{>0} \land |x|{=}R} \cdot ( 1 + 0 ) \\
		\Longleftrightarrow \quad & \iverson{x {\in} \mathbb{Z}_{>0} \land |x|{=}R} \lleq \iverson{x {\in} \mathbb{Z}_{>0} \land |x|{=}R} \\
		\Longleftrightarrow \quad & \true~.
\end{align*}
This shows that all preconditions of Theorem~\ref{t1651} are satisfied: thus we have ~$\iverson{x {\in} \mathbb{Z}_{\geq 0}}\leq \wp{\textrm{\Eqn{e1143}}}{\one}$~,
i.e.\ that the negative binomial loop terminates almost-surely from all initial states in which $x$ is a non-negative integer.

\subsection{The Demonically Fair Random Walk}\label{ss1311}

Next, we consider a while-loop that contains both probabilistic- and demonic choice.

\paragraph{Description of the loop.}
Consider the following while-loop:
\begin{align*}
	& \WHILE{x > 0} \\
	& \qquad \PCHOICE{\vphantom{\big(}\ASSIGN{x}{x - 1}}{\NF{1}{2}}{\NDCHOICE{\ASSIGN{x}{x + 1}}{\SKIP}\vphantom{\big(}} \\
	& \}
\end{align*}
In order not to clutter the reasoning below, we assume without loss of generality that for this example $x$ is of type $\Nats$.
The execution of the loop is illustrated in \autoref{fig:fairloop}. 
\begin{figure}[t]
	\begin{center}
		\let\oldarraycolsep\arraycolsep
		\arraycolsep=0pt
		\begin{tikzpicture}[every state/.append style={thick, inner sep=0pt, minimum size=1cm}]
			\draw[help lines, use as bounding box,white] (-0.5,-0.5) grid (6.5, 2);
			\node[state, accepting] (c0) at (0, 0) {$%
				\begin{array}{rc}%
					x\colon &  0
				\end{array}%
			$};
			
			\node[state] (c1) at (2, 0) {$%
				\begin{array}{rc}%
					x\colon &  1
				\end{array}%
			$};
			
			\node[state] (c2) at (4, 0) {$%
				\begin{array}{rc}%
					x\colon &  2
				\end{array}%
			$};
			
			\node[] (dots) at (6, 0) {$\cdots$};
			
			
			\node[inner sep=2pt] (nd1) at (2, 2) {$\Box$}; 
			
			\node[inner sep=2pt] (nd2) at (4, 2) {$\Box$}; 
			
			\node[inner sep=2pt] (nddots) at (6, 2) {\phantom{$\Box$}}; 

			
			\path[->,thick] (c1) edge[bend left] node[left,pos=0.7] {$\NF{1}{2}$} (nd1);
			\path[->,thick,dashed] (nd1) edge[left] node[above] {} (c2);
			\path[->,thick,dashed] (nd1) edge[bend left] node[left] {} (c1);

			\path[->,thick] (c2) edge[bend left] node[left,pos=0.7] {$\NF{1}{2}$} (nd2);
			\path[->,thick,dashed] (nd2) edge[left] node[above] {} (dots);
			\path[->,thick,dashed] (nd2) edge[bend left] node[left] {} (c2);

			
			\path[->,thick] (c1) edge[left] node[above] {$\NF{1}{2}$} (c0);
			\path[->,thick] (c2) edge[left] node[above] {$\NF{1}{2}$} (c1);
			\path[->,thick] (dots) edge[left] node[above] {$\NF{1}{2}$} (c2);
		\end{tikzpicture}
		\arraycolsep=\oldarraycolsep
	\end{center}
	\caption{Execution of the demonically fair random walk. 
	The $\Box$ nodes together with the dashed arrows represent demonic choices. 
	The value of the variant is equal to the value of $x$ in each state. 
	The values of $p$ and $d$ are constantly $\NF{1}{2}$ and $1$, respectively.}
\label{fig:fairloop}
\end{figure}
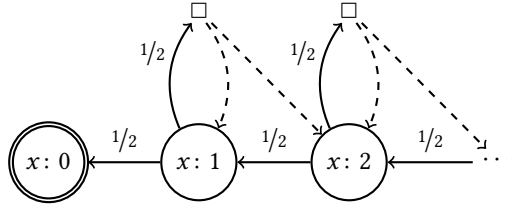

The motivation for this loop is the recursive procedure $P$ inspired by an example of~\citet{Olmedo:2016aa}; its definition is
\begin{align*}
	P~ \rhd \quad \PCHOICE{\vphantom{\big(}\SKIP}{\NF{1}{2}}{\,\mathtt{call}~P;~\NDCHOICE{\mathtt{call}~P}{\SKIP}\vphantom{\big(}}~,
\end{align*}
and we have rewritten it as a loop by viewing it as a random walk of a particle $x$ whose position represents the height of the call stack.
Intuitively, the loop keeps moving $x$ in a random and demonic fashion until the particle hits the origin $0$ (empty call stack, all procedure calls have terminated).
For that at each stage it either with probability $\NF{1}{2}$ decrements the position of $x$ by one (procedure call terminates after $\SKIP$; call stack decremented by one), or with probability $\NF{1}{2}$ it performs a demonic choice between incrementing the position of $x$ by one (perform two consecutive procedure calls, then terminate; call stack in effect incremented by one (${}+2 - 1 = {}+1$)) or letting $x$ remain at its position (perform one procedure call, then terminate; call stack in effect unchanged (${}+1 - 1 = 0$)).

\paragraph{Proof of almost-sure termination.}

The loop guard is given by\quad$G \eeq x {>} 0$\quad
and the loop body by
\begin{align*}
	\mathit{Com} \eeq \PCHOICE{\vphantom{\big(}\ASSIGN{x}{x - 1}}{\NF{1}{2}}{\NDCHOICE{\ASSIGN{x}{x + 1}}{\SKIP}\vphantom{\big(}}~.
\end{align*}

We now prove \AST\ of the loop by choosing the standard invariant $I = \true$ \footnote{Predicate $\true$ is an invariant for any loop whose body is terminating, e.g.\ is itself loop-free.} and an appropriate $p,d$ and quasi-variant $V$:
\begin{align*}
	V \eeq x, \qquad\textnormal{for } d \eeq 1 \quad\textnormal{and}\quad p \eeq 1/2~.
\end{align*}
Intuitively this choice of $V$, $p$, and $d$ tells us that the value of $x$ decreases with probability at least $\NF{1}{2}$ by at least $1$  through an iteration of the loop body in the case that initially $x{>}0$.

The second precondition of Theorem~\ref{t1651} is satisfied because\quad$G \land I ~\Longleftrightarrow~ x{>}0 ~\Longleftrightarrow~ V{>}0$~. Furthermore, $V$ satisfies the super-martingale property:
\begin{align*}
		\iverson{G \land I} \cdot \awp{\mathit{Com}}{V}  &\eeq \iverson{x{>}0} \cdot \awp{\PCHOICE{\ASSIGN{x}{x - 1}}{\NF{1}{2}}{\NDCHOICE{\ASSIGN{x}{x + 1}}{\SKIP}}}{x} \\
		& \eeq \iverson{x{>}0}\cdot \frac{1}{2} \cdot \left( x - 1 + \Max{x + 1}{x} \right) \\
		& \eeq \iverson{x{>}0}\cdot \frac{1}{2} \cdot \left( x - 1 + x + 1\right) \\
		& \eeq \iverson{x{>}0}\cdot x \\
		& \lleq x \\
		& \eeq V~.
\end{align*}
Lastly, $V$, $p$, and $d$ satisfy the progress condition for all $R$:
\begin{align*}
		&p(R) \cdot \iverson{G \land I \land V{=}R} \lleq \wp{\mathit{Com}}{\iverson{V \leq R - d(R)}} \\
			\Longleftrightarrow \quad &\frac{1}{2} \cdot \iverson{x{>}0 \land \True \land x{=}R} \lleq \wp{\PCHOICE{\vphantom{\big(}\ASSIGN{x}{x - 1}}{\NF{1}{2}}{\NDCHOICE{\ASSIGN{x}{x + 1}}{\SKIP}\vphantom{\big(}}}{\iverson{x \leq R{-}1}} \\
		\Longleftrightarrow \quad &\frac{1}{2} \cdot \iverson{x{>}0\land x{=}R} \lleq \frac{1}{2} \cdot \big( \iverson{x{-}1 \leq R{-}1} + \Max{\iverson{x{+} 1\leq R{-}1}}{\iverson{x\leq R{-}1}}\big) \\
		\Longleftrightarrow \quad &\iverson{x{>}0\land x{=}R} \lleq \iverson{x \leq R} + \iverson{x\leq R{-}1}\\
		\Longleftrightarrow \quad &\iverson{x{>}0\land x{=}R} \lleq \iverson{x \leq x} + \iverson{x\leq x{-}1}\\
		\Longleftrightarrow \quad &\iverson{x{>}0\land x{=}R} \lleq 1 + 0\\
		\Longleftrightarrow \quad & \true~.
\end{align*}%
This shows that all preconditions of Theorem~\ref{t1651} are satisfied and as a consequence the demonic random walk loop above terminates almost-surely.
Interestingly, the procedure $P'$ given by
\begin{align*}
	P'~ \rhd \quad \PCHOICE{\SKIP}{\NF{1}{2}}{\mathtt{call}~P';
	~\mathtt{call}~P';~\NDCHOICE{\mathtt{call}~P'}{\SKIP}}~,
\end{align*}%
i.e.\ potentially three consecutive procedure calls instead of two~\citep{Olmedo:2016aa}, is not \AST: it terminates with probability only $\NF{(\sqrt{5} - 1)}{2} < 1$.

\subsection{The Fair-in-the-Limit Random Walk}\label{s0827}

While so far we have considered constant probabilities and constant decreases, we now consider a while-loop requiring use of a non-constant decrease function $d$.

\paragraph{Description of the loop.}

Consider the following while-loop:
\begin{align*}
	& \WHILE{x > 0} \\
	& \qquad \ASSIGN{q}{\NF{x}{2x{+}1}};\\
	& \qquad \PCHOICE{\ASSIGN{x}{x - 1}}{q}{\ASSIGN{x}{x + 1}} \\
	& \}
\end{align*}
Assume again that $x{\in}\Nats$.
The execution of the loop is illustrated in \autoref{fig:fair-in-the-limit-loop}.
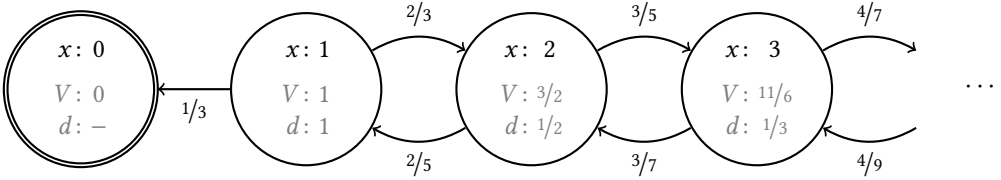
\begin{figure}[t]
	\begin{center}
		\let\oldarraycolsep\arraycolsep
		\arraycolsep=0pt
		\begin{tikzpicture}[every state/.append style={thick, inner sep=0pt, minimum size=2cm}]
			\draw[help lines, use as bounding box, white] (-1,-1.25) grid (12, 1.25);
			\node[state, accepting] (c0) at (0, 0) {$%
				\begin{array}{rc}%
					x\colon &  0\\[.5em]
					\textcolor{gray}{V\colon} & \textcolor{gray}{ 0}\\
					\textcolor{gray}{d\colon} & \textcolor{gray}{{-}}
				\end{array}%
			$};
			
			\node[state] (c1) at (3, 0) {$%
				\begin{array}{rc}%
					x\colon &  1\\[.5em]
					\textcolor{gray}{V\colon} & \textcolor{gray}{ 1}\\
					\textcolor{gray}{d\colon} & \textcolor{gray}{ {1}}
				\end{array}%
			$};
			
			\node[state] (c2) at (6, 0) {$%
				\begin{array}{rc}%
					x\colon &  2\\[.5em]
					\textcolor{gray}{V\colon} & \textcolor{gray}{ \NF{3}{2}}\\
					\textcolor{gray}{d\colon} & \textcolor{gray}{ \NF{1}{2}}
				\end{array}%
			$};
			
			\node[state] (c3) at (9, 0) {$%
				\begin{array}{rc}%
					x\colon &  3\\[.5em]
					\textcolor{gray}{V\colon} & \textcolor{gray}{ \NF{11}{6}}\\
					\textcolor{gray}{d\colon} & \textcolor{gray}{ \NF{1}{3}}
				\end{array}%
			$};
			
			\node[state, draw=white] (dots) at (12, 0) {$\cdots$};

			
			\path[->,thick] (c1) edge[left] node[below] {$\NF{1}{3}$} (c0);
			\path[->,thick] (c2) edge[bend left] node[below] {$\NF{2}{5}$} (c1);
			\path[->,thick] (c3) edge[bend left] node[below] {$\NF{3}{7}$} (c2);
			\path[->,thick] (dots) edge[bend left] node[below] {$\NF{4}{9}$} (c3);

			\path[->,thick] (c1) edge[bend left] node[above] {$\NF{2}{3}$} (c2);
			\path[->,thick] (c2) edge[bend left] node[above] {$\NF{3}{5}$} (c3);
			\path[->,thick] (c3) edge[bend left] node[above] {$\NF{4}{7}$} (dots);
		\end{tikzpicture}
		\arraycolsep=\oldarraycolsep
	\end{center}
	\caption{Execution of the fair-in-the-limit random walk. 
	Inside the nodes we give the valuations of variable $x$ as well as the values of the variant $V$ and the decrease function $d$. 
	The value of $p$ is constantly $\NF{1}{3}$. Note that in \Thm{t1928} it does not matter what $d$'s value is when $V{=}0$, because the \LHS\ of \Itm{i1651-3} is zero in that case. \Ct{Check.}}
	\label{fig:fair-in-the-limit-loop}
\end{figure}

Intuitively, the loop models an asymmetric random walk of a particle $x$,
terminating when the particle hits the origin $0$.
In one iteration of the loop body, the program either with probability $\NF{x}{2x{+}1}$ decrements the position of $x$ by one, or with probability $\NF{x{+}1}{2x{+}1}$ increments the position of $x$ by one.
The further the particle $x$ is away from 0, the more fair becomes the random walk since $\NF{x}{2x{+}1}$ approaches $\NF{1}{2}$ asymptotically.
Yet, it is not so obvious that this random walk indeed also terminates with probability 1.

\paragraph{Proof of almost-sure termination.}
The loop guard is given by\quad$G \eeq x {>} 0$\quad
and the loop body by
\begin{align*}
	\mathit{Com} \eeq \COMPOSE{\ASSIGN{q}{\NF{x}{2x{+}1}}}{\PCHOICE{\ASSIGN{x}{x - 1}}{q}{\ASSIGN{x}{x + 1}}}~.
\end{align*}
We now prove almost-sure termination of the loop by choosing standard invariant $I = \true$ and an appropriate $p,d$ quasi-variant $V$:
\begin{align*}
	V \eeq H_x, \qquad\textnormal{for } d(v) \eeq \begin{cases}
		\frac{1}{x}, & \text{if } v > 0 \text{ and } v \in (H_{x-1},\, H_x] \\
		1, & \text{if } v = 0
	\end{cases}
	\quad\textnormal{and}\quad p(v) \eeq \frac{1}{3}~,
\end{align*}
where $H_x$ is the $x$-th harmonic number.\footnote{$H_x = \sum_{n = 1}^{x} \frac{1}{n}$. Notice that $H_0 = 0$.}
Notice that the variant $V$ is non-affine here, i.e.\ not of the form $a + bx + cq$, and we will show below that no affine variant can satisfy a super-martingale property.
Intuitively our choice of $p$ and $d$ tells us that the variant $V$, i.e.\ the harmonic number of the value of $x$, decreases with probability at least $\NF{1}{3}$ by at least $\frac{1}{x}$  through an iteration of the loop body in case that initially $x > 0$.

The second precondition of Theorem~\ref{t1651} is satisfied because
\begin{align*}
	G \land I ~\Longleftrightarrow~ x{>}0 ~\Longleftrightarrow~ H_x{>}0 ~\Longleftrightarrow~ V{>}0~.
\end{align*}
Furthermore, $V$ satisfies the super-martingale property:
\begin{align*}
		\iverson{G} \cdot \awp{\mathit{Com}}{V}  &\eeq \iverson{x{>}0} \cdot \awp{\COMPOSE{\ASSIGN{q}{\NF{x}{2x{+}1}}}{\PCHOICE{\ASSIGN{x}{x - 1}}{q}{\ASSIGN{x}{x + 1}}}}{H_x} \\
		&\eeq \iverson{x{>}0} \cdot \awp{\ASSIGN{q}{\NF{x}{2x{+}1}}}{\left( q \cdot H_{x - 1} + (1{-}q) \cdot H_{x + 1} \right)} \\
		&\eeq \iverson{x{>}0} \cdot \left(\frac{x}{2x{+}1} \cdot H_{x-1} + \left(1-\frac{x}{2x{+}1}\right) \cdot H_{x + 1}\right)\\
		&\eeq \iverson{x{>}0} \cdot \left(\frac{x}{2x{+}1} \cdot \left(H_x  - \frac{1}{x} \right) + \left(\frac{x{+}1}{2x{+}1}\right) \cdot \left( H_x + \frac{1}{x{+}1}\right)\right)\\
		&\eeq \iverson{x{>}0} \cdot \left(\left(\frac{x}{2x{+}1} + \frac{x{+}1}{2x{+}1}\right) \cdot H_x  - \frac{1}{2x{+}1} + \frac{1}{2x{+}1} \right)\\
		&\eeq \iverson{x{>}0} \cdot H_x\\
		& \lleq H_x \\
		& \eeq V~.
\end{align*}
Lastly, $V$, $p$, and $d$ satisfy the progress condition for all $R$. Notice that $d(H_\C{x}) = \NF{1}{\C{x}}$ and consider the following:
\begin{align*}
		&p(R) \cdot \iverson{G \land I \land V{=}R} \lleq  \wp{\mathit{Com}}{\iverson{V \leq R - d(R)}} \\
		\Longleftrightarrow \quad & \frac{1}{3} \cdot \iverson{x{>}0 \land H_x{=}R} \lleq \wp{\COMPOSE{\ASSIGN{q}{\NF{x}{2x{+}1}}}{\PCHOICE{\ASSIGN{x}{x - 1}}{q}{\ASSIGN{x}{x + 1}}}}{\iverson{H_x \leq R - d(R)}} \\
		\Longleftrightarrow \quad &\frac{1}{3} \cdot \iverson{x{>}0 \land H_x{=}R} \lleq \wp{\ASSIGN{q}{\NF{x}{2x{+}1}}}{\left(q \cdot \iverson{H_{x-1} \leq R - d(R)}+ (1{-}q) \cdot \iverson{H_{x + 1} \leq R - d(R)}\right)} \\
		\Longleftrightarrow \quad &\frac{1}{3} \cdot \iverson{x{>}0 \land H_x{=}R} \lleq \frac{x}{2x{+}1} \cdot \iverson{H_{x - 1} \leq R - d(R)}
		+ \left(1 - \frac{x}{2x{+}1}\right) \cdot \iverson{H_{x + 1} \leq R - d(R)}\\
		\Longleftrightarrow \quad &\frac{1}{3} \cdot \iverson{x{>}0 \land H_x{=}R} \lleq \frac{x}{2x{+}1} \cdot \iverson{H_{x - 1} \leq R - d(R)}
		+ \left(\frac{x{+}1}{2x{+}1}\right) \cdot \iverson{H_{x{+}1} \leq R - d(R)}\\
		\Longleftrightarrow \quad &\frac{1}{3} \cdot \iverson{x{>}0 \land H_x{=}R} \lleq \frac{x}{2x{+}1} \cdot \iverson{H_{x - 1} \leq H_x - \frac{1}{x}}
		+ \left(\frac{x{+}1}{2x{+}1}\right) \cdot \iverson{H_{x + 1} \leq H_x - \frac{1}{x}}\\
		\Longleftrightarrow \quad &\iverson{x{>}0} \cdot \frac{1}{3} \lleq \left(\frac{x}{2x{+}1} \cdot 1 + \frac{x{+}1}{2x{+}1} \cdot 0\right) \\
		\Longleftrightarrow \quad &\iverson{x{>}0} \cdot \frac{1}{3} \lleq \frac{x}{2x{+}1} \\
		\Longleftrightarrow \quad &\true~.
\end{align*}
This shows that all preconditions of Theorem~\ref{t1651} are satisfied and as a consequence the fair-in-the-limit random walk terminates almost-surely.

\paragraph{Proof of non-existence of an affine variant.}

For this program, there exists \emph{no affine variant} that satisfies the super-martingale property as used e.g.\ by \citet{Chatterjee:2017aa}. Any affine
\footnote{Some authors call this a \emph{linear} variant.}
variant $V$ would have to be of the form
\begin{align*}
	V \eeq a + bx + cq~,
\end{align*}
for some (positive) coefficients $a$, $b$, $c$.
\footnote{Coefficients need to be positive because otherwise $V \geq 0$ cannot be ensured. However, this is not crucial in this proof.}
Now we attempt to check the super-martingale property for a variant of that form:
\begin{align*}
		&\iverson{G} \cdot \awp{\mathit{Com}}{V}  \\
		&\eeq \iverson{x{>}0} \cdot \awp{\COMPOSE{\ASSIGN{q}{\NF{x}{2x{+}1}}}{\PCHOICE{\ASSIGN{x}{x - 1}}{q}{\ASSIGN{x}{x + 1}}}}{(a + bx + cq)} \\
		&\eeq \iverson{x{>}0} \cdot \awp{\ASSIGN{q}{\NF{x}{2x{+}1}}}{\left( q \cdot (a + b(x{-}1) + cq) + (1{-}q) \cdot (a + b(x{+}1) + cq) \right)} \\
		&\eeq \iverson{x{>}0} \cdot \awp{\ASSIGN{q}{\NF{x}{2x{+}1}}}{\left( a - 2bq + bx + b + cq \right)} \\
		&\eeq \iverson{x{>}0} \cdot \left(a - 2b\cdot \frac{x}{2x + 1} + bx + b + c \cdot \frac{x}{2x + 1}\right)\\
		& ~{}\stackrel{!}{\leq}{}~ a + bx + cq\\
		& \eeq V~.
\end{align*}

If $x \leq 0$ this is trivially satisfied. 
If $x{>}0$, then the above is satisfied iff
\begin{align*}
	&a - 2b\cdot \frac{x}{2x + 1} + bx + b + c \cdot \frac{x}{2x + 1} \lleq a + bx + cq \\
	\Longleftrightarrow \quad &-2b\cdot \frac{x}{2x + 1} + b + c \cdot \frac{x}{2x + 1} \lleq cq~,
\end{align*}
which is only satisfiable for all possible valuations of $q$ and $x{>}0$ iff $b = c = 0$.
Thus if $V$ is forced to be affine, then $V$ has to be constantly $a$, for $a \geq 0$.
Indeed, $a$ is a super-martingale.
However, it is clear that a constant $V$ cannot possibly indicate termination as
\begin{align*}
	\iverson{V = 0} \eeq 1 \nneq \iverson{x \leq 0} \eeq \iverson{\neg G}~.
\end{align*}
Thus, there cannot exist an affine variant that satisfies the super-martingale property.

\subsection{The Escaping Spline}\label{s0838}
We now consider a while-loop where we we will make use of both non-constant probability function $p$ and non-constant decrease function $d$.

\paragraph{Description of the loop.}

Consider the following while-loop:
\begin{align*}
	& \WHILE{x > 0} \\
	& \qquad \ASSIGN{q}{\NF{1}{x+1}};\\
	& \qquad \PCHOICE{\ASSIGN{x}{0}}{q}{\ASSIGN{x}{x + 1}} \\
	& \}
\end{align*}
Assume again that $x{\in}\Nats$.
The execution of the loop is illustrated in \autoref{fig:escaping-spline-loop}.
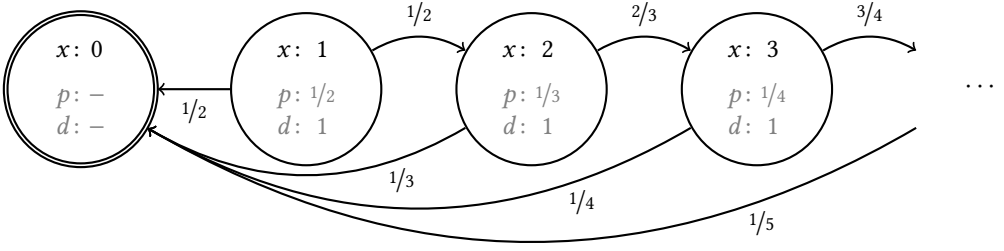
\begin{figure}[t]
	\begin{center}
		\let\oldarraycolsep\arraycolsep
		\arraycolsep=0pt
		\begin{tikzpicture}[every state/.append style={thick, inner sep=0pt, minimum size=2cm}]
			\draw[help lines, use as bounding box,white] (-1,-2) grid (12, 1.25);
			\node[state, accepting] (c0) at (0, 0) {$%
				\begin{array}{rc}%
					x\colon &  0\\[.5em]
					\textcolor{gray}{p\colon} & \textcolor{gray}{{-}}\\
					\textcolor{gray}{d\colon} & \textcolor{gray}{{-}}
				\end{array}%
			$};
			
			\node[state] (c1) at (3, 0) {$%
				\begin{array}{rc}%
					x\colon &  1\\[.5em]
					\textcolor{gray}{p\colon} & \textcolor{gray}{ \NF{1}{2}}\\
					\textcolor{gray}{d\colon} & \textcolor{gray}{ 1}
				\end{array}%
			$};
			
			\node[state] (c2) at (6, 0) {$%
				\begin{array}{rc}%
					x\colon &  2\\[.5em]
					\textcolor{gray}{p\colon} & \textcolor{gray}{ \NF{1}{3}}\\
					\textcolor{gray}{d\colon} & \textcolor{gray}{ 1}
				\end{array}%
			$};
			
			\node[state] (c3) at (9, 0) {$%
				\begin{array}{rc}%
					x\colon &  3\\[.5em]
					\textcolor{gray}{p\colon} & \textcolor{gray}{ \NF{1}{4}}\\
					\textcolor{gray}{d\colon} & \textcolor{gray}{ 1}
				\end{array}%
			$};
			
			\node[state, draw=white] (dots) at (12, 0) {$\cdots$};

			
			\path[->,thick] (c1) edge[left] node[below] {$\NF{1}{2}$} (c0);
			\path[->,thick] (c2) edge[bend left] node[below, pos=0.2] {$\NF{1}{3}$} (c0);
			\path[->,thick] (c3) edge[bend left] node[below, pos=0.2] {$\NF{1}{4}$} (c0);
			\path[->,thick] (dots) edge[bend left] node[below, pos=0.2] {$\NF{1}{5}$} (c0);

			\path[->,thick] (c1) edge[bend left] node[above] {$\NF{1}{2}$} (c2);
			\path[->,thick] (c2) edge[bend left] node[above] {$\NF{2}{3}$} (c3);
			\path[->,thick] (c3) edge[bend left] node[above] {$\NF{3}{4}$} (dots);
		\end{tikzpicture}
		\arraycolsep=\oldarraycolsep
	\end{center}
	\caption{Execution of the escaping spline loop. 
	The value of the variant is equal to the value of the variable $x$ in each state. 
	Inside the nodes we give the valuations of variable $x$ as well as the values of the probability function $p$ and the decrease function $d$ in each state. Note that in \Thm{t1928} it does not matter what $d,p$'s values are when $V{=}0$, because the \LHS\ of \Itm{i1651-3} is zero in that case.
	}
\label{fig:escaping-spline-loop}
\end{figure}

Intuitively, the loop models a random walk of a particle $x$ that terminates when the particle hits the origin $0$.
The random walk either with probability $\NF{1}{x + 1}$ immediately terminates or with probability $\NF{x}{x + 1}$ increments the position of $x$ by one.
This means that for each iteration where the loop does not terminate, it is even \emph{more likely not to terminate in the next iteration}.
Thus, the longer the loop runs, the less likely it will terminate since the probability to continue looping approaches 1 asymptotically.
Yet this loop terminates almost-surely, as we will now prove.

\paragraph{Proof of almost-sure termination.}
The loop guard is given by\quad$G \eeq x {>} 0$\quad
and the loop body by
\begin{align*}
	C \eeq \COMPOSE{\ASSIGN{q}{\NF{1}{x+1}}}{\PCHOICE{\ASSIGN{x}{0}}{q}{\ASSIGN{x}{x + 1}}}~.
\end{align*}%
We now prove almost-sure termination of the loop by choosing the standard invariant $I = \true$ and an appropriate $p,d$ and quasi-variant $V$:
\begin{align*}
	V \eeq x, \qquad\textnormal{for } d(v) \eeq 1 \quad\textnormal{and}\quad p(v) \eeq \frac{1}{v+1}~.
\end{align*}
Intuitively this tells us that the variant $V$, i.e. the value of $x$, decreases with probability at least $\NF{1}{V + 1} = \NF{1}{x + 1}$ by at least $1$  through an iteration of the loop body in case that the guard is satisfied. Now $V$ satisfies the super-martingale property:
\begin{align*}
		\iverson{G} \cdot \awp{C}{V}  &\eeq \iverson{x{>}0} \cdot \awp{\COMPOSE{\ASSIGN{q}{\NF{1}{x+1}}}{\PCHOICE{\ASSIGN{x}{0}}{q}{\ASSIGN{x}{x + 1}}}}{x} \\
		&\eeq \iverson{x{>}0} \cdot \awp{\ASSIGN{q}{\NF{1}{x+1}}}{\left( q \cdot 0 + (1-q) \cdot (x + 1) \right)} \\
		&\eeq \iverson{x{>}0} \cdot \left(1-\frac{1}{x+1}\right) \cdot (x + 1)\\
		&\eeq \iverson{x{>}0} \cdot \left(x + 1 - 1\right)\\
		& \eeq \iverson{x{>}0}\cdot x \\
		& \lleq x \\
		& \eeq V~.
\end{align*}
And $V$, $p$, and $d$ satisfy the progress condition for all $R$:
\begin{align*}
		&p(R) \cdot \iverson{G \land I \land x{=}R} \lleq \wp{C}{\iverson{V \leq R - d(R)}} \\
		\Longleftrightarrow \quad &\frac{1}{R + 1} \cdot \iverson{x{>}0 \land x{=}R} \lleq  \wp{\COMPOSE{\ASSIGN{q}{\NF{1}{x+1}}}{\PCHOICE{\ASSIGN{x}{0}}{q}{\ASSIGN{x}{x + 1}}}}{\iverson{x \leq R{-}1}} \\
		%
		\Longleftrightarrow \quad &\frac{1}{R + 1} \cdot \iverson{x{>}0 \land x{=}R} \lleq  \wp{\ASSIGN{q}{\NF{1}{x+1}}}{\left(q \cdot \iverson{0 \leq R - 1} + (1{-}q) \cdot \iverson{x{+}1 \leq R - 1}\right)} \\
		%
		\Longleftrightarrow \quad &\frac{1}{R + 1} \cdot \iverson{x{>}0 \land x{=}R} \lleq  \NF{1}{x+1} \cdot \iverson{0 \leq R - 1} + \NF{x}{x+1} \cdot \iverson{x{+}1 \leq R - 1} \\
		\Longleftrightarrow \quad &\frac{1}{R + 1} \cdot \iverson{R{>}0 \land x{=}R} \lleq  \NF{1}{R+1} \cdot \iverson{0 \leq R - 1 \land x{=}R} + \NF{R}{R+1} \cdot \iverson{R{+}1 \leq R - 1 \land x{=}R} \\
		\Longleftrightarrow \quad &\frac{1}{R + 1} \cdot \iverson{R{>}0 \land x{=}R} \lleq  \NF{1}{R+1} \cdot \iverson{0 \leq R - 1 \land x{=}R} \\
		\Longleftarrow{} \quad &x \in \Nats~. \tag{true by assumption}\\
\end{align*}
This shows that all preconditions of Theorem~\ref{t1651} are satisfied and as a consequence the escaping spline loop terminates almost-surely.

In fact in retrospect \AST\ for this loop is not so surprising after all: by inspection, the probability associated with the sole diverging path from say $x{=}1$ is $\NF{1}{2}\cdot\NF{2}{3}\cdots = 0$. It is interesting however that this criterion applies in general: if the probability of going up from $x$ is $p_x$, then the variant $V(x)= \NF{1}{p_1p_2\cdots p_{x-1}}$ is a martingale by construction. And if $p_1p_2\cdots >0$, i.e.\ the probability of divergence is non-zero, then this variant is bounded and, for reasons discussed below at \Cor{c1026}, it therefore acts as a certificate for \emph{non-termination}.  Moreover, as illustrated in \Sec{s1643}, indeed our \Thm{t1928} does not apply when $p_1p_2\cdots >0$ since then there is no everywhere positive but antitone $p()$.
\footnote{~If $\NF{1}{p_1p_2\cdots} =K<\infty$ then necessarily the escape probabilities $1{-}p(v)$ tend to zero as $V(x){=}v$ tends to $K$, and so $p(v)$ for any $v{>}K$ must actually \emph{be} zero --- which is not allowed, even if the process never reaches $x$ with $V(x){>}K$.}
If however $p_1p_2\cdots =0$, i.e.\ the probability of divergence is zero, then the construction $V(x)= \NF{1}{p_1p_2\cdots p_{x-1}}$ works (because the variant is unbounded) --- a (limited) completeness property.

{\Fx
\subsection{The Lazy Loper}\label{s1117}

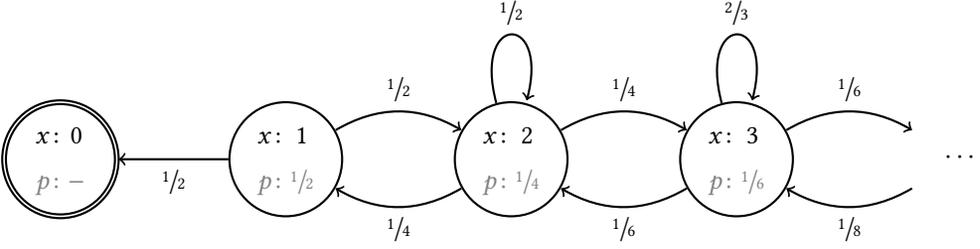
\begin{figure}[t]
	\begin{center}
		\let\oldarraycolsep\arraycolsep
		\arraycolsep=0pt
		\begin{tikzpicture}[every state/.append style={thick, inner sep=0pt, minimum size=1.5cm}]
			\draw[help lines, use as bounding box,white] (-0.75,-1) grid (12, 2.25);
			\node[state, accepting] (c0) at (0, 0) {$%
				\begin{array}{rc}%
					x\colon &  0\\[.5em]
					\textcolor{gray}{p\colon} & \textcolor{gray}{{-}}
				\end{array}%
			$};
			
			\node[state] (c1) at (3, 0) {$%
				\begin{array}{rc}%
					x\colon &  1\\[.5em]
					\textcolor{gray}{p\colon} & \textcolor{gray}{ \sfrac{1}{2}}
				\end{array}%
			$};
			
			\node[state] (c2) at (6, 0) {$%
				\begin{array}{rc}%
					x\colon &  2\\[.5em]
					\textcolor{gray}{p\colon} & \textcolor{gray}{ \sfrac{1}{4}}
				\end{array}%
			$};
			
			\node[state] (c3) at (9, 0) {$%
				\begin{array}{rc}%
					x\colon &  3\\[.5em]
					\textcolor{gray}{p\colon} & \textcolor{gray}{ \sfrac{1}{6}}
				\end{array}%
			$};
			
			\node[state, draw=white] (dots) at (12, 0) {$\cdots$};

			
			\path[->,thick] (c1) edge[left] node[below] {$\sfrac{1}{2}$} (c0);
			\path[->,thick] (c2) edge[bend left] node[below] {$\sfrac{1}{4}$} (c1);
			\path[->,thick] (c3) edge[bend left] node[below] {$\sfrac{1}{6}$} (c2);
			\path[->,thick] (dots) edge[bend left] node[below] {$\sfrac{1}{8}$} (c3);

			\path[->,thick] (c1) edge[bend left] node[above] {$\sfrac{1}{2}$} (c2);
			\path[->,thick] (c2) edge[bend left] node[above] {$\sfrac{1}{4}$} (c3);
			\path[->,thick] (c3) edge[bend left] node[above] {$\sfrac{1}{6}$} (dots);

			\path[->,thick] (c2) edge[loop above] node[above] {$\sfrac{1}{2}$} (c2);
			\path[->,thick] (c3) edge[loop above] node[above] {$\sfrac{2}{3}$} (c3);
		\end{tikzpicture}
		\arraycolsep=\oldarraycolsep
	\end{center}
	\caption{Transition system for the Lazy Loper program. 
	Inside the nodes we give the valuations of variable $x$ as well as the value of the probability function $p$. 
	The value of the variant is equal to the value of variable $x$ in each state.
	The value of the decrease function $d$ is constantly $1$.}
	\label{fig:1117}
\end{figure}

The Lazy Loper is a random walker that ``dawdles'' at $x$ before finally moving either up to $x{+}1$ or down to $x{-}1$. The code is
\begin{align*}
	\WHILEDO{x > 0}{\quad\PCHOICE{
	  \PCHOICE{\ASSIGN{x}{x{+}1}}{\sfrac{1}{2}}{\ASSIGN{x}{x{-}1}}}{1/x}{\SKIP}\quad},
\end{align*}
and it  corresponds to the transition system in \Fig{fig:1117} where a walker flips a \emph{biased} coin so that the larger the (integer) value of $x$, the more likely it is that the state remains unchanged (i.e.\ by selecting the $\SKIP$ branch). When however $x$ is (eventually) updated, as it \AS\ must be, it is either incremented or decremented with the choice between the two options determined fairly just as in the ordinary 1dSRW.

Informally we can see that the loop terminates almost surely, since at any value of $x$ it is guaranteed eventually to select the left-hand branch of the outer probabilistic choice;  then the overall ``movement behaviour''  becomes that of an unbiased random walker, albeit one who remains in the same position for longer and longer periods the greater the distance from $0$ .  

Formally, we can prove termination using \Thm{t1928}:  we take $V(x)= x$ for the super-martingale,  and $p(v)= (1/2v \min 1)$ and $d(v) = 1$ for  $p, d$ progress. It is clear that the super-martingale is reduced by $1$ with probability $p(v)$. 

Observe also that the average \emph{absolute move} of $V$ on each step is $\NF{1}{x} \cdot 1 + \NF{x{-}1}{x}\cdot 0 = \NF{1}{x}$ which approaches $0$ as $x$ approaches infinity; that seems to put this choice of variant beyond the reach of \citeauthor{Chatterjee:2017ab}'s Thm.~5 \citeyear {Chatterjee:2017ab}, as we remark in \Sec{s1521}.

But we can argue further that \emph{no} variant in the style of \cite{Chatterjee:2017ab} suffices for their Thm.~5 in this case. That is, if $V(x)$ is \emph{any} non-negative super-martingale for the Lazy Loper, its average absolute move for each iteration must also become arbitrarily small as $x$ becomes arbitrarily large.  We reason as follows.

\begin{enumerate}
\item Note first that the super-martingale property implies that, for all $x\geq 1$, we have
 \[
 V(x{+}1) + V(x{-}1) \Wide{\leq} 2V(x)~ .
 \]

\item From (1) we see that \emph{either}
\begin{enumerate}
\item $0 \leq |V(x{+}1) - V(x)| \leq |V(x) - V(x{-}1)|$ for all $x$,~ \emph{or}
\item there is some $N{>}0$ such that $V(x{+}1) \leq V(x)$ for all $x{\geq}N$.
\end{enumerate}
\end{enumerate}

We note that (2)(a) follows if  $V(x{+}1) \geq V(x)$ for all $x$. However if \emph{ever} $V(x{+}1) \leq V(x)$ then so too must  $V(x{+}2) \leq V(x{+}1)$, from which (2)(b) follows.  

To see that $V(x{+}1) \leq V(x)$ implies $V(x{+}2) \leq V(x{+}1)$, we reason as follows:

\medskip
\begin{align*}
	& { V(x{+}2) + V(x) \leq 2V(x{+}1)} \tag*{(1): $V$ is a super-martingale} \\
	\Longrightarrow\quad & V(x{+}2) + V(x) \leq 2V(x) \tag*{assumption $V(x{+}1) \leq V(x)$} \\
	\Longleftrightarrow\quad & V(x{+}2)\leq V(x) \tag*{arithmetic} \\
	\Longleftrightarrow\quad & 2V(x{+}2)\leq V(x{+}2) + V(x) \tag*{arithmetic} \\
	\Longrightarrow\quad & 2V(x{+}2)\leq 2V(x{+}1) \tag*{(1): $V$ is a super-martingale} \\
	\Longleftrightarrow\quad & V(x{+}2)\leq V(x{+}1)\quad . \tag*{arithmetic} 
\end{align*}


\bigskip

Finally we can see that (a) and (b) together imply that the expected average move of the super-martingale $V$ for each step of the Lazy Loper is bounded above by $\NF{1}{x}\cdot A$ for some (possibly large) constant $A{>}0$ and therefore, as $x$ approaches infinity, the average absolute move must approach zero, as required to exclude \cite{Chatterjee:2017ab}'s Thm.~5. 
}

{\Fx
\subsection{The \emph{\underline{Ver}y} Lazy Loper: Nested Loops, Program Algebra and Lexicographic Variants}\label{s0836}

\begin{figure}[t]
	\begin{center}
		\let\oldarraycolsep\arraycolsep
		\arraycolsep=0pt
		\begin{tikzpicture}[every state/.append style={thick, inner sep=0pt, minimum size=1.5cm}]
			\draw[help lines, use as bounding box,white] (-0.75,-2) grid (12, 2.25);
			\node[state, accepting] (c0) at (0, 0) {$%
				\begin{array}{c}%
					x\colon  0\\[.5em]
					\textcolor{gray}{n\kern.1em{\geq}1}
				\end{array}%
			$};
			
			\node[state,initial below, initial text=\mbox{$x=1 \land n\geq1$}] (c1) at (3, 0) {$%
				\begin{array}{c}%
					x\colon  1\\[.5em]
					\textcolor{gray}{n\kern.1em{\geq}1}
				\end{array}%
			$};
			
			\node[state] (c2) at (6, 0) {$%
				\begin{array}{c}%
					x\colon  2\\[.5em]
					\textcolor{gray}{n\kern.1em{\geq}1}
				\end{array}%
			$};
			
			\node[state] (c3) at (9, 0) {$%
				\begin{array}{c}%
					x\colon  3\\[.5em]
					\textcolor{gray}{n\kern.1em{\geq}1}
				\end{array}%
			$};
			
			\node[state, draw=white] (dots) at (12, 0) {$\cdots$};

			
			\path[->,thick] (c1) edge[left] node[below] {$\sfrac{1}{2n}$} (c0);
			\path[->,thick] (c2) edge[bend left] node[below] {$\sfrac{1}{2n}$} (c1);
			\path[->,thick] (c3) edge[bend left] node[below] {$\sfrac{1}{2n}$} (c2);
			\path[->,thick] (dots) edge[bend left] node[below] {$\sfrac{1}{2n}$} (c3);

			\path[->,thick] (c1) edge[bend left] node[above] {$\sfrac{1}{2n}$} (c2);
			\path[->,thick] (c2) edge[bend left] node[above] {$\sfrac{1}{2n}$} (c3);
			\path[->,thick] (c3) edge[bend left] node[above] {$\sfrac{1}{2n}$} (dots);

			\path[->,thick] (c1) edge[loop above] node[above] {$\begin{array}{c}1{-}\NF{1}{n}\\\mbox{\small$\ASSIGN{n}{n{+}1}$}\end{array}$} (c1);
			\path[->,thick] (c2) edge[loop above] node[above] {$\begin{array}{c}1{-}\NF{1}{n}\\\mbox{\small$\ASSIGN{n}{n{+}1}$}\end{array}$} (c2);
			\path[->,thick] (c3) edge[loop above] node[above] {$\begin{array}{c}1{-}\NF{1}{n}\\\mbox{\small$\ASSIGN{n}{n{+}1}$}\end{array}$} (c3);
		\end{tikzpicture}
		\arraycolsep=\oldarraycolsep
	\end{center}
	\caption{Transition system for the Very Lazy Loper program. 
	As in \Fig{fig:1117}, inside the nodes we give the valuation of variable $x$; but here the value of $n$ depends on how many self-loop steps have been taken, not on the walker's position. In spite of that, \AST\ is still guaranteed for any initial $n{\geq}1$.}
	\label{f0936}
\end{figure}
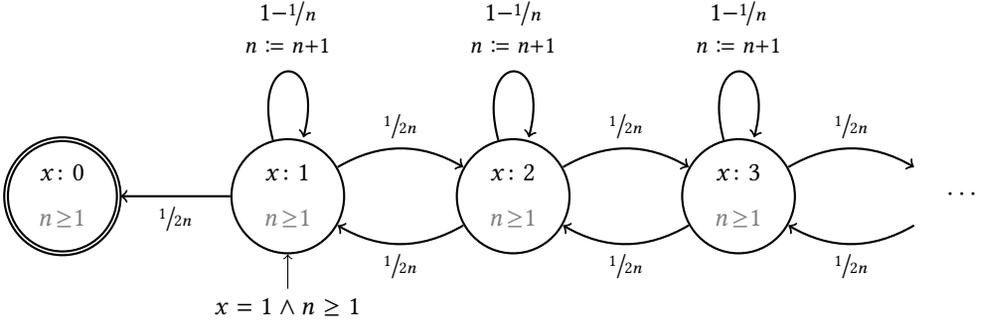

The Very Lazy Loper, like the Lazy Loper, increases/decreases $x$ only after dawdling possibly for some time at $x$'s current value.
In the ``very'' case, however, the dawdling time remorselessly increases, independently of $x$. This is the code of the Very Lazy Loper:
\begin{Equation}\label{e1801}
	\begin{array}{l}
		\{n\geq1\} \\
		\ASSIGN{x}{1} \\
		\WHILE{x{\neq}0} \\
		\quad	\PCHOICE{
					\PCHOICE{\ASSIGN{x}{x{-}1}
					}{\NF{1}{2}
					}{\ASSIGN{x}{x{+}1}
					}
				}{\NF{1}{n}
				}{\ASSIGN{n}{n{+}1}
				}\\
		\}\quad.
	\end{array}
\end{Equation}
It differs from our earlier (moderately) Lazy Loper of \Fig{fig:1117} in that for the loitering probability we use $(\PC{\NF{1}{n}})$ rather than $(\PC{\NF{1}{x}})$.


In the style of \Fig{fig:1117}, the Very Lazy Loper's transitions would be as in \Fig{f0936}, where the differences (and similarities) are clear: the transition probabilities now depend on a variable $n$, not on the position $x$ of the loper; and the self-loops update a counter $n$.
Working directly from the source code of \Eqn{e1801} however gives us \Fig{f0812}. 
Note however that we won't  use either of those figures in our formal reasoning; they are only for intuition. Instead we work from the text of Program \Eqn{e1801}, i.e.\ from the source code directly, 
and show that it is equivalent to this program, in which one loop is nested within another:
\begin{figure}[t]
	\newcommand\Scale {0.8}
	\begin{tikzpicture}[every state/.append style={thick, inner sep=.2em, minimum size=1.5cm}, every node/.style={scale=\Scale}, scale=\Scale]
		\draw[use as bounding box, white] (0,2) grid (12,9);
		\node[state, initial above, initial text=\mbox{$x=1 \land n\geq1$}] at (5,7) (n1) {\raisebox{.3em}{$\mathtt{if}~x{=}0$}};
		\node[state,accepting] at (3,5) (n2) {$x=0$};
		\node[state] at (7,5)  (n3) {\huge $\PC{\NF{1}{n}}$};
		\node[state] at (5,3) (n4) {\tiny$\begin{array}[t]{@{}c}x:=\\~x{-}1\PC{1/2}\kern.1emx{+}1~\end{array}$};
		\node[state] at (9,3) (n5) {$~\ASSIGN{n}{n{+}1}~$};

		\path[->,thick] (n1) edge [left] node [yshift = .5em] {$x{=}0$} (n2);
		\path[->,thick] (n1) edge [right] node [yshift = .5em] {$x{\neq}0$} (n3);
		\path[->,thick] (n3) edge [left] node [yshift = .5em] {$\NF{1}{n}$} (n4);
		\path[->,thick] (n3) edge [right] node [yshift = .5em] {$1-\NF{1}{n}$} (n5);
		\draw[->,thick] (n4.west) .. controls (0,3) and (0,7) ..  node [xshift = -1em] {$\beta$}(n1.west);
		\draw[->,thick] (n5.east) .. controls (13,3) and (10,7.2) .. node [xshift = 1.3em] {$\alpha$} (n1.east);
	\end{tikzpicture}
	\caption{\Fx The \emph{Very} Lazy Loper ---
		The random walk on $x$ is delayed by dawdling at any $x$ by successively more likely self-loops
		via the arc $\alpha$.
		The loitering is controlled by $n$, unrelated to $x$; and $n$ never decreases  even though $x$ might.
		An informal argument for termination is that the loitering itself terminates \AS\ (no matter how large $n$ might be),
		because the probability of incrementing $n$ forever is $\NF{N-1}{N}{\cdot}\NF{N}{N+1}{\cdot}\NF{N+1}{N+2}\cdots = 0$,
		where $N$ is the value of $n$ as the loop is entered initially or, later, re-entered from $\beta$. (A more rigorous proof for this part would
		be analogous to Escaping Spline example of \Sec{s0838}.)
		And so when the inner loop exits, as it \AS\ must, then $x$ will move one step up or down as in the normal symmetric random walk.
	}

	\label{f0812}
\end{figure}
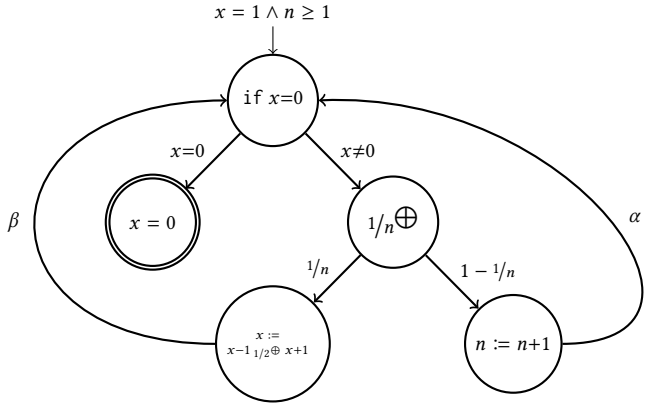

\begin{figure}[t]
	\newcommand\Scale {0.7}
	\begin{tikzpicture}[every state/.append style={thick, inner sep=.2em, minimum size=1.5cm}, every node/.style={scale=\Scale}, scale=\Scale]
		\draw[use as bounding box, white] (0,2) grid (12,9);
		\node[state, initial above, initial text=\mbox{$x=1 \land n\geq1$}] at (5,7) (n1) {\raisebox{.3em}{$\mathtt{if}~x{=}0$}};
		\node[state,accepting] at (3,5) (n2) {$x=0$};
		\node[state] at (7,5)  (n3) {\huge $\PC{\NF{1}{n}}$};
		\node[state] at (5,3) (n4) {\tiny$\begin{array}[t]{@{}c}x:=\\~x{-}1\PC{1/2}\kern.1emx{+}1~\end{array}$};
		\node[state] at (9,3) (n5) {$~\ASSIGN{n}{n{+}1}~$};

		\path[->,thick] (n1) edge [left] node [yshift = .5em] {$x{=}0$} (n2);
		\path[->,thick] (n1) edge [right] node [yshift = .5em] {$x{\neq}0$} (n3);
		\path[->,thick] (n3) edge [left] node [yshift = .5em] {$1/n$} (n4);
		\path[->,thick] (n3) edge [right] node [yshift = .5em] {$1-1/n$} (n5);
		\draw[->,thick] (n4.west) .. controls (0,3) and (0,7) ..  (n1.west);
		\draw[->,thick] (n5.east) .. controls (13,3) and (10,5.25) .. node [xshift = 2.2em, yshift = -3.3ex] {$\alpha'$} (n3.east);
	\end{tikzpicture}
	\caption{\Fx The Very Lazy Loper again, but rewritten with an inner loop: see Program \Eqn{e1814}.
	We have moved arc $\alpha$ in \Fig{f0812} from terminating at $\mathtt{if}~x{=}0$ to position $\alpha'$ terminating at $\PC{\NF{1}{n}}$
	because at the end of arc $\alpha$ we know already that $x{\neq}0$. That is the only change.}
	\label{f1733}
\end{figure}
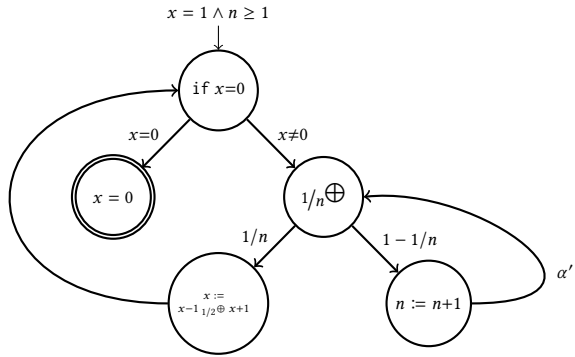

\begin{Equation}\label{e1814}
	\begin{array}{l}
		\{n\geq1\} \\
		\ASSIGN{x}{1} \\
		\WHILE{x{\neq}0} \\
		\quad	\WHILEDO{1{-}\NF{1}{n}}{\ASSIGN{n}{n{+}1}} \quad\footnotemark\\
		\quad	\PCHOICE{\ASSIGN{x}{x{-}1}
					}{\NF{1}{2}
					}{\ASSIGN{x}{x{+}1}
					}\\
		\} \quad.
	\end{array}
\end{Equation}
\footnotetext{This probabilistic while-loop enters the loop with the probability shown, otherwise terminates. It can easily be written in the conventional style with the help of an auxiliary Boolean, but the probabilistic guard reduces clutter. It is defined in \cite[Sec.~7.7]{McIver:05a}; in any case the with-Boolean version is given at \Eqn{e1516}.}%
Its transition diagram would look as in \Fig{f1733}, again an informal remark. But it supports a plausible argument for that equality: that the transition $\alpha$ in \Fig{f0812} can be moved to its position $\alpha'$ in \Fig{f1733} because at the point it is taken it is known already that $x{\neq}0$, and so the test ``$\mathtt{if}~x{=}0$'' always takes its $x{\neq}0$ branch if reached via $\alpha$ --- in effect allowing the arrowhead of $\alpha$ to be moved down to $\alpha'$ on the right-hand path. We show rigorously the equality $\textrm{\Eqn{e1801}}{=}\textrm{\Eqn{e1814}}$ in \App{s1315}.

With that equality, the proof of \AST\ for Program \Eqn{e1801} is relatively simple:
\footnote{A direct proof of \AST\ for \Eqn{e1801} should be simple too: there should be a single variant for that in our style, a single expression in $x,n$. But it does not seem simple to find --- we do yet not know what that variant is. Perhaps a lexicographic variant would be easier to find \cite{Agrawal:2018}, although we do not know whether lexicographic variants represent an increase in power.}
we prove \AST\ for Program \Eqn{e1814} instead. The variant for the outer loop is $x$, as for the ordinary 1dSRW; and we must show that the probability of $x$'s decrease (by at least 1) via the outer loop's body is at least $\NF{1}{2}$. For that, with a separate nested argument, we use variant $n$ (more or less) in the inner loop to show its \AST, and then use \Thm{T0909} to show that $x$ is a probabilistic invariant of that loop.
\footnote{Actually we would be using $H{-}x$ for arbitrary $x$ if we were following \Thm{t1928} exactly.}
Thus the outer loop is guaranteed to decrease $x$ by 1 with probability $\NF{1}{2}$ \emph{eventually}.
\footnote{A similar ``eventually the variant will decrease'' argument is for the example of ``the random stumbler'' in the quantitative temporal logic of \cite[Sec.~10.4.2]{McIver:05a}.}

We now give the proof of termination of Program \Eqn{e1814}, comprising an outer loop containing an inner loop. We use an ``inner variant'' based on $n$ and an ``outer variant'' based on $x$, in summary as follows:
\begin{itemize}
\item We use our new \Thm{t1651} and variant more-or-less $n$ (as in The Escaping Spline of \Sec{s0838}) to show \AST\ for the inner loop on $n$. Then, observing that $x$ is a probabilistic invariant of that loop (\Def{d0923}) we use \Thm{thm:lem-2-4-1} to conclude that $x$ is preserved by the inner loop as a whole, which is a precursor to showing it is a variant for the outer loop.
\item The facts just established for the inner loop are then used to show that $x$ is a super-martingale (actually that $H{\ominus}x$ is a sub-martingale) for the outer loop.
\item \AST\ for the outer loop is then shown by a second application of \Thm{t1651}, this time with variant $x$.
\end{itemize}

Here is the proof in more detail:
\begin{enumerate}[(a)]
\item \label{i1616} The inner loop $\WHILEDO{1{-}\NF{1}{n}}{\ASSIGN{n}{n{+}1}}$ is similar to \Sec{s0838} --- the difference is that the inner loop does not set $n$ to zero in order to terminate (as \Sec{s0838} does). To be very clear,
\footnote{In \cite[Sec.~7.7.5]{McIver:05a} we give termination rules for probabilistic guards directly; but we have not yet extended them to take advantage of the new $p$,$d$-parametric technique we present here.}
we therefore introduce a local Boolean variable $b$ to control termination explicitly and re-write the loop as
\begin{Equation}\label{e1516}
    \begin{array}{l}
        \PCHOICE{\ASSIGN{b}{\false}}{\NF{1}{n}}{\ASSIGN{b}{\true}} \\
        \WHILE{b} \\
        \quad \ASSIGN{n}{n{+}1} \\
        \quad\PCHOICE{\ASSIGN{b}{\false}}{\NF{1}{n}}{\ASSIGN{b}{\true}} \\
        \}
    \end{array}
\end{Equation}
For variant $V$ we use ${\iverson{b}}{*}n$, that is $n$ itself when $b$ is true and zero otherwise; the invariant $I$ is $\true$. (We are thus using $b$ and $n$ together to mimic the variant $x$ in \Sec{s0838}.) Then we reason
\begin{align*}
		                     & \awp{\PCHOICE{(\ASSIGN{b}{\false}}{\NF{1}{n}}{\ASSIGN{b}{\true}})}{V} \\
	=\quad		    & \awp{\PCHOICE{(\ASSIGN{b}{\false}}{\NF{1}{n}}{\ASSIGN{b}{\true}})}{({\iverson{b}}{*}n)}
						\eeq \NF{1}{n}{\cdot}0+(1{-}\NF{1}{n}){\cdot}n \eeq n{-}1 \quad, \\
	\textrm{and}\quad & \awp{(\ASSIGN{n}{n{+}1})}{(n{-}1)} \eeq n \quad, \\
	\textrm{and}\quad & n \eeq {\iverson{b}}{*}n \eeq V\textrm{~at loop entry}\quad, \tag*{because $b$ is true on loop entry}
\end{align*}
so that $V$ is a martingale. And $V$ decreases by at least 1 with probability $\NF{1}{n}$ when $V{>}0$, so that $p(V)=\NF{1}{V}$ and $d(V){=}1$ suffices. Thus the inner loop \Eqn{e1516} terminates \AS.

\item We now use \Thm{thm:lem-2-4-1} to show that the inner loop ``preserves'' $x$ --- but we recall that the $\mathit{Sub}$ used in that theorem (our $x$) must be bounded --- and $x$ is not bounded. Accordingly we use $H{\ominus}x$, for arbitrary $H$ as in \Thm{t1651}, reasoning
\begin{align*}
		                     & \wp{\PCHOICE{(\ASSIGN{b}{\false}}{\NF{1}{n}}{\ASSIGN{b}{\true}})}{H{\ominus}x} \eeq H{\ominus}x \quad, \\
	\textrm{and}\quad & \wp{(\ASSIGN{n}{n{+}1})}{(H{\ominus}x)} \eeq H{\ominus}x \quad, \\
	\textrm{and}\quad & {\iverson{b}}{*}(H{\ominus}x) \leq (H{\ominus}x)\quad, 
\end{align*}
so that with  \Thm{thm:lem-2-4-1} and \Itm{i1616} we indeed have $H{\ominus}x \leq \wp{\textrm{$\mathit{INNER}$-$\mathit{LOOP}$}}{(H{\ominus}x)}$.

\item We now reason over the whole of the body of the outer loop, i.e.\ including the assignment to $x$, to show that $H{\ominus}x$ is a sub-martingale for the \emph{outer} loop: we have
\begin{align*}
		                     & \wp{\PCHOICE{(\ASSIGN{x}{x{-}1}}{\NF{1}{2}}{\ASSIGN{x}{x{+}1}})}{(H{\ominus}x)} \\
	=\quad		    & (H\ominus(x{-}1))/2 + (H\ominus(x{+}1))/2\quad\geq\quad H\ominus x \tag*{(Careful with the $(\geq)$! See \Lem{l1331}.)} \quad, \\
	\textrm{and}\quad & \wp{\textrm{$\mathit{INNER}$-$\mathit{LOOP}$}}{(H{\ominus}x)} \quad\geq\quad (H{\ominus}x)\quad, \tag*{from above}
\end{align*}
as required.

\item The remainder of the \AST\ proof for the whole of Program \Eqn{e1814} is now just as for the 1dSRW, i.e.\ using variant $x$ and $p(x){=}\NF{1}{2}$ and $d(x){=}1$.
\end{enumerate}

This  \emph{Very Lazy Loper} example was inspired by two features of \cite{Agrawal:2018}. The first was that the VLL looks like a good target for lexigographic techniques: one would use $x$ for the ``major'' component and $n$ for the minor, with the lexicographic aspect being that when $x$ (probabilistically) decreases, it does not matter what happens to $n$. Here we have sidestepped that by using nested loops; but there is no guarantee that such tricks would work in general.

The second feature was that their Example 4.8 (Figure 2) [\textit{op.\ cit.}]\ admits a very direct argument, without lexicography, if one works with the source code in (something like) \pGCL: their example program is the sequential composition of two while-loops $A,B$ say, and each of those loops trivially terminates \AS. Thus \AST\ for the whole program follows immediately from the fact that
\[
	\wp{(A\,;B)}{1} \Wide{=} \wp{A}{(\wp{B}{1})} \Wide{=} \wp{A}{1} \Wide{=} 1~.
\]
Thus this example tangentially makes the case for considering algebraic reasoning as part of the ``arsenal'' for showing \AST.
}

\section{Review of mathematical literature on super-martingale methods}\label{s1530} 
\subsection{Recurrent Markov Chains, and Super-Martingales}

Early work on characterising recurrent behaviours of infinite-state Markov processes using super-martingale methods is primarily due to \citet{Foster:1951aa, Foster:1952ab}, \citet{Kendall:1951aa} and \citet{Blackwell:55}. In this section  we review some of these important results and explain how they relate to \AST\ for probabilistic programs and \Thm{t1651}. Note that their arguments are given directly in an underlying model of (deterministic) transition systems.

 Following the conventions of the authors above, we assume an enumeration of the (countable) state space $i= 0, 1, 2,\dots$, and transition probabilities $p_{ij}$ for the probability of transitioning from state $i$ to state $j$. The probability of reaching $j$ from $i$ on the $n$'th transition is $p^n_{ij}$, where  $p^n$ is computed from single transitions $p_{ik}$  using matrix multiplication.
\citet{Foster:1951aa} identified three kinds of long-term average behaviours for infinite-state Markov processes, which behaviours he called dissipative, semi-dissipative and non-dissipative.    A process is said to be \emph{non-dissipative} if  its long term average behaviour does not ``dissipate'', i.e.\
if  $\sum_{j\geq 0}\pi_{ij}=1$ for all states $i$, where $\pi_{ij}= \lim_{n\rightarrow \infty}\frac{1}{n}\sum_{r=1}^np^r_{ij}$ \citep{Kendall:1951aa}. An illustration of a \emph{dissipative} process is the biased random walk, with an extreme example given by transition probabilities $p_{i(i{+}1)}=1$.  The non-dissipative condition is more general than \AST, but the methods used to prove that a process is non-dissipative nevertheless do use super-martingales. 
 In particular Foster's Theorem 5 \citeyear{Foster:1951aa} gives such a sufficient condition for a process to be non-dissipative. It is
 \begin{Equation}\label{e1113}
 \sum_{j\geq 0} j{\cdot}p_{ij}\leq i~, \quad\textit{for all states $i\geq 0$~.
} 
 \end{Equation} 

\citet{Kendall:1951aa} generalised Foster's \Eqn{e1113} by removing the strict relation between the ``super-martingale'' values and the enumeration of the state space, whilst articulating an important finitary property of  a super-martingale that he used in his proof. In Kendall's work, a Markov process is guaranteed to be non-dissipative if there is a function $V$ from states to reals such that 
\begin{Equation}\label{e1156}
 \sum_{j\geq 0} V(j){\cdot}p_{ij}\leq V(i) \quad\textit{for all states $i{\geq}0$}
\end{Equation}
\emph{and} for each value $\delta{\geq}0$ there are only finitely many states $i$ such that $V(i)\leq \delta$.   Finiteness is crucial here: for the dissipative process with $p_{i(i{+}1)} = 2/3$ and $p_{i(i{-}1)}=1/3$ (which we return to in \Sec{s1643}) we have $V(i)= \pi_{i0}$ satisfies \Eqn{e1156} but, of course, in general $\sum_{j\geq 0}\pi_{ij} = \pi_{i0}< 1$, since it can be shown that $\pi_{i0}$ is the probability of ever reaching $0$ from $i$.

Then \citet{Blackwell:55} further developed the ideas of Foster and Kendall (sketched above) in order to obtain a complete characterisation of Markov-process behaviour in terms of martingales (i.e.\ exact); some of Blackwell's results can be adapted to work for probabilistic programs generally to provide a certificate to prove \emph{non}-\AST. We summarise Blackwell's results here and then show how we can apply them. We continue with the historical notations.

Let $C$ be a subset of the state space, and fix some initial state $\hat{i}$. Say that $C$ is \emph{almost closed} (with respect to that $\hat{i}$\,) iff the following conditions hold:
\begin{enumerate}
\item The probability that $C$ is entered infinitely often, as the process takes transitions (initially) starting from  $\hat{i}$, is strictly greater than zero and
\item If $C$ is indeed visited infinitely often, starting from $\hat{i}$, then eventually the process remains within $C$ permanently.
\end{enumerate}

Say further that a set $C$ is \emph{atomic} iff $C$ does not contain two disjoint almost-closed subsets.
Finally, call a Markov process \emph{simple atomic} if it has a single almost-closed atomic set such that once started from $\hat{i}$ the process eventually with probability one is trapped in that set. We then have:

\begin{Theorem}{(Corollary of Blackwell's Thm.\,2 on p656) \citep{Blackwell:55}}{t1125}~\\
A Markov process is simple atomic (as above) just when the only bounded solution of the equation $\sum_{j\geq 0}p_{ij}\cdot V(j)=V(i)$, that is Blackwell's Equation (his 6) stating that $V$ is an exact martingale, is constant for all $i$ in $S\backslash C$ and transitions $p_{ij}$.
\end{Theorem}

We now show how to apply \Thm{t1125} to general probabilistic programs to obtain a certificate for non-termination.
\begin{corollary}[Non-termination certificate~]\label{c1026}\label{cor:NT}
	We use the conventions of \Thm{t1651}, restated here.
	Let $\mathit{I}, \mathit{G} \subseteq \Sigma$ be predicates;
	let $\mathit{V}\colon \Sigma{\To}\Rpos$ be a non-negative real-valued function on the state; and
	let $\mathit{Com}$ be a $\pGCL$ program.
	Then the conditions
	\begin{enumerate}[(i)]
		\item\label{i1651-2b}
			$\mathit{I}$ is a standard invariant for the loop\quad$\WHILEDO{\mathit{G}}{\mathit{Com}}$~, and
		\item\label{i1651-1b}
			$\mathit{G}\land\mathit{I}\Implies V{>}0$~, and
		\item\label{i1651-4b}
			$V$ is a non-constant and \underline{bounded} \emph{exact} martingale  on $\mathit{I}\land G$
	\end{enumerate}
	together imply that there is a state $\sigma$ in $I$ such that\quad
	$\wp{\WHILEDO{\mathit{G}}{\mathit{Com}}}{1}(\sigma) <  1$.
	That is
	\begin{quote}
	If a predicate $\mathit{I}$ is a standard invariant,
	and there exists a non-negative real-valued variant function $V$ on the state,
	an exact martingale on $\mathit{I}\land\mathit{G}$,
	such that $V$ is bounded and non-constant,
	then there is some initial state satisfying $I$ from which loop\quad$\WHILEDO{\mathit{G}}{\mathit{Com}}$\quad does not terminate \AS.
	\end{quote}

\begin{proof}
Fix a starting state $\hat{s}$, and collapse the termination set $S_0$ (i.e. all states that do not satisfy the guard) to a single state $s_0$. Now adjust the underlying transition system corresponding to the given program so that any transition to a state in $S_0$ becomes a transition into $s_0$, and assume that there is a single transition from $s_0$ to $s_0$. Suppose now that the probability of $\hat{s}$'s reaching $s_0$ is one. We now note:
\begin{enumerate}
\item Our termination set  $\{s_0\}$ is almost-closed and atomic (in the sense of Blackwell), because
\begin{enumerate}
\item almost closed: Our process reaches $s_0$ with non-zero probability (in fact we assumed with probability one, for a contradiction) and, once at $s_0$, it remains there.
\item atomic: Our set $\{s_0\}$ has no non-empty subsets.
\end{enumerate}
\item We now recall that in fact $s_0$ is reached with probability one, so that the whole process is simple atomic.
\item From Blackwell's \Thm{t1125} we conclude that the only possible non-trivial martingale is unbounded.
\end{enumerate}
We deduce therefore, that if there exists a non-constant bounded martingale then there is some state from which termination is not guaranteed with probability $1$.
\end{proof}
\end{corollary}
Thus --in summary-- we have specialised Blackwell's result to demonstrate a new refutation certificate for programs: if the martingale is finite and non-constant it actually refutes termination with probability 1,  not just finite expected time to termination.

In fact \Cor{cor:NT} provides an interesting embellishment to recent work by \citet{Chatterjee:2017aa} who introduce the notion of ``repulsing super-martingales''. Their \emph{Theorem 6} uses an $\varepsilon$-repulsing 
super-martingale with $\varepsilon{>}0$ to refute almost-sure termination. And their \emph{Theorem 7} uses an $\varepsilon$ repulsing super-martingale with $\varepsilon{\geq}0$ to refute finite expected time to termination. In particular
to refute finite expected time to termination only a martingale is required.  

Our \Cor{cor:NT} takes this further to use non-constant and bounded martingales as certificates to refute almost-sure termination. 
For example the one-dimensional random walker
\begin{align*}
	\WHILEDO{x > 0}{\PCHOICE{\ASSIGN{x}{x - 1}}{\sfrac{1}{2}}{{\ASSIGN{x}{x + 1}}}}
\end{align*}
has an exact \emph{unbounded} martingale, and therefore our rule \Thm{t1928} shows that it terminates with probability $1$. On the other hand the \emph{biased} walker\quad
$\WHILEDO{x > 0}{\PCHOICE{\ASSIGN{x}{x - 1}}{\sfrac{1}{3}}{{\ASSIGN{x}{x + 1}}}}$\quad
(from \Sec{s1643}) has a \emph{non-constant} and \emph{bounded} martingale based on the function $V(s)= 1{-}\pi_{s{\leadsto}0}$ where $\pi_{s{\leadsto}0}$ is the probability that,
starting from state $s$, eventually state $0$ (i.e.\ $x{=}0$) is reached . By \Cor{cor:NT} we can conclude that  the program does not terminate with probability $1$. Note that Chatterjee's Theorem 7 \citeyear{Chatterjee:2017aa} does not distinguish between these two cases in terms of their behaviour: it implies that neither has finite expected time to terminate.
And \Cor{cor:NT} holds even when demonic choice is present.

\subsection{Towards Completeness: The Case of the Random Walker in Two Dimensions}\label{s1544}
\citet{Foster:1952ab} further considers the question of conditions on a Markov process that imply the existence of a super-martingale;  this is relevant for our Theme C.
His conditions are:

\begin{enumerate}
\item The state space $\Sigma$ is countable;
\item There is a finite subset $C \subseteq \Sigma$ that is reached with probability $1$  from any other state;
\item\label{i1543} The states are numbered so that given any pair of states $s_i, s_j$ there is some probability of reaching $s_j$ from $s_i$ whenever $i {<} j$;
\item\label{i1544} There is a single probability $0{<}\delta{<}1$ for the whole system such that for any $N$ there is an $i$ such that for all $j{\geq}i$ the state $s_j$ cannot reach $C$ within $N$ steps and with probability at least $\delta$.
\end{enumerate}
Under these conditions, Foster shows that there exists an \emph{unbounded} super-martingale
 function $V$ on $S$ such that $V(s)$ tends to infinity as the numbering of $s$ tends to infinity.

The construction is a variation on the expected time to termination but, as he remarks, expected time cannot be used because in many situations 
the expected time to termination is infinite.  However using Foster's construction we can prove the existence of a super-martingale that also satisfies the progress  conditions of our rule \Thm{t1651}, and thus could be used  to prove termination for the 2-dimensional symmetric random walk
\begin{align*}
 \WHILEDO{x{\neq}0 \lor y{\neq}0}{\ASSIGN{x}{x{-}1}~\oplus~\ASSIGN{x}{x{+}1}~\oplus~
                                                          \ASSIGN{y}{y{-}1}~\oplus~ \ASSIGN{y}{y{+}1}}
\end{align*}
where iterated $\oplus$ is shorthand for uniform choice (in this case $\NF{1}{4}$ each).

\begin{corollary}[Two-dimensional random walk]\label{thm:ddrw}
There exists a super-martingale  which satisfies the conditions of \Thm{t1651} to prove termination of the two-dimensional random walker.
\begin{proof}
(Sketch.) We follow Foster's argument  \citeyear{Foster:1952ab} to show that there is a numbering of the states that satisfy his conditions for constructing a super-martingale; then we show that the constructed super-martingale also satisfies the progress conditions.
Foster enumerates the states by ``spiralling out'' through increasing Manhattan distance, observing that simple scheme to satisfy his enumeration conditions.
Then he shows that there is a variant function $V$ which satisfies the conditions for a super-martingale;
\footnote{The Manhattan distance itself is not a super-martingale because, on the axes, the distance actually \emph{increases} in expectation by $(-1+1+1+1)/4=\NF{1}{2}$. Indeed if the Manhattan variant worked for two dimensions, it would also work for three; but the 3dSRW is not \AST. }
and in fact as the numbering of $s$ approaches infinity so too does $V(s)$; in particular Foster shows that there are no accumulation points in the image of $V$. Foster's general proof is by construction. (We sketch it in \App{a1551}.)

To show that our rule \Thm{t1651} applies, we need however to establish a progress condition.  First define $p(v)$ to be $\NF{1}{4}$ for all $v$. Then for $d$, first consider the subset $S_{\leq v}$ of $S$ comprising all those $s$ with $V(s){\leq}v$. Because there are no accumulation points in the image of $V$, we must have that $S_{\leq v}$ is finite. Now set $d(v)$ to be the minimum non-zero distance between any two of them, that is $(\textit{min}~(V(s'){-}V(s)) \mid s,s'\in(S_{\leq v})\land V(s'){>}V(s))$.
\Af{See \App{s1743}.}
Since $V(s)$ increases arbitrarily we have that $d$ is non-zero whenever $v{=}V(s)$ for some state with Manhattan distance strictly greater than $0$.

Thus there is guaranteed to be a $V$ satisfying the progress condition \Thm{t1651}(\ref{i1651-3}) that establishes termination for the 2dSRW --- even if we don't know what it is in closed form.
\end{proof}
\end{corollary}

\section{Review of related work on termination for probabilistic programs}\label{s1521}
Our earlier variant rule \Thm{thm:lem-2-7-1} \citep[Sec.\,6]{Morgan:96b},\citep[Sec.\,2.7]{McIver:05a}  effectively made $p,d$ constants, imposed no super-martingale  condition but instead bounded $V$ above, making it not sufficient for the random walk. Later however we did prove the symmetric random walk to be \AST\ using a rule more like the current one \citep[Sec. 3.3]{McIver:05a}.

\textbf{\citet{Chakarov:2013}} consider the use of martingales for the analysis of infinite-state probabilistic programs, and  \citet{Chakarov:2016aa} has done further, more extensive work.

\citeauthor{Chakarov:2013} also show that a \emph{ranking super-martingale} implies \AST, and a key property of their definition for ranking super-martingale is that there is some constant $\varepsilon {>}0$ such that the average decrease of the super-martingale is everywhere (except for the termination states) at least $\varepsilon$. Their program model is operates over discrete distributions, without nondeterminism.

That work is an important step towards applying results from probability theory to the verification of infinite-state probabilistic programs.

\textbf{\citet{Fioriti:2015}} also use ranking super-martingales, with results that provide a significant extension to Chakarov and Sankaranarayanan's work \citep{Chakarov:2013}. Their program model includes both non-determinism and continuous probability distributions over transitions. They also show completeness for the class of programs whose expected time to termination is finite.  That excludes the random walk however; but they do demonstrate by example that the method can still apply to some systems which do not have finite termination time.

We note that  it can be shown that a ranking super-martingale that proves \AS\ also satisfies $p,d$ progress for \Thm{t1928}; see \App{a1339}. 

\textbf{\citet{Chatterjee:2017aa}} study techniques for proving that programs terminate with some probability (not necessarily one). Their innovation is to introduce the concept of ``repulsing super-martingales'' --- these are also super-martingales  with values that decrease outside of some defined set. Repulsing super-martingales can obtain lower bounds on termination probabilities, and as certificates can refute almost-sure termination and finite expected times to termination. 

More recently still \textbf{\citet{Chatterjee:2017ab}} have studied termination for probabilistic and nondeterministic recursive programs. In particular they show that ``conditionally difference-bounded ranking super-martingales'' can be used to prove almost-sure termination. As we do, Chatterjee and Fu allow super-martingales (i.e.\ not necessarily ranking); and their Thm.\ 5 requires that the average absolute difference between $V(\sigma)$ and $V(\sigma')$ must be at least some fixed $\delta{>}0$.  
This constraint seems to imply some kind of progress
%
%
and it will be an interesting exercise to understand exactly the differences in applicability between the two rules.  For example the existence of a fixed $\delta{>}0$ allows Chatterjee and Fu to give an estimate for ``tail probabilities''.

{\Fx
On the other hand the variation of the random walker given by the ``Lazy Loper'' program of \Sec{s1117}, in which the walker ``dawdles'' at a location depending on the distance to the origin, nevertheless can be proved to terminate almost surely using our \Thm{t1651} with definitions $V(x)=x$, and $p(v)= 1 \min 1/2v$ and $d(v) = 1$ for progress; but Chatterjee's Thm. 5 \citeyear{Chatterjee:2017ab} does not seem to apply here. Moreover there appears to be no super-martingale for this program that has average absolute move bounded away from $0$, as we explained in \Sec{s1117}.
}

Finally, \textbf{\citet{Agrawal:2018}} have extended the $\varepsilon$-strict super-martingale approach to include lexicographic orderings, and present techniques for their automatic synthesis. (We explore parametrised-$\varepsilon$ super-martingales, but not lexicographic, in \citet[Sec. 5]{McIver:2016aa}.)

A different approach to the same issue is the work of \textbf{\citet{Lago:2017aa}} in which expressions themselves are probabilistic artefacts, and their termination properties can be ``inherited'' by functional programs containing them: that allows the expressions' behaviour to be studied separately, outside of the clutter of the program containing them.

There are a number of other works that demonstrate tool support based on the above and similar techniques.  All the authors above \citep{Chakarov:2013,Fioriti:2015,Chatterjee:2017aa} have developed and implemented algorithms to support verification based on super-martingales.  \textbf{\citet{Esparza:2012}} develop algorithmic support for \AST\ of ``weakly finite'' programs, where a program is \emph{weakly finite} if the set of states reachable from any initial state is finite.  \textbf{\citet{Kaminski:16}} have studied the analysis of expected termination times of infinite state systems using probabilistic invariant-style reasoning, with some applications to \AST. In even earlier work \textbf{\citet{Celiku:05}} explore the mechanisation of upper bounds on expected termination times, taking probabilistic weakest pre-expectations \citep{McIver:05a} for their model of probability and non-determinism.

\section{Theoretical issues, limitations and caveats}\label{s1648}
\subsection{How Much Nondeterminism?}\label{s1130}
Our arguments above are over ``expectation transformers'', i.e.\ functions from post-expectations to pre-expectations and thus going in effect ``backwards''. But equivalently our programs are functions from initial state to (discrete) distribution over final states or, when demonic choice is present, to \emph{sets} of such distributions (but only sets satisfying certain ``healthiness'' conditions). That equivalence was shown by \citet{Kozen:85} for deterministic (i.e.\ non-demonic) programs, and extended by \citet{Morgan:96d,McIver:05a} when demonic choice was added. \autoref{table:wp} interprets programs (syntax) into that semantic space, and e.g.\ \Thm{thm:lem-2-7-1} and \Thm{thm:lem-2-4-1}, crucial to our argument, have been shown to be true in that space \citep{McIver:05a}.

Important is that those two theorems were \emph{not} proved by structural induction over $\PGCL$ syntax directly; rather they follow from a different structurally inductive proof, that all $\PGCL$ programs are mapped into the semantic space (where the theorems hold) --- that is, a proof that the space is closed under program-combining operators. The significance of the difference is that our results therefore hold for any elements of that space, whether they come from $\PGCL$ or not, including operational descriptions of programs as transition systems provided they satisfy the healthiness conditions the space demands. One such condition is the restriction to discrete distributions.
\footnote{Thus e.g.\ part of the structurally inductive proof would be to show that loops with discrete-distribution bodies cannot somehow ``in the limit'' require a proper measure to define their overall effect: the worst it can get is a countably infinite but still discrete distribution.}

Another healthiness condition concerns the degree of demonic choice our semantic space allows: is it finite?\ countable?\ unlimited? In fact our space requires that the sets of distributions be closed in the product topology over the set of discrete (sub-)distributions on $\Sigma$, that is distributions whose total weight is \emph{no more than} 1. (Any missing weight indicates non-termination.) All (meanings) of $\PGCL$ programs have that property \citep{McIver:05a}; and all \emph{finitely} branching transition systems do. But that property is not as simple e.g.\ as countable vs.\ uncountable branching. For example, Program
\begin{align}\label{e1739}
	& \COMPOSE{\ASSIGN{c,x}{\True,0}}{\quad\COMPOSE{\WHILEDO{c}{\PCHOICE{\ASSIGN{c}{\False}}{\NF{1}{2}}{\ASSIGN{x}{x + 2}}}}{\quad\NDCHOICE{\ASSIGN{x}{x+1}}{\SKIP}}}\,,
\end{align}
\F{if expressed as a transition system with one large demonic branch followed by a probabilistic branch at each tip,}
makes \emph{uncountably} many demonic choices (over geometric-style discrete distributions).
\footnote{
First pick any real number $b$ in the unit interval $[0,1]$ (which action cannot be written using $\PGCL$'s only-binary demonic choice); consider its binary expansion $0.b_1b_2\cdots b_n\cdots$. Construct the discrete (countably infinite) distribution\quad
\(
	0{+}b_1\;\AT \NF{1}{2},\quad
	2{+}b_2\;\AT \NF{1}{4},\quad
	\cdots\quad
	2n{+}b_n\;\AT \NF{1}{2^n}\quad
	\cdots
\)
\quad where ``$\AT$'' means ``with probability''.  (That second step can be done using $\PGCL$, for already-determined $b$.) For every $b$ chosen in the first step, the above distribution is a possible result, different for each $b$ and so uncountably numerous. But still the set of them all is closed, since the $\PGCL$ \Eqn{e1739} produces it.
}~
Nevertheless, because the program is written in $\PGCL$, that set is closed. On the other hand, the (standard) program\quad``choose $n$ from the natural numbers''\quad has only \emph{countably infinite} branching, and yet cannot be written in the $\PGCL$ of \autoref{table:wp}. Embedded in the probabilistic model \citep{McIver:05a}, its output set of distributions is not closed --- and so this program is out-of-scope for us. But Program \Eqn{e1739} is within our scope.

Thus the conceptual boundary of our result is \emph{not} countable vs.\ uncountable branching: rather it is topological closure vs.\ non-closure of sets of discrete distributions. But this issue is important only for examples ``imported'' from outside of $\PGCL$; for any $\PGCL$ program, closure of the corresponding transition system's results sets is automatic \citep[Sec.\ 8.2]{McIver:05a}.

A second example of an uncountable-but-closed set of distributions is given in \App{a1713}.

\subsection{``Progress'' is More Demanding than it Looks}\label{s1643}
Consider the asymmetric random walker\quad $\COMPOSE{\ASSIGN{x}{1}}{\WHILEDO{x{\neq}0}{\PCHOICE{\ASSIGN{x}{x{-}1}}{\NF{1}{3}}{\ASSIGN{x}{x{+}1}}}}$\,. We can easily synthesise an exact- (and thus super-) martingale $V(x)=\NF{2^x-1}{2^{x-1}}$ by solving the associated recurrence. It is bounded asymptotically above by 2, so that for progress we are tempted by $p(v){=}\NF{1}{3}$ and $d(v)=2{-}v$, both satisfying our positive-and-antitone requirements when $v{<}2$.

But this $d()$ in fact does not satisfy our requirements, because they apply for \emph{all}\/ $v$, not just those generated by states that the program can actually reach. And in this case there is no suitable value for $d(2)$, since it would have to be 0 for $d$ to be antitone. That is, even though the program can never reach a state $x$ where $V(x){=}2$, the requirements on $d(2)$ still apply.

As well as saving us from unsoundness (since the that asymmetric walker is not \AST\/), this exposes an important methodological issue: the properties of $p,d$, their being non-zero and antitone, \emph{do not refer to the program text at all}. However the properties of those functions might be proved --by hand, or with Mathematica or Sage-- the semantics of $\PGCL$ is not required: one needs only analytic arguments over the reals. And those arguments can be delegated to other people who have never heard of $\PGCL$ or transition systems, or Markov processes, random variables or program termination. That is, if we want to use powerful external analytical tools, we should avoid as far as possible  that they must be ``taught'' our semantics.

\subsection{Why Do we Express $\boldsymbol{V}$'s Being a Super-Martingale by Writing a Sub-Martingale Inequality?}
\label{s1127}
In \Thm{t1651} we wrote the super-martingale property of $V$ as a sub-martingale property of $H{\ominus}V$; yet in \Sec{s0918}, the case studies, we introduce the ``angelic'' $\awpsymbol$ and check the super-martingale property directly. Why didn't we use $\awpsymbol$ in \Thm{t1651} in the first place?

The reason is that \Thm{thm:lem-2-4-1} is proved over the semantic space of \citet{McIver:05a} mentioned in \Sec{s1130} above, and the brief treatment of angelic choice there [\textit{op. cit.}, Sec.~8.5] gives no $\awpsymbol$-based results for loops. To refer to the literature in its own terms --and to avoid building new special-purpose semantics here-- we therefore must use only $\wpsymbol$ when importing existing results.

On the other hand, the equivalence introduced for convenience in \Sec{s0918} --and whose property \Eqn{e1134} is established by structural induction over straight-line programs-- is used for \Eqn{e0908} only and does not rely on closure, or any other sophisticated property of the semantic space. 


\subsection{Bounded Expectations}\label{s1413}
In the symmetric random walk on naturals $x$, the expectation $x$ is an exact martingale in fact; and that process terminates \AS. If however we had used \underline{unbounded} $x$ as $\mathit{Sub}$ in \Thm{thm:lem-2-4-1}, we could conclude that the expected final value of $x$ is at least the (exact) initial value of $x$. If the process started at $x{=}1$, therefore, we would conclude that its expected value on termination is at least 1; but we know that its $x$'s expected (in fact exact) value on termination is 0 --- a contradiction.

That is why one assumption of \Thm{thm:lem-2-4-1} is that $\mathit{Sub}$ is bounded, and is one reason that, instead of using the potentially unbounded $V$, we use the bounded $H{\ominus}V$ instead.
(See also \App{a0904}.)


\section{Conclusion}

We have investigated ``parametric'' super-martingale methods for proving almost-sure termination for probabilistic- and demonic programs, and our main result \Thm{t1651} presents a new method, described earlier by \citet{McIver:2016aa} over a transition system, but now expressed and proved in the probabilistic programming logic of $\PGCL$; the rule can therefore be applied at the source level.  Although our presentation is in terms of $\wpsymbol$-style reasoning,  our innovation of parametrised $p, d$ progress is also applicable to transition-style models of programs. (See, for example \citeauthor{GretzKM14}'s interpretation  \citeyear{GretzKM14} of $\wpsymbol$ in terms of explicit transition systems.)

Our rule seems to be able to prove some tricky cases that go beyond other published rules, and moreover we have shown that $p, d$ progress can also be used as alternatives to rules based on ranking super-martingales, and rules based on conditional absolute difference. Furthermore, we believe our rule suffices for the two-dimensional symmetric random walk (\Sec{s1544}).

Completeness remains an open problem however, although the mathematical literature provides some insight to its solution in certain cases \citep{Blackwell:55, Foster:1952ab}.



\appendix

\section*{APPENDICES}
\bigskip

\section{In-the-limit termination implies termination\hfill [from \Sec{s1656}]}\label{a1139}
The following lemma is used in Part \Itm{i1656-3} of the proof of \Thm{t1651} in \Sec{s1656}. (Its proof is structurally identical to the analogous proof for non-probabilistic programs.)

\begin{lemma}\label{l0814}
Let $A,B$ be any two predicates on the state. Then
\begin{align}\label{e1057B}
	\wp{\WHILEDO{A\land B}{\mathit{Com}}}{\iverson{\neg A}}
	\Wide{\leq}\wp{\WHILEDO{A}{\mathit{Com}}}{\iverson{\neg A}} \quad. \footnotemark
\end{align}
\end{lemma}
\footnotetext{It might suprise at first that a stronger loop-guard could induce less- rather than more termination: so let $A$ be ``in the desert'' and $B$ be ``still have water'' and $\mathit{Com}$ be ``crawl''. In Loop $f$ we keep crawling only while we still have water; in Loop $g$ we crawl even without water. Success is ``leaving the desert'' --- termination by ``dying of thirst'' is failure, because the postcondition is not satisfied.\par In standard $\GCL$ the same inequality holds \emph{with the same proof}, mutatis mutandis, once the $\leq$ is replaced by $\Implies$. In that case it says that whenever the $A{\land}B$-loop is guaranteed to establish $\neg A$, so is the $A$-loop. The $\PGCL$ version simply converts ``implies'' into an ``is no more likely to''.}
\begin{proof}
We use the general rule for fixed points that $F(g){\leq}g\Implies \mu F{\leq}g$. In this case $f$ is \Eqn{e1057B}'s \LHS\ and $F$ its defining functional, with $g,G$ for \Eqn{e1057B}'s \RHS; and we are showing that $f{\leq}g$.
To apply the general fixed-point rule, we must therefore establish
\begin{align}
	& \makebox[0pt][r]{$F(g) \rightarrow$\hspace{6em}}
	\wp{(\IF{A\land B}(\COMPOSE{\mathit{Com}}{\WHILEDO{A}{\mathit{Com}})})}{\iverson{\neg A}} \label{e1127B-1} \\
	\leq\quad & \makebox[0pt][r]{$g \rightarrow$\hspace{6em}}
	 \wp{\WHILEDO{A}{\mathit{Com}}}{\iverson{\neg A}} \label{e1127B-2} \\
	=\quad & \makebox[0pt][r]{$G(g) \rightarrow$\hspace{6em}}
		\wp{(\IF{A}(\COMPOSE{\mathit{Com}}{\WHILEDO{A}{\mathit{Com}})})}{\iverson{\neg A}}~, \label{e1127B-3}
\end{align}
where we have used $g{=}G(g)$ to rewrite \Eqn{e1127B-2} to \Eqn{e1127B-3}, whose comparison with \Eqn{e1127B-1} is easier.

Now if $B$ holds in the (initial) state to which these $\wpsymbol$'s are applied, then \Eqn{e1127B-1} and \Eqn{e1127B-3} are equal, since the\quad$\IF{A\land B}$\quad of the former reduces to the\quad$\IF{A}$\quad of the latter, and they are otherwise textually identical. But  otherwise, i.e.\ initially $\neg B$, we have that \Eqn{e1127B-1} is exactly $\iverson{\neg A}$, and we know that \Eqn{e1127B-3} is at least $\iverson{\neg A}$ from any state.
\end{proof}

In \Sec{s1656} we were considering a sequence $B_n$ of predicates for which we had established that
\begin{equation}\label{e1113a}
 \wp{\WHILEDO{A\land B_n}{\mathit{Com}}}{\iverson{\neg A}}\Wide{\geq}t_n
\end{equation}
for a corresponding sequence of $t_n$'s. From \Eqn{e1113a} and \Lem{l0814} we can conclude that
\begin{equation}\label{e0826}
 \wp{\WHILEDO{A}{\mathit{Com}}}{\iverson{\neg A}}\Wide{\geq}\mbox{$\sup_n t_n$} ~,
\end{equation}
where in \Sec{s1656} in fact $\sup_n t_n$ was one. That gives us \AST\ for the \LHS\ of \Eqn{e0826}.

\newpage
\section{On super-martingales and the proof of \Thm{t1928} \hfill[from \Sec{s0946}]}\label{a1216}

The formal statement and proof of \Thm{t1928} does not refer to (nor depend on) a super-martingale property directly. Instead, it uses an $H{\ominus}V$ construction, for two reasons (as briefly described earlier):
\begin{enumerate}
\item The instantiation of $\mathit{Sub}$ in \Thm{thm:lem-2-4-1} must be bounded, and $H{\ominus}V$ is bounded (even if $V$ is not); and
\item\label{i1316} When demonic choice is present, leading to a \emph{set} of final distributions rather than only one, the $\PGCL$ logic allows us only to bound expected values below --- not above, as a super-martingale must be. This is because $\PGCL$ interprets demonic choice as $\min$. So by bounding $H{\ominus}V$ below, we bound $V$ itself above.
\end{enumerate}
In this section we give more background for \Itm{i1316}. In order to make the argument self-contained, however, we reason here over transitions directly. We stress however that our contribution \Thm{t1928} stands or falls in its $H{\ominus}V$ form: the lemma below is not necessary for its validity.

\bigskip
We consider a single transition $\sigma\mapsto\Delta$, for $\sigma$ a state in $\Sigma$ and $\Delta$ a (discrete) distribution over $\Sigma$. For function $f$ in $\Sigma{\To}\Rgen$, possibly negative valued, we write $\Expec{\Delta}{f}$ for the expected value of $f$ on $\Delta$.

\begin{lemma}[super-martingales vs.\ $(H\ominus)$]\label{l1331}
For any $f$ in $\Sigma{\To}\Rgen$ and transition $\sigma\mapsto\Delta$ with for $\sigma$ a state in $\Sigma$ and $\Delta$ a (discrete) distribution over $\Sigma$, we have
\[
 	f(\sigma)\geq \Expec{\Delta}{f}
	\WideRm{iff}
	H{\ominus}f(\sigma) \leq \Expec{\Delta}(\lambda\sigma'\kern-.3em.\,H{\ominus}f(\sigma')) \quad \textrm{for all real $H>0$.}
\]
\end{lemma}
\begin{proof} (Forwards, \emph{only if})
We reason for any $H$ that
\begin{Reason}
\Step{}{f(\sigma)\geq \Expec{\Delta}{f}}
\Step{iff}{H-f(\sigma)~\leq~H-\Expec{\Delta}{f}}
\StepR{iff}{distribute linear operation through expected value}{
	H-f(\sigma)~\leq~\Expec{\Delta}{(\lambda\sigma'\kern-.3em.\,H{-}f(\sigma'))}
}
\StepR{implies}{monotonicity of $\Expec{\Delta}{}$; and $H{-}f(\sigma')\leq H{\ominus}f(\sigma')$}{
	H-f(\sigma)~\leq~\Expec{\Delta}{(\lambda\sigma'\kern-.3em.\,H{\ominus}f(\sigma'))}
}
\StepR{implies}{\RHS\ is non-negative}{
	H\ominus f(\sigma)~\leq~\Expec{\Delta}{(\lambda\sigma'\kern-.3em.\,H{\ominus}f(\sigma'))}~,
}
\end{Reason}
as required.
\end{proof}
\begin{proof} (Backwards, \emph{if}) We reason
\begin{Reason}
\Step{}{
	H{\ominus}f(\sigma) \leq \Expec{\Delta}(\lambda\sigma'\kern-.3em.\,H{\ominus}f(\sigma'))
}
\StepR{iff}{arithmetic}{
	H-(H\min f(\sigma)) \leq \Expec{\Delta}(\lambda\sigma'\kern-.3em.\,H-(H\min f(\sigma')))
}
\StepR{iff}{distribute $(H-)$ through $\Expec{\Delta}{}$}{
	H\min f(\sigma) \geq \Expec{\Delta}(\lambda\sigma'\kern-.3em.\,H\min f(\sigma'))
}
\StepR{implies}{$\min$}{
	f(\sigma) \geq \Expec{\Delta}(\lambda\sigma'\kern-.3em.\,H\min f(\sigma'))~,
}
\end{Reason}
which implies incidentally that $\Expec{\Delta}(\lambda\sigma'\kern-.3em.\,H\min f(\sigma'))$ does not diverge (i.e.\ is finite).

Using the above, our desired $f(\sigma)\geq\Expec{\Delta}f$ would follow from $\sup_H \Expec{\Delta}(\lambda\sigma'\kern-.3em.\,H\min f(\sigma')) = \Expec{\Delta}{f}$, which we now prove. Note that because $\Delta$ is discrete its support is countable. We have
\begin{Reason}
\Step{}{
	\sup_H \Expec{\Delta}(\lambda\sigma'\kern-.3em.\,H\min f(\sigma'))
}
\StepR{$=$}{$\Delta$ has countable support; $\Delta_{\downarrow\Sigma'}$ restricts $\Delta$ to $\Sigma'$}{
	\sup_H \sup_{\substack{\Sigma'\subseteq\Sigma\\\Sigma'\textrm{~finite}}} \Expec{\Delta_{\downarrow\Sigma'}}(\lambda\sigma'\kern-.3em.\,H\min f(\sigma'))
}
\StepR{$=$}{commute $\sup$'s}{
	\sup_{\substack{\Sigma'\subseteq\Sigma\\\Sigma'\textrm{~finite}}} \sup_H \Expec{\Delta_{\downarrow\Sigma'}}(\lambda\sigma'\kern-.3em.\,H\min f(\sigma'))
}
\StepR{$=$}{$\Delta_{\downarrow\Sigma'}$ has finite support}{
	\sup_{\substack{\Sigma'\subseteq\Sigma\\\Sigma'\textrm{~finite}}} \Expec{\Delta_{\downarrow\Sigma'}}f
}
\StepR{$=$}{as above}{
	\Expec{\Delta}f ~.
}
\end{Reason}
%
%
%
\end{proof}

\newpage
\section{On the constraints imposed by \Thm{T0909} --- boundedness and \AST\hfill[\Sec{s1413}]}\label{a0904}
Theorem~\ref{T0909} requires that the sub-martingale $\mathit{Sub}$ be bounded, and in \Sec{s1413} a counter-example shows that to be necessary.

A less prominent constraint imposed by the theorem however is that it applies only from initial states where termination is \AS\ --- even though informal, operational reasoning might suggest a weaker requirement as here:

\medskip\begin{quote}
\begin{minipage}{.8\linewidth}
If a loop body is guaranteed never to decrease the expected value of a bounded random variable, i.e.\ has the sub-martingale property, then that random variable's (conditional) expected value on termination (if it occurs) is no less than the value it had in the initial state.
\footnotemark
\end{minipage}
\hfill\emph{(Is not true.)}
\end{quote}
\footnotetext{Here for comparison is the actual (informal) requirement:
\begin{quote}
\medskip\begin{minipage}{.85\linewidth}
If a loop body is guaranteed never to decrease the expected value of a bounded random variable, and the loop terminates \AS\ from a given initial state, then that random variable's expected value on termination is no less than the value it had in that initial state.
\end{minipage}\hfill\emph{(Is true.)}
\end{quote}}

\medskip
But here is a counter-example to that weaker requirement:
\begin{align*}
	& \ASSIGN{x}{1} \\	
	& \WHILE{x \neq 0} \\
	& \qquad \IFELSE{x{=}1}{\PCHOICE{\ASSIGN{x}{0}}{\NF{1}{2}}{\ASSIGN{x}{2}}}{\ASSIGN{x}{2}} \\
	& \}\quad .
\end{align*}
We take variant $x$, and note the loop invariant $x{\in}\{0,1,2\}$ so that $x$ is bounded. The loop terminates  from initial state $x{=}1$ with probability only $\NF{1}{2}$ --- thus it \emph{might} terminate from there, but its termination is not \AS.

Both branches of the conditional establish a final expected value of $x$ that is no less than (in fact is equal to) its initial value for the conditional --the invariant and guard ensure the input value is either 1 or 2-- so that the loop body's sub-martingale property is satisfied. Yet, for the whole loop, the initial value of $x$ is 1 and the (conditional) expected value on termination (if it occurs) is 0.

\newpage
\section{Remarks on the utility of $\PGCL$}
\subsection{Compositionality\hfill[from \Sec{s0940}]}\label{a1221}
It is reasonable to ask why statements like \Eqn{e0804} in \Sec{s0940}, that is
\[
	p\cdot\iverson{A}~\leq~\wp{\mathit{Com}}{\iverson{B}}~,
\]
could not be written more directly and intuitively
\begin{equation}\label{e0628}
 p\vdash\{A\}\mathit{Com}\{B\}~,
\end{equation} 
meaning by analogy with Hoare logic ``the Hoare triple $\{A\}\mathit{Com}\{B\}$ holds with probability $p$.'' Why bother with all the machinery of expectation transformers?

As shown by \citet[App A]{McIver:05a}, the reason is that the approach of \Eqn{e0628} is not compositional if $\mathit{Com}$ contains demonic nondeterminism. The more general expectation-transformer generalisation we use here \emph{is} compositional for probabilistic- and demonic choice together. 

\subsection{Linearity (or not) of Expectation Transformers \hfill[from \Sec{s0940}]}\label{a0950}
At \Eqn{e1752} in \Sec{s0940} we stated that expectation transformers are scaling --that is they distribute multiplication by a scalar-- and that that property was the analogue of multiplication's distributing through expected value in elementary probability theory.

Note however that the elementary property of distribution of \emph{addition} through expected value does not translate directly into a probabilistic $\wpsymbol$ rule if demonic choice is present. This additivity failure is the analogue of disjunctivity's failure for standard demonic programs: for program $\mathit{Com} = \NDCHOICE{\ASSIGN{x}{\True}}{\ASSIGN{x}{\False}}$ both $\wp{\mathit{Com}}{x}$ and $\wp{\mathit{Com}}{\neg x}$ are $\False$; but $\wp{\mathit{Com}}{(x\lor\neg x)}$ is $\True$, i.e.\ not equal to $\False\lor\False$.
\par A probabilistic version would be
\begin{quote}
Consider program $\mathit{Com} = \NDCHOICE{\PCHOICE{\ASSIGN{x}{\True}}{\NF{1}{3}}{\ASSIGN{x}{\False}}}{\PCHOICE{\ASSIGN{x}{\False}}{\NF{1}{3}}{\ASSIGN{x}{\True}}}$.\\
Both $\wp{\mathit{Com}}{\iverson{x}}$ and $\wp{\mathit{Com}}{\iverson{\neg x}}$ are $\NF{1}{3}$; but $\wp{\mathit{Com}}{(\iverson{x}{+}\iverson{\neg x})}$ is $1$, i.e.\ not equal to $\NF{1}{3}{+}\NF{1}{3}$.
\end{quote}

\subsection{Semantic vs.\ Syntactic Arguments\hfill[from \Sec{s0906}]}\label{s1031}
One way of comparing our new rule with others (including our own earlier 
\Thm{thm:lem-2-7-1}) is simply to say that whereas others often require progress to be bounded away from zero (when the state-space is infinite), we require only that progress be non-zero but \emph{provided that}, if the state-space is indeed infinite, the variant have no accumulation points.

It is the antitone restriction on $p,d$, and their interpretation via the program logic as progress conditions, that allows the proof to be carried out on the program text directly: i.e.\ it helps to avoid having to prove non-accumulation by a separate semantic-based argument that e.g.\ would begin by determining the reachable states and then continue with a mathematical analysis outside the program text.

\newpage
\section{Sketch of Foster's construction\hfill[from \Sec{s1544}]}
\label{a1551}
We sketch the proof of Foster's construction \citeyear{Foster:1952ab} for the existence of an unbounded super-martingale in the case that the transition system satisfies the conditions set out in \Sec{s1544}. This historical work supports our contention in \Cor{thm:ddrw} that \Thm{t1651} will work for the two-dimensional random walk.

We use the notation and definitions from \Sec{s1530} to present Foster's Theorem 2, but adapt the notation to Foster's for ease of checking his proof steps.

Recall that we have assumed that $S_0{=}\{s_0\}$, i.e.\ that termination occurs in a single state, and that we have adjusted (the assumed deterministic) transition system so that it takes $s_0$ to itself.

Write $f^{(t)}_i$ for the probability that $T$ started from $s_i$ reaches $s_0$ for the first time in the $t$-th step and (as Foster does) write $p_{ij}$ for  the probability of transitioning from $s_i$ to $s_j$; more generally write $p^{(t)}_{ij}$ for the probability that it takes $t$ steps to do that. Foster remarks that a simple special case is where time-to-termination is bounded, but notes that such an assumption excludes the symmetric random walk and moves immediately to the more general case.
\footnote{Also \citet{Fioriti:2015} treat the bounded-termination case explicitly.}

For the more general case we note first that for $i{>}0$ we have $f^{(t+1)}_i=\sum_j p_{ij}\cdot f^{(t)}_j$.
So if we were hopefully to proceed simply by setting $V(s_0){=}0$ and $V(s_i) = \sum_{1\leq t} f^{(t)}_i$ for $i{>}0$, then in the latter case we would check the super-martingale property  by calculating
\begin{Reason}
\Step{}{\sum_j p_{ij}\cdot V(s_j)}
\Step{$=$}{\sum_j p_{ij}\cdot \sum_{1\leq t} f^{(t)}_j}
\Step{$=$}{\sum_{1\leq t}\sum_j p_{ij} \cdot f^{(t)}_j}
\StepR{$=$}{above and $i{>}0$}{\sum_{1\leq t}f^{(t+1)}_i}
\StepR{$\leq$}{(actually equal unless $f^{(1)}_i{>}0$)}{
 \sum_{1\leq t}f^{(t)}_i ,
}
\Step{$=$}{V(s_i)~,}
\end{Reason}
so that $V$ would in fact be an exact martingale.
\footnote{Think of the symmetric random walk, where everywhere-1 is an exact martingale except when $|x|{=}1$, where it is a proper super-martingale.}
But this looks too good to be true, and indeed it is: in fact $ \sum_{1\leq t} f^{(t)}_i = 1$ by assumption, so this is just the special case where $V$ is 1 everywhere except at $s_0$; and the martingale property is exact everywhere, except at states one step away from $s_0$. And this trivial $V$ does not satisfy the progress condition.
\footnote{It is trivial in Blackwell's sense \citeyear{Blackwell:55}, a constant solution.}

Still, the above is the seed of a good idea. Using ``a theorem of Dini'' \cite[Foster's citation (4)]{Knopp:1928aa},
\footnote{There seems to be a typographical error here in Foster's paper, where he writes $\sum_{r=1}^\infty\lambda^{(r)}f^{(r)}_i$ instead of $\sum_{r=1}^\infty\lambda^{(r)}f^{(r)}_1$.}
that
\begin{quote}\it
If $c_n$ is a sequence of positive terms with $\sum_n{c_n}<\infty$, then also
\[
 \sum_n \frac{c_n}{(c_n{+}c_{n+1}+\cdots)^\alpha} \Wide{<} \infty
\]
when $\alpha{<}1$,
\end{quote}
Foster \emph{increases} the $f^{(t)}_i$ terms above by dividing them by {\scriptsize$\sqrt{f^{(t)}_1+f^{(t+1)}_1+\cdots}$}\,, which is non-zero but no more than one,
\footnote{It is the square-root of the probability that $s_1$ does not reach $s_0$ in fewer than $t$ steps.}
and still (as we will see) the new, larger terms still have a finite sum. (A minor detail is that he must show that the sum $f^{(t)}_1+f^{(t+1)}_1+\cdots$ does not become zero at some large $t$ and make terms from then on infinite: his assumption (F7) prevents that by ensuring that from no state does a single transition step go entirely into $S_0$.) With the revised $V$ replacing the earlier ``hopeful'' definition, the calculation above becomes instead
\begin{Reason}
\Step{}{\sum_j p_{ij}\cdot V(s_j)}
\StepR{$=$}{revised definition of $V$, \\ and $V(s_0){=}0$}{
 \sum_{j\geq1} p_{ij}\cdot \sum_{1\leq t}\NF{f^{(t)}_j}{\sqrt{f^{(t)}_1+f^{(t+1)}_1+\cdots}}
}
\Step{$=$}{
 \sum_{1\leq t}\sum_j p_{ij}\cdot  \NF{f^{(t)}_j}{\sqrt{f^{(t)}_1+f^{(t+1)}_1+\cdots}}
}
\Step{$=$}{
 \sum_{1\leq t}\NF{f^{(t+1)}_j}{\sqrt{f^{(t)}_1+f^{(t+`)}_1+\cdots}}
}
\StepR{$=$}{denominator is not increased}{
 \sum_{1\leq t}\NF{f^{(t+1)}_j}{\sqrt{f^{(t+1)}_1+f^{(t+2)}_1+\cdots}}
}
\Step{$\leq$}{
 \sum_{1\leq t}\NF{f^{(t)}_j}{\sqrt{f^{(t)}_1+f^{(t+1)}_1+\cdots}}
}
\Step{$=$}{V(s_i)~.}
\end{Reason}
This is encouraging: but we still must prove (F3) for our revised definition
\footnote{Note the $f$'s in the denominator are subscripted ``1'', not ``$i$''.}
\begin{Equation}\label{e0838}
 V(s_i)\Wide{=}\sum_{1\leq t}\frac{f^{(t)}_i}{\sqrt{f^{(t)}_1+f^{(t+1)}_1+\cdots}}~,
 \mbox{\quad for $i{\geq}1$}
\end{Equation}%
i.e.\ that it's finite for all $i$ and not only for the $i{=}1$ that Dini gave us; and we must show that it approaches infinity as $i$ does.

For the first, Foster proves that $V(s_i){\leq}V(s_1)/p^{(t')}_{1i}$ for any $i{>}1$ and some $t'{>}0$ with $p^{(t')}_{1i}{>}0$, which is one place he uses \Sec{s1544}(\ref{i1543}), in particular that every $s_i$ is accessible from $s_1$.
\par Specifically, he reasons as follows:
\begin{enumerate}
\item For that $t'$ and any $t$ we have $f^{(t'+t)}_1\geq p^{(t')}_{1i}f^{(t)}_i$, because we know that $s_1$'s journey to $s_0$ can go via $s_i$. 
\item The numerator $f_i^{(t)}$ in \Eqn{e0838} can therefore be replaced by $f^{(t'+t)}_1/p^{(t')}_{1i}$ provided $(\leq)$ replaces the equality.
\item The sum in the denominator of \Eqn{e0838} can be adjusted to start at $t'{+}t$ rather than $t$, still preserving the inequality.
\item The overall sum in \Eqn{e0838} of non-negative terms for $V(s_i)$ is now the ``drop the first $t'$ terms suffix'' of that same sum for $V(s_1)$, which we already know to be finite (from Dini), but divided by $p^{(t')}_{1i}$ which we know to be non-zero.
\end{enumerate}

For the second, Foster uses the $\delta$ from \Sec{s1544}(\ref{i1544}), showing that $V(s_i)$ is at least $\NF{(1{-}\delta)}{\sqrt{f^{(t_i)}_1+f^{(t_i+1)}_1+\cdots}}$ where $t_i$ is the number of steps after which $s_i$ reaches $s_0$ with probability at least $\delta$ for the first time. By \Sec{s1544}(\ref{i1544}) that $t_i$ approaches infinity as $i$ does, and thus so does $V(s_i)$.
\par His detailed reasoning is as follows:
\begin{enumerate}
\item Since $t_i$'s tending to infinity is all that is required, any at-most-finite number of $i$'s where $t_i{=}0$ can be ignored. Thus pick $t_i{\geq}1$.
\Cf{I'm not sure why $t_i{\geq}1$ helps, though.}
\item Then $V(s_i)$ is at least $\sum_{1{+}t_i\leq t}\NF{f^{(t)}_i}{\sqrt{f^{(t)}_1+f^{(t+1)}_1+\cdots}}$\,, a suffix of its defining series \Eqn{e0838}.
\item Since the denominators only decrease, we can replace all of the denominators by {\scriptsize${\sqrt{f^{(t_i)}_1+f^{(t_i+1)}_1+\cdots}}$}\, while making the sum only smaller.
\item From (F8) however and the choice of $t_i$ we know that $\sum_{t_i\leq t}f_0^{(t)}$ is no more than $1{-}\delta$. Thus similarly we can replace $f^{(t)}_i$ by $1{-}\delta$ and remove the summation.
\item We are left with $V(s_i)\geq \NF{(1{-}\delta)}{\sqrt{f^{(t_i)}_1+f^{(t_i+1)}_1+\cdots}}$\,, as appealed to above.
\end{enumerate}

\bigskip That completes the proof sketch.

The symmetric two-dimensional random walk satisfies Foster's conditions, and so there is a variant in the style of our \Thm{t1651} --- indeed he constructs it in general terms at \Eqn{e0838}. But it is not in closed form: it depends on the probabilities $f_i^{(t)}$ that surely exist, even though we do not know what they are.

\newpage
\section{Some properties of \Thm{t1928}\hfill[from \Sec{s1521}]} \label{a1339} 

In earlier work \citet{Chakarov:2013}  and \citet{Fioriti:2015} use ``ranking'' super-martingales to prove almost-sure termination.  

\begin{definition}[Ranking super-martingale]\label{s1448}
Expectation $V$ in $\E$, with $V{<}\infty$, is a \emph{ranking} super-martingale for\quad$\WHILEDO{G}{\mathit{Com}}$\quad if it is a super-martingale with the extra condition that there is some $\epsilon{>}0$ such that
\footnote{As in the Case Studies \Sec{s0918}, we use $\awpsymbol$ here.}
\begin{Equation}\label{e1533-1}
V{-} \epsilon \Wide{\geq} \iverson{G\land I}\cdot\awp{\mathit{Com}}{V}~.
\end{Equation}
\end{definition}
We now show that any program that  has a ranking super-martingale for some $\epsilon$ also can be proved with our \Thm{t1928}, because the ranking property of the super-martingale guarantees the existence of $p,d$ that satisfy our progress condition.

\begin{lemma}[Ranking super-martingale and progress]\label{t1525}
Let $V$ in $\E$ be a ranking super-martingale for program \quad$\WHILEDO{G}{\mathit{Com}}$\,. Then there are $p,d$ functions such that $V,p,d$ satisfy the $p,d$-progress condition of \Thm{t1928}.
\begin{proof}
Let $\epsilon{>}0$ satisfy \Eqn{e1533-1}.  Observe first that \Eqn{e1533-1} implies that for any state $\sigma$ satisfying $G\land I$, we must have $V(\sigma){\geq}\epsilon$. Let $R^\ast$ be the infimum of the image of $V$, so that we have $R^\ast{\geq}\epsilon{>}0$.

Define $p,d$ so that for any $R{>}0$ we have $d(R) = \epsilon/2$ and
\begin{align}
  p(R) =\quad&\epsilon/(2R{-}\epsilon)\quad \textit{if}\quad R{\geq}R^\ast\\
          =\quad&\epsilon/(2R^\ast{-}\epsilon)~. \quad \textit{otherwise}
\end{align} 
We show that these definitions satisfy the conditions in \Thm{t1928} for $p, d$ progress.

Given any $\sigma$, set $R=V(\sigma)$; then since $(R{-}\epsilon/2)\cdot{\iverson{V>R{-}\epsilon/2}} \leq V$, and $\awp{\mathit{Com}}{}$ is monotone and scaling,
\footnote{The first inequality is actually an instance of \emph{Markov's Inequality \citep{Grimmett:86}}.}
we have
\begin{Equation}\label{e1751}
(R{-}\epsilon/2)\cdot\iverson{G\land I}\cdot\awp{\mathit{Com}} {\iverson{V>R{-}\epsilon/2}} \Wide{\leq}  \iverson{G\land I}\cdot\awp{\mathit{Com}} {V}  \Wide{\leq} V {-} \epsilon~,
\end{Equation}
where the second inequality follows from \Eqn{e1533-1}.
We now reason:
\begin{Reason}
\Step{}{\wp{\mathit{Com}} {\iverson{V\leq R{-}d(R)}}}
\StepR{$=$}{definition $d$ above}{\wp{\mathit{Com}} {\iverson{V\leq R{-}\epsilon/2}}}
\StepR{$\geq$}{$\one \geq \iverson{G\land I \land V{=}R }$}
{\iverson{G\land I \land V{=}R }\cdot\wp{\mathit{Com}} {\iverson{V\leq R{-}\epsilon/2}}}
\WideStepR{$\geq$}{$\awp{\mathit{Com}}{\iverson{X}} + \wp{\mathit{Com}}{\iverson{\neg X}}\leq 1$; see below.}
{\iverson{G\land I \land V{=}R }\cdot(1 - \awp{\mathit{Com}} {\iverson{V>R{-}\epsilon/2}})}
\StepR{$=$}{\Eqn{e1751}}
{\iverson{G\land I \land V{=}R }\cdot(1 -(V {-} \epsilon)/(R{-}\epsilon/2))}
\StepR{$=$}{non-zero only when $V{=}R$}
{\iverson{G\land I \land V{=}R }\cdot(1 -(R {-} \epsilon)/(R{-}\epsilon/2))}
\StepR{$=$}{arithmetic}
{\iverson{G\land I \land V{=}R }\cdot(\epsilon/(2R{-}\epsilon))}
\StepR{$=$}{definition of $p(R)$, and $R{\geq}R^\ast$}
{\iverson{G\land I \land V{=}R }\cdot p(R)~.}
\end{Reason}
\end{proof}
\end{lemma}
For ``see below'' we note that the property can easily be established for straight-line programs, using a structural induction similar to the one in \Sec{s0918}.

A consequence of \Lem{t1525}'s \Eqn{e1751} is that if a \pGCL\ loop terminates in finite expected time, then there exists a super-martingale satisfying \Thm{t1928}; this follows from the existence of a ranking super-martingale \citep{Fioriti:2015}. 


\newpage
\section{Appendix: Countable- vs.\ Uncountable Branching\hfill[from \Sec{s1130}]}\label{a1713}


\begin{figure}[t]
	\begin{center}
		\let\oldarraycolsep\arraycolsep
		\arraycolsep=0pt
		\begin{tikzpicture}[every state/.append style={thick, inner sep=0pt, minimum size=1.5cm}]
			\path[use as bounding box] (0,-1) rectangle (14,7);
			\path	node at (6,4) {\includegraphics[scale=0.2]{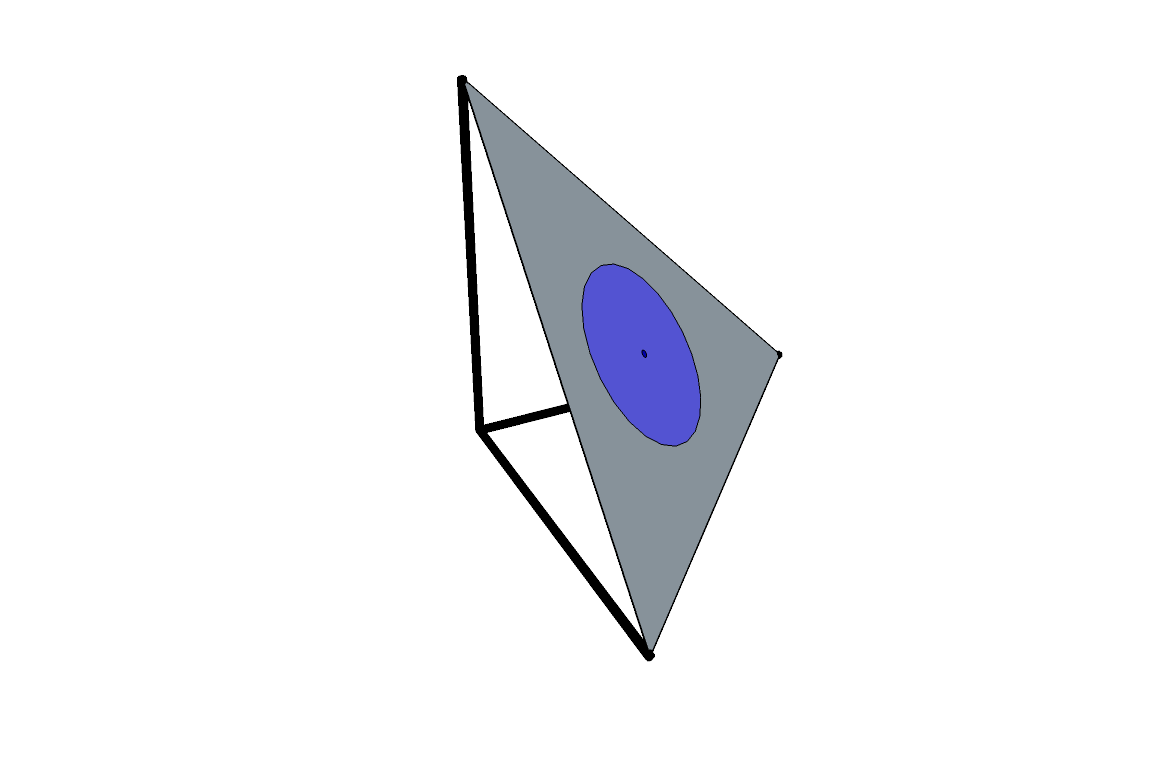}};
			\draw[->] (7.8,6.4) -- +(-2.15,-1);
			\draw[->] (7.8,5.4) -- +(-1.25,-1.1);
			\draw[->] (7.8,2.955) -- +(-1.25,0.9);

			\path	node at (8,6.5) {\makebox[0pt][l]{\small All (full) distributions on $\Sigma$}}
				node at (8,5.5) {\makebox[0pt][l]{\small Uniform distribution $(\NF{1}{3},\NF{1}{3},\NF{1}{3})$ on $\Sigma$}}
				node at (8,3) {\makebox[0pt][l]{\small \parbox{15em}{Uncountable-but-closed subset (blue) of distributions on $\Sigma$}}}
				node at (3,5) {$\Sigma = \{0,1,2\}$}
				node at (4.5,3.65) {$(0,0,0)$}
				node at (5.2,6.55) {$\sigma=2$}
				node at (8.1,4.24) {$\sigma=1$}
				node at (6.5,1.7) {$\sigma=0$}
				node at (7,0) {	\parbox{.9\linewidth}{\small
							In general a point $(x,y,z)$ represents the distribution $\sigma{=}0$ with probability $x$\ldots\ etc.
							Thus point $(1,0,0)$ represents ``$\sigma$ is definitely 0.'' \\
							The disk of distributions (in blue, and including its border) is closed but uncountable;
							and it is not the convex closure of any countable subset.}};
		\end{tikzpicture}
		\arraycolsep=\oldarraycolsep
	\end{center}
	\caption{An uncountable yet still closed set of discrete distributions.}
	\label{f1722}
\end{figure}


The blue set of distributions in \Fig{f1722}, including its border, is a closed disk on the plane $x{+}y{+}z=1$ in the space of discrete distributions over $\Sigma{=}\{0,1,2\}$, where the $x$-co{\"o}rdinate of a point is the probability that the distribution it represents assigns to element $\sigma{=}0$ of the state-space $\Sigma$ etc. Thus point $(1,0,0)$ represents the point distribution where $\sigma$ is 0 (with probability one). The uniform distribution is therefore at point $(\NF{1}{3},\NF{1}{3},\NF{1}{3})$, the centre of mass of the grey triangle; indeed, considering the triangle on its own, we see that the set of full (i.e.\ summing to one) distributions are represented barycentrically.

All \AST\ programs' final distributions over $\Sigma$ lie on this plane; however loops whose termination probability was less than one would produce sub-distributions ``below'' the plane, more precisely points lying in the proper interior of the tetrahedron whose base is that plane $x{+}y{+}z=1$ and whose apex is the origin $(0,0,0)$. The origin in particular represents the program $\ABORT$, equivalently\quad$\WHILEDO{\true}{\SKIP}$\quad whose termination probability is actually 0.
\footnote{In fact \pGCL\ incorporates Smyth-style upwards closure \cite{Smyth:78}, so that $\ABORT$ would be the whole tetrahedron.}

We believe that \Fig{f1722} is not the output of any \pGCL\ program with finite (or even countably infinite) demonic choice; but a transition system of this kind would still fall within the scope of our termination rule, because the set is closed.

A full account of this geometric view of probabilistic/demonic programs is given is \cite[Ch.~6]{McIver:05a}.

{\Fx
\newpage
\section{Proof of Equivalence of \Eqn{e1801} and \Eqn{e1814} \hfill[from \Sec{s0836}]}\label{s1315}
In this section we establish (actually, confirm) a loop- de-nesting identity that is familiar from non-probabilistic, i.e.\ standard programming; and indeed the proof has the same structure as in the standard case.

We proceed slightly more generally than in \Sec{s0836}, considering two variables $a,b$ that are tested by predicates $a\textrm{?},b\textrm{?}$ and updated by commands $\Delta a,\Delta b$. The identity we establish is that
\[
	\WHILEDO{a\textrm{?}}{\IFELSE{b\textrm{?}}{\Delta b}{\Delta a}}
	\Wide{=}
	\WHILEDO{a\textrm{?}}{\COMPOSE{\WHILEDO{b\textrm{?}}{\Delta b}}{\Delta a}} ~,
\]
and the principal reasoning step is that if the $\Delta b$ path is taken on the left then $a\textrm{?}$ is still true afterwards, so that the re-entry test for the loop is redundant. That idea is captured by the inner loop on the right, where the unnecessary $a\textrm{?}$-tests have been removed. (The remaining steps are routine unfolding and folding of loops.)

In this presentation, to exploit the connection with standard program algebra, we allow an $\mathtt{if}$-test to be a probability, so that $\PCHOICE{\Delta b}{b\textrm{?}}{\Delta a}$ can be written $\IFELSE{b\textrm{?}}{\Delta b}{\Delta a}$.
\footnote{This works nicely because we have already that
$\quad
	\IFELSE{c}{P}{Q} = \PCHOICE{P}{\iverson{c}}{Q}
\quad$
without any hand-waving at all. (Recall that $\iverson{c}$ for Boolean $c$ is 1 if $c$ else 0.))}
This means for example that the defining identity
\[
 \WHILEDO{p}{C} \Wide{=} \IF{p}{\{\COMPOSE{C}{\WHILEDO{p}{C}}\}} \Wide{=} \PCHOICE{\COMPOSE{C}{\WHILEDO{p}{C}}}{p}{\SKIP}\quad,
\]
for a probabilistically guarded loop, is very natural.
We calculate first
\begin{align*}
&\WHILEDO{a\textrm{?}}{\IFELSE{b\textrm{?}}{\Delta b}{\Delta a}} \tag*{\Eqn{e1801}} \\
=\quad&\RECDO{P}{\IF{a\textrm{?}}{\{\COMPOSE{\{\IFELSE{b\textrm{?}}{\Delta b}{\Delta a}\}}{P}\}}} \tag*{defn.\ $\mathtt{while}$} \\
\makebox[0pt][r]{$\dagger$\quad}
=\quad&\RECDO{P}{\IF{a\textrm{?}}{\{\IFELSE{b\textrm{?}}{\{\COMPOSE{\Delta b}{P}\}}{\{\COMPOSE{\Delta a}{P}\}}\}}}\quad, \tag*{move $P$ into $\mathtt{if}$}
\end{align*}
and then continue with
\begin{align*}
&\WHILEDO{a\textrm{?}}{\COMPOSE{\WHILEDO{b\textrm{?}}{\Delta b}}{\Delta a}} \tag*{\Eqn{e1814}}\\
=\quad& \RECDO{Q}{\IF{a\textrm{?}}{\COMPOSE{\{\COMPOSE{\WHILEDO{b\textrm{?}}{\Delta b}}{\Delta a}\}}{Q}}} \tag*{defn.\ $\mathtt{while}$} \\
=\quad& \RECDO{Q}{\IF{a\textrm{?}}{\COMPOSE{\{\COMPOSE{\RECDO{R}{\IF{b\textrm{?}}{\{\COMPOSE{\Delta b}{R}\}}}}{\Delta a}\}}{Q}}} \tag*{defn.\ $\mathtt{while}$} \\
=\quad& \RECDO{Q}{\IF{a\textrm{?}}{\COMPOSE{\RECDO{R}{\IFELSE{b\textrm{?}}{\{\COMPOSE{\Delta b}{R}\}}{\Delta a}}}{Q}}} \tag*{see below\makebox[0pt][l]{\quad\$}} \\
=\quad& \RECDO{Q}{\IF{a\textrm{?}}{\RECDO{R}{\IFELSE{b\textrm{?}}{\{\COMPOSE{\Delta b}{R}\}}{\{\COMPOSE{\Delta a}{Q}\}}}}} \tag*{again below\makebox[0pt][l]{\quad\$}} \\
=\quad& \RECDO{Q}{\IF{a\textrm{?}}{\RECDO{R}{\IFELSE{b\textrm{?}}{\{\COMPOSE{\Delta b}{Q}\}}{\{\COMPOSE{\Delta a}{Q}\}}}}} \tag*{$R{=}Q$ when $a\textrm{?}$} \\
=\quad& \RECDO{Q}{\IF{a\textrm{?}}{\{\IFELSE{b\textrm{?}}{\{\COMPOSE{\Delta b}{Q}\}}{\{\COMPOSE{\Delta a}{Q}\}}}\}} \tag*{remove unused $R$} \quad,
\end{align*}
whence by alpha conversion to $(\dagger)$ above we have the equality $\textrm{\Eqn{e1801}}$=$\textrm{\Eqn{e1814}}$ we sought.
\footnotemark 

The ``see below'' remarks at \$\ above refer to the routine identity
\[
	\COMPOSE{\RECDO{X}{\IFELSE{B}{(\COMPOSE{\mathit{then}}{X})}{\mathit{else}}}}{\mathit{after}}
	\quad=\quad
	\RECDO{}{\IFELSE{X}{(\COMPOSE{\mathit{then}}{X})}{(\COMPOSE{\mathit{else}}{\mathit{after}})}} \quad,
\]
shown by equating the iterates in the $\sup$-expression for the least fixed point, just as in standard program algebra.
}

\newpage
\begin{acks}                            
McIver and Morgan are grateful to David Basin and the Information Security Group at ETH Z{\"u}rich for hosting a six-month stay in Switzerland, during part of which this work began. And thanks particularly to Andreas Lochbihler, who shared with us the probabilistic termination problem that led to it. They acknowledge the support of \grantnum{ARC grant}{ARC grant DP140101119}.
  
Part of this work was carried out during the Workshop on Probabilistic Programming Semantics at McGill University's Bellairs Research Institute on Barbados organised by Alexandra Silva and Prakash Panangaden. 

Kaminski and Katoen are grateful to Sebastian Junges for spotting a flaw in \Sec{s0838}.
\end{acks}
\footnotetext{
The explicitly recursive versions of \Eqn{e1801} and \Eqn{e1814} are
\begin{Equation}\label{e1811}
	\begin{array}{l}
		\{n\geq1\} \\
		\ASSIGN{x}{1} \\
		\REC{P} \\
		\quad\IF{x{\neq}0} \\
		\quad\quad\PCHOICE{
					\PCHOICE{\ASSIGN{x}{x{-}1}
					}{\NF{1}{2}
					}{\ASSIGN{x}{x{+}1}
					}
				}{\NF{1}{n}
				}{\ASSIGN{n}{n{+}1}
				}\\
		\quad\quad P \\
		\}
	\end{array}
\end{Equation}
and
\begin{Equation}\label{e1813}
	\begin{array}{l}
		\{n\geq1\} \\
		\ASSIGN{x}{1} \\
		\REC{Q} \\
		\quad\IF{x{\neq}0} \\
		\quad\quad	\RECDO{R}{\PCHOICE{\SKIP}{\NF{1}{n}}{\COMPOSE{\ASSIGN{n}{n{+}1}}{R}}} \\
		\quad\quad	\PCHOICE{\ASSIGN{x}{x{-}1}
						}{\NF{1}{2}
						}{\ASSIGN{x}{x{+}1}
					} \\
		\quad\quad Q \\
		\}\quad,
	\end{array}
\end{Equation}
}

\newpage


\begin{thebibliography}{27}


\ifx \showCODEN    \undefined \def \showCODEN     #1{\unskip}     \fi
\ifx \showDOI      \undefined \def \showDOI       #1{#1}\fi
\ifx \showISBNx    \undefined \def \showISBNx     #1{\unskip}     \fi
\ifx \showISBNxiii \undefined \def \showISBNxiii  #1{\unskip}     \fi
\ifx \showISSN     \undefined \def \showISSN      #1{\unskip}     \fi
\ifx \showLCCN     \undefined \def \showLCCN      #1{\unskip}     \fi
\ifx \shownote     \undefined \def \shownote      #1{#1}          \fi
\ifx \showarticletitle \undefined \def \showarticletitle #1{#1}   \fi
\ifx \showURL      \undefined \def \showURL       {\relax}        \fi
\providecommand\bibfield[2]{#2}
\providecommand\bibinfo[2]{#2}
\providecommand\natexlab[1]{#1}
\providecommand\showeprint[2][]{arXiv:#2}

\bibitem[\protect\citeauthoryear{Agrawal, Chatterjee, and Novotn\'{y}}{Agrawal
  et~al\mbox{.}}{2018}]%
        {Agrawal:2018}
\bibfield{author}{\bibinfo{person}{Sheshansh Agrawal},
  \bibinfo{person}{Krishnendu Chatterjee}, {and} \bibinfo{person}{Petr
  Novotn\'{y}}.} \bibinfo{year}{2018}\natexlab{}.
\newblock \showarticletitle{Lexicographic Ranking Supermartingales: An
  Efficient Approach to Termination of Probabilistic Programs}. In
  \bibinfo{booktitle}{{\em Proceedings of the 45th ACM SIGPLAN Symposium on
  Principles of Programming Languages}} {\em (\bibinfo{series}{POPL 2018})}.
  \bibinfo{publisher}{ACM}, \bibinfo{address}{New York, NY, USA}.
\newblock


\bibitem[\protect\citeauthoryear{Blackwell}{Blackwell}{1955}]%
        {Blackwell:55}
\bibfield{author}{\bibinfo{person}{David Blackwell}.}
  \bibinfo{year}{1955}\natexlab{}.
\newblock \showarticletitle{On Transient {Markov} Processes with a Countable
  Number of States and Stationary Transition Probabilities}.
\newblock \bibinfo{journal}{{\em Ann. Math. Statist.\/}}  \bibinfo{volume}{26}
  (\bibinfo{year}{1955}), \bibinfo{pages}{654--658}.
\newblock


\bibitem[\protect\citeauthoryear{Celiku and McIver}{Celiku and McIver}{2005}]%
        {Celiku:05}
\bibfield{author}{\bibinfo{person}{Orieta Celiku} {and}
  \bibinfo{person}{Annabelle McIver}.} \bibinfo{year}{2005}\natexlab{}.
\newblock \showarticletitle{Compositional Specification and Analysis of
  Cost-Based Properties in Probabilistic Programs}. In \bibinfo{booktitle}{{\em
  {FM}}} {\em (\bibinfo{series}{Lecture Notes in Computer Science})},
  Vol.~\bibinfo{volume}{3582}. \bibinfo{publisher}{Springer},
  \bibinfo{pages}{107--122}.
\newblock


\bibitem[\protect\citeauthoryear{Chakarov}{Chakarov}{2016}]%
        {Chakarov:2016aa}
\bibfield{author}{\bibinfo{person}{Aleksandar Chakarov}.}
  \bibinfo{year}{2016}\natexlab{}.
\newblock {\em \bibinfo{title}{Deductive Verification of Infinite-State
  Stochastic Systems using Martingales}}.
\newblock \bibinfo{thesistype}{Ph.D. Dissertation}. \bibinfo{school}{University
  of Colorado at Boulder}.
\newblock


\bibitem[\protect\citeauthoryear{Chakarov and Sankaranarayanan}{Chakarov and
  Sankaranarayanan}{2013}]%
        {Chakarov:2013}
\bibfield{author}{\bibinfo{person}{Aleksandar Chakarov} {and}
  \bibinfo{person}{Sriram Sankaranarayanan}.} \bibinfo{year}{2013}\natexlab{}.
\newblock \showarticletitle{Probabilistic Program Analysis with Martingales}.
  In \bibinfo{booktitle}{{\em {CAV}}} {\em (\bibinfo{series}{Lecture Notes in
  Computer Science})}, Vol.~\bibinfo{volume}{8044}.
  \bibinfo{publisher}{Springer}, \bibinfo{pages}{511--526}.
\newblock


\bibitem[\protect\citeauthoryear{Chatterjee and Fu}{Chatterjee and Fu}{2017}]%
        {Chatterjee:2017ab}
\bibfield{author}{\bibinfo{person}{Krishnendu Chatterjee} {and}
  \bibinfo{person}{Hongfei Fu}.} \bibinfo{year}{2017}\natexlab{}.
\newblock \showarticletitle{Termination of Nondeterministic Recursive
  Probabilistic Programs}.
\newblock \bibinfo{journal}{{\em CoRR\/}}  \bibinfo{volume}{abs/1701.02944}
  (\bibinfo{year}{2017}).
\newblock


\bibitem[\protect\citeauthoryear{Chatterjee, Novotn\'{y}, and
  \v{Z}ikeli\'{c}}{Chatterjee et~al\mbox{.}}{2017}]%
        {Chatterjee:2017aa}
\bibfield{author}{\bibinfo{person}{Krishnendu Chatterjee},
  \bibinfo{person}{Petr Novotn\'{y}}, {and} \bibinfo{person}{Dorde
  \v{Z}ikeli\'{c}}.} \bibinfo{year}{2017}\natexlab{}.
\newblock \showarticletitle{Stochastic Invariants for Probabilistic
  Termination}. In \bibinfo{booktitle}{{\em Proceedings of the 44th ACM SIGPLAN
  Symposium on Principles of Programming Languages}} {\em
  (\bibinfo{series}{POPL 2017})}. \bibinfo{publisher}{ACM},
  \bibinfo{address}{New York, NY, USA}, \bibinfo{pages}{145--160}.
\newblock
\showISBNx{978-1-4503-4660-3}
\showDOI{%
\url{https://doi.org/10.1145/3009837.3009873}}


\bibitem[\protect\citeauthoryear{Dijkstra}{Dijkstra}{1976}]%
        {Dijkstra:76}
\bibfield{author}{\bibinfo{person}{Edsger~W. Dijkstra}.}
  \bibinfo{year}{1976}\natexlab{}.
\newblock \bibinfo{booktitle}{{\em A Discipline of Programming}}.
\newblock \bibinfo{publisher}{Prentice-Hall}.
\newblock


\bibitem[\protect\citeauthoryear{Esparza, Gaiser, and Kiefer}{Esparza
  et~al\mbox{.}}{2012}]%
        {Esparza:2012}
\bibfield{author}{\bibinfo{person}{Javier Esparza}, \bibinfo{person}{Andreas
  Gaiser}, {and} \bibinfo{person}{Stefan Kiefer}.}
  \bibinfo{year}{2012}\natexlab{}.
\newblock \showarticletitle{Proving Termination of Probabilistic Programs Using
  Patterns}. In \bibinfo{booktitle}{{\em {CAV}}} {\em (\bibinfo{series}{Lecture
  Notes in Computer Science})}, Vol.~\bibinfo{volume}{7358}.
  \bibinfo{publisher}{Springer}, \bibinfo{pages}{123--138}.
\newblock


\bibitem[\protect\citeauthoryear{Ferrer~Fioriti and Hermanns}{Ferrer~Fioriti
  and Hermanns}{2015}]%
        {Fioriti:2015}
\bibfield{author}{\bibinfo{person}{Luis~Mar\'{\i}a Ferrer~Fioriti} {and}
  \bibinfo{person}{Holger Hermanns}.} \bibinfo{year}{2015}\natexlab{}.
\newblock \showarticletitle{Probabilistic Termination: Soundness, Completeness,
  and Compositionality}. In \bibinfo{booktitle}{{\em Proceedings of the 42nd
  Annual ACM SIGPLAN-SIGACT Symposium on Principles of Programming Languages}}
  {\em (\bibinfo{series}{POPL 2015})}. \bibinfo{publisher}{ACM},
  \bibinfo{address}{New York, NY, USA}, \bibinfo{pages}{489--501}.
\newblock
\showISBNx{978-1-4503-3300-9}
\showDOI{%
\url{https://doi.org/10.1145/2676726.2677001}}


\bibitem[\protect\citeauthoryear{Foster}{Foster}{1951}]%
        {Foster:1951aa}
\bibfield{author}{\bibinfo{person}{F.~G. Foster}.}
  \bibinfo{year}{1951}\natexlab{}.
\newblock \showarticletitle{Markoff chains with an enumerable number of states
  and a class of cascade processes}.
\newblock \bibinfo{journal}{{\em Cambridge Philosophical Society\/}}
  \bibinfo{volume}{1}, \bibinfo{number}{47} (\bibinfo{year}{1951}),
  \bibinfo{pages}{77--85}.
\newblock


\bibitem[\protect\citeauthoryear{Foster}{Foster}{1952}]%
        {Foster:1952ab}
\bibfield{author}{\bibinfo{person}{F.~G. Foster}.}
  \bibinfo{year}{1952}\natexlab{}.
\newblock \showarticletitle{On {Markov} Chains with an Enumerable Infinity of
  States}.
\newblock \bibinfo{journal}{{\em Mathematical Proceedings of the Cambridge
  Philosophical Society\/}} \bibinfo{number}{4} (\bibinfo{date}{Oct}
  \bibinfo{year}{1952}), \bibinfo{pages}{587--591}.
\newblock
\showDOI{%
\url{https://doi.org/10.1017/S0305004100076362}}


\bibitem[\protect\citeauthoryear{Gretz, Katoen, and McIver}{Gretz
  et~al\mbox{.}}{2014}]%
        {GretzKM14}
\bibfield{author}{\bibinfo{person}{Friedrich Gretz},
  \bibinfo{person}{Joost{-}Pieter Katoen}, {and} \bibinfo{person}{Annabelle
  McIver}.} \bibinfo{year}{2014}\natexlab{}.
\newblock \showarticletitle{Operational versus weakest pre-expectation
  semantics for the probabilistic guarded command language}.
\newblock \bibinfo{journal}{{\em Perform. Eval.\/}}  \bibinfo{volume}{73}
  (\bibinfo{year}{2014}), \bibinfo{pages}{110--132}.
\newblock


\bibitem[\protect\citeauthoryear{Grimmett and Welsh}{Grimmett and
  Welsh}{1986}]%
        {Grimmett:86}
\bibfield{author}{\bibinfo{person}{G.R. Grimmett} {and} \bibinfo{person}{D.
  Welsh}.} \bibinfo{year}{1986}\natexlab{}.
\newblock \bibinfo{booktitle}{{\em Probability: an Introduction}}.
\newblock \bibinfo{publisher}{Oxford Science Publications}.
\newblock


\bibitem[\protect\citeauthoryear{Hart, Sharir, and Pnueli}{Hart
  et~al\mbox{.}}{1983}]%
        {Hart:83}
\bibfield{author}{\bibinfo{person}{Sergiu Hart}, \bibinfo{person}{Micha
  Sharir}, {and} \bibinfo{person}{Amir Pnueli}.}
  \bibinfo{year}{1983}\natexlab{}.
\newblock \showarticletitle{Termination of Probabilistic Concurrent Programs}.
\newblock \bibinfo{journal}{{\em ACM Trans. Program. Lang. Syst.\/}}
  \bibinfo{volume}{5}, \bibinfo{number}{3} (\bibinfo{date}{July}
  \bibinfo{year}{1983}), \bibinfo{pages}{356--380}.
\newblock
\showISSN{0164-0925}
\showDOI{%
\url{https://doi.org/10.1145/2166.357214}}


\bibitem[\protect\citeauthoryear{Hoare}{Hoare}{1969}]%
        {Hoare:69}
\bibfield{author}{\bibinfo{person}{C.~A.~R. Hoare}.}
  \bibinfo{year}{1969}\natexlab{}.
\newblock \showarticletitle{An Axiomatic Basis for Computer Programming}.
\newblock \bibinfo{journal}{{\em Commun. {ACM}\/}} \bibinfo{volume}{12},
  \bibinfo{number}{10} (\bibinfo{year}{1969}), \bibinfo{pages}{576--580}.
\newblock


\bibitem[\protect\citeauthoryear{Kaminski, Katoen, Matheja, and
  Olmedo}{Kaminski et~al\mbox{.}}{2016}]%
        {Kaminski:16}
\bibfield{author}{\bibinfo{person}{Benjamin~Lucien Kaminski},
  \bibinfo{person}{Joost{-}Pieter Katoen}, \bibinfo{person}{Christoph Matheja},
  {and} \bibinfo{person}{Federico Olmedo}.} \bibinfo{year}{2016}\natexlab{}.
\newblock \showarticletitle{Weakest Precondition Reasoning for Expected
  Run-Times of Probabilistic Programs}. In \bibinfo{booktitle}{{\em {ESOP}}}
  {\em (\bibinfo{series}{Lecture Notes in Computer Science})},
  Vol.~\bibinfo{volume}{9632}. \bibinfo{publisher}{Springer},
  \bibinfo{pages}{364--389}.
\newblock


\bibitem[\protect\citeauthoryear{Kendall}{Kendall}{1951}]%
        {Kendall:1951aa}
\bibfield{author}{\bibinfo{person}{David~G. Kendall}.}
  \bibinfo{year}{1951}\natexlab{}.
\newblock \showarticletitle{On non-dissipative Markoff chains with an
  enumerable infinity of states}.
\newblock \bibinfo{journal}{{\em Mathematical Proceedings of the Cambridge
  Philosophical Society\/}} \bibinfo{volume}{47}, \bibinfo{number}{3}
  (\bibinfo{date}{001 007} \bibinfo{year}{1951}), \bibinfo{pages}{633--634}.
\newblock
\showDOI{%
\url{https://doi.org/10.1017/S0305004100027055}}


\bibitem[\protect\citeauthoryear{Knopp}{Knopp}{1928}]%
        {Knopp:1928aa}
\bibfield{author}{\bibinfo{person}{Konrad Knopp}.}
  \bibinfo{year}{1928}\natexlab{}.
\newblock \bibinfo{booktitle}{{\em Theory and Application of Infinite Series}}.
\newblock \bibinfo{publisher}{London}.
\newblock


\bibitem[\protect\citeauthoryear{Kozen}{Kozen}{1985}]%
        {Kozen:85}
\bibfield{author}{\bibinfo{person}{Dexter Kozen}.}
  \bibinfo{year}{1985}\natexlab{}.
\newblock \showarticletitle{A Probabilistic {PDL}}.
\newblock \bibinfo{journal}{{\em J. Comput. Syst. Sci.\/}}
  \bibinfo{volume}{30}, \bibinfo{number}{2} (\bibinfo{year}{1985}),
  \bibinfo{pages}{162--178}.
\newblock


\bibitem[\protect\citeauthoryear{Lago and Grellois}{Lago and Grellois}{2017}]%
        {Lago:2017aa}
\bibfield{author}{\bibinfo{person}{Ugo~Dal Lago} {and} \bibinfo{person}{Charles
  Grellois}.} \bibinfo{year}{2017}\natexlab{}.
\newblock \showarticletitle{Probabilistic Termination by Monadic Affine Sized
  Typing}. In \bibinfo{booktitle}{{\em {ESOP}}} {\em (\bibinfo{series}{Lecture
  Notes in Computer Science})}, Vol.~\bibinfo{volume}{10201}.
  \bibinfo{publisher}{Springer}, \bibinfo{pages}{393--419}.
\newblock


\bibitem[\protect\citeauthoryear{McIver and Morgan}{McIver and Morgan}{2005}]%
        {McIver:05a}
\bibfield{author}{\bibinfo{person}{Annabelle McIver} {and}
  \bibinfo{person}{Carroll Morgan}.} \bibinfo{year}{2005}\natexlab{}.
\newblock \bibinfo{booktitle}{{\em Abstraction, Refinement and Proof for
  Probabilistic Systems}}.
\newblock \bibinfo{publisher}{Springer}.
\newblock


\bibitem[\protect\citeauthoryear{McIver and Morgan}{McIver and Morgan}{2016}]%
        {McIver:2016aa}
\bibfield{author}{\bibinfo{person}{Annabelle McIver} {and}
  \bibinfo{person}{Carroll Morgan}.} \bibinfo{year}{2016}\natexlab{}.
\newblock \showarticletitle{A New Rule for Almost-Certain Termination of
  Probabilistic and Demonic Programs}.
\newblock \bibinfo{journal}{{\em CoRR\/}}  \bibinfo{volume}{abs/1612.01091}
  (\bibinfo{year}{2016}).
\newblock


\bibitem[\protect\citeauthoryear{Morgan}{Morgan}{1996}]%
        {Morgan:96b}
\bibfield{author}{\bibinfo{person}{C.C. Morgan}.}
  \bibinfo{year}{1996}\natexlab{}.
\newblock \showarticletitle{Proof Rules for Probabilistic Loops}. In
  \bibinfo{booktitle}{{\em Proc BCS-FACS 7th Refinement Workshop}} {\em
  (\bibinfo{series}{Workshops in Computing})},
  \bibfield{editor}{\bibinfo{person}{He~Jifeng}, \bibinfo{person}{John Cooke},
  {and} \bibinfo{person}{Peter Wallis}} (Eds.). \bibinfo{publisher}{Springer}.
\newblock
\newblock
\shownote{\texttt{http://www.bcs.org/upload/pdf/ewic\underline{
  }rw96\underline{ }paper10.pdf}.}


\bibitem[\protect\citeauthoryear{Morgan, McIver, and Seidel}{Morgan
  et~al\mbox{.}}{1996}]%
        {Morgan:96d}
\bibfield{author}{\bibinfo{person}{Carroll Morgan}, \bibinfo{person}{Annabelle
  McIver}, {and} \bibinfo{person}{Karen Seidel}.}
  \bibinfo{year}{1996}\natexlab{}.
\newblock \showarticletitle{Probabilistic Predicate Transformers}.
\newblock \bibinfo{journal}{{\em ACM Trans. Program. Lang. Syst.\/}}
  \bibinfo{volume}{18}, \bibinfo{number}{3} (\bibinfo{date}{May}
  \bibinfo{year}{1996}), \bibinfo{pages}{325--353}.
\newblock
\showISSN{0164-0925}
\showDOI{%
\url{https://doi.org/10.1145/229542.229547}}


\bibitem[\protect\citeauthoryear{Olmedo, Kaminski, Katoen, and Matheja}{Olmedo
  et~al\mbox{.}}{2016}]%
        {Olmedo:2016aa}
\bibfield{author}{\bibinfo{person}{Federico Olmedo},
  \bibinfo{person}{Benjamin~Lucien Kaminski}, \bibinfo{person}{Joost-Pieter
  Katoen}, {and} \bibinfo{person}{Christoph Matheja}.}
  \bibinfo{year}{2016}\natexlab{}.
\newblock \showarticletitle{Reasoning About Recursive Probabilistic Programs}.
  In \bibinfo{booktitle}{{\em Proceedings of the 31st Annual ACM/IEEE Symposium
  on Logic in Computer Science}} {\em (\bibinfo{series}{LICS '16})}.
  \bibinfo{publisher}{ACM}, \bibinfo{address}{New York, NY, USA},
  \bibinfo{pages}{672--681}.
\newblock
\showISBNx{978-1-4503-4391-6}
\showDOI{%
\url{https://doi.org/10.1145/2933575.2935317}}


\bibitem[\protect\citeauthoryear{Smyth}{Smyth}{1978}]%
        {Smyth:78}
\bibfield{author}{\bibinfo{person}{M.B. Smyth}.}
  \bibinfo{year}{1978}\natexlab{}.
\newblock \showarticletitle{Power Domains}.
\newblock \bibinfo{journal}{{\em Jnl Comp Sys Sci\/}}  \bibinfo{volume}{16}
  (\bibinfo{year}{1978}), \bibinfo{pages}{23--36}.
\newblock


\end{thebibliography}

\end{document}

\EndDocument
\newpage
\MakeGreenRoom 
\ifNoGreenRoom\relax\else 

\bigskip\Ct{In the Green Room, the coloured horizontal lines separate contributions, with the colour indicating who's responsible for the material that follows it.}

\bigskip
\A{\Divider}

\section{Notes on the 2 dimensional random walk}\label{s1743}

To tie up the last little piece, we need to show that the variant constructed by the two dimensional random walk specifically satisfies the progress condition \Itm{i1651-3} for $V$. This will follow if, for example, there is a non-zero probability that the mover lands on a state with strictly smaller variant. This is what we shall show for a \emph{slight variation} of the 2dSRW.

From Foster's paper, we find the following equation on page 3, used to show the super-martingale property of the constructed variant. I will use Foster's notation to make sure that nothing is lost in translation: it is
\[
 	\sum_{j\geq 0} p_{ij}x_j
	= \sum_{j\geq 0} p_{ij} \sum_{r\geq 1}\lambda^{(r)}f^{(r)}_j
	= \sum_{r\geq 1}\lambda^{(r)}f^{(r+1)}_i
	\leq \sum_{r\geq 1}\lambda^{(r+1)}f^{(r+1)}_i
	\leq  \sum_{r\geq 1}\lambda^{(r)}f^{(r)}_i = x_i ~,
\]
where we must show that one of the two inequalities is strict for the special case of the 2dSRW, which would mean that at least one one of the neighbours of state labelled $s_i$ has one of its variant values ${s_j}$ with $V(s_i){>}V(s_j)$.
In fact the first inequality is strict if  $\lambda^{(r)}{>}\lambda^{(r')}$ whenever $r{>}r'$, because then we are comparing two sums of terms where each corresponding terms are strictly larger --- we just need that $f^{(r)}_i{>}0$ for at least one $r$. For that we recall the definition
\[
	\lambda^{(r)} \Wide{=} 1/\sqrt{\sum_{s \geq r} f_1^{(s)}}\quad,
\]
and then argue that $\lambda^{(r)}{<}\lambda^{(r')}$ if and only if $\sum_{s \geq r'} f_1^{(s)}<\sum_{s \geq r} f_1^{(s)}$ when $r'{>}r$. This for example is true if $f_1^{r}{>}0$ for all $r{>}1$. But recall that $f_1^{r}$ is the probability that starting from state labelled $1$ the walker reaches the origin for the first time at the $r$'th step. 

In fact this is not true for the walker described in Foster,
\footnote{\Cx\ldots because $f_1^{r}$ is zero for all even $r$.}
but it is true for e.g.\ the walker that has a $1/2$ chance of staying where it is, and a $1/8$ chance of moving to each of its four neighbours. This slight variation of the walker does have $f_1^{r}{>}0$ for all $r{>}1$.
\Cf{That's a nice, symmetric and elegant modification. But it would be enough even if only State 1 was equipped with a loper-loop, admittedly a bit of a hack: for then the walker could ``dally'' there for $r{-}1$ steps, and on the very last, i.e.\ the $r$th step, move finally to the origin. Nicht wahr?}

Therefore this version of the walker satisfies our rule, and still does not seem to be within reach of others' rules.

Oh, but wait! If we use this $V$ on the original 2dSRW then we still obtain the strict inequality, so this also works for the original presentation.
 \Af{Finally we also note that \Thm{t1928} can be used in compositional arguments. What do we mean here exactly?}

\bigskip
\B{\Divider}
\section{Material about the $\nabla$-rule}
\begin{definition}[$\nabla$-Rule]
	A quasi variant $V$ for $\WHILEDO{G}{C}$ is called a \emph{$\nabla$-Variant for $\WHILEDO{G}{C}$}, iff there exists a function $\nabla\colon \Rpos \To \Rpos$, such that
	\begin{enumerate}
		\item
			$\nabla$ evaluates to zero exactly at zero, i.e.\
			\begin{align*}
				\nabla(x) ~{}={}~ 0 \quad\text{iff}\quad x ~{}={}~ 0~,\quad \ and
			\end{align*}
		\item 
			$\nabla$ is antitone, i.e.\
			\begin{align*}
				\forall\, 0 < v \leq v'\colon\quad \nabla(v') ~{}\leq{}~ \nabla(v)
			\end{align*}
		\item
			$V$, $p$, and $d$ satisfy a decrease condition, namely
			\begin{align*}
				\iverson{G} \cdot \wp{C}{V} ~{}\leq{}~ V -  \lambda \sigma \textbf{.}\, \nabla\big(V(\sigma)\big)~. \tag*{$\triangle$}
			\end{align*}
	\end{enumerate}
\end{definition}

\begin{theorem}[$\nabla$-Variants Witness Almost-Sure Termination]
\label{thm:n}
If $V$ is a $\nabla$-variant for $\WHILEDO{G}{C}$, then $\WHILEDO{G}{C}$ terminates universally almost-surely.
\end{theorem}
\begin{proof}
See arXiv?!
\end{proof}

\newpage\A{\Divider}
\section{Relation between \Thm{t1651} and Chatterjee and Fu \Cite{C \& F}}

Let $V$ in $\E$ be a ranking super-martingale for program \quad$\WHILEDO{G}{\mathit{Com}}$\, and suppose that there is some $\delta >0$ such that, for all $s$:
\begin{equation}\label{e1723}
\wp{\mathit{Com}}{|V{-}R|}(s) ~~\geq~~ \delta~,
\end{equation}
whenever $V(s)=R$ and $s$ satisfies $G$.   Assume also that there is no non-determinism in ${\mathit{Com}}$.
\Af{I don't think we have to assume this, but I've left it like this for now...}
 Under these assumptions,  $V$ satisfies $p, d$ progress: \Def{def:pd-var}.

\begin{proof}
Let $s^{-}$ be the subset of states defined $s'\in s^{-}$ if and only if the transition probability $p_{s'}= \wp{\mathit{Com}}{\iverson{s'}}(s)>0$ and $V(s') < V(s)$. Similarly let 
 $s^{+}$ be the subset of states defined $s'\in s^{+}$ if and only if $p_{s'}>0$ and $V(s') \geq V(s)$.  Thus $s^{+}\cup s^{-}$ is the subset of states reachable from $s$ in  single execution of ${\mathit{Com}}$, partitioned into those that strictly decrease the super-martingale ($s^{-}$), and those that increase it ($s^{+}$).

Now,  the super-martingale condition, and \Eqn{e1723} imply the following:
\begin{enumerate}
\item $s^{-}$ is not empty;
\item $\sum_{s'\in s^{-}}p_{s'}\cdot V(s') + \sum_{s'\in s^{+}}p_{s'}\cdot V(s') \leq V(s)$~;
\item $\sum_{s'\in s^{-}}p_{s'}\cdot(V(s){-}V(s')) + \sum_{s'\in s^{+}}p_{s'}\cdot (V(s'){-}V(s)) \geq \delta$~.
\end{enumerate}
\medskip

Observe that (2) is equivalent to
\[
\sum_{s'\in s^{-}}p_{s'}\cdot(V(s){-}V(s')) -  \sum_{s'\in s^{+}}p_{s'}\cdot (V(s'){-}V(s)) \geq 0~,
\]
and therefore that (2) and (3) together imply:

\begin{equation}\label{e1753-b}
\sum_{s'\in s^{-}}p_{s'}\cdot(V(s){-}V(s')) \geq \delta/2~.
\end{equation}
\end{proof}

Next we further partition $s^{-}$ into two sets $A$ and $B$, defined $s'\in A$ if $V(s){-}V(s') > \delta/4$ and $s'\in B$ if $V(s){-}V(s') \leq \delta/4$. By \label{e1753} and that $\sum_{s'\in s^{-}}p_{s'} \leq 1$ we must have that $A$ is not empty. We now reason that $\sum_{s'\in A}p_{s'} \geq \delta/(4\cdot V(s))$, as follows:

\begin{Reason}
\StepR{}{\Eqn{e1753-b} and $s^{-} = A \cup B$}
{\sum_{s'\in A}p_{s'}\cdot(V(s){-}V(s')) +  \sum_{s'\in B}p_{s'}\cdot(V(s){-}V(s')) \geq \delta/2}
\StepR{implies}{$(V(s)-V(s') \leq V(s)$}
{\sum_{s'\in A}p_{s'}\cdot V(s) +  \sum_{s'\in B}p_{s'}\cdot(V(s){-}V(s')) \geq \delta/2}
\StepR{implies}{$s'\in B \Rightarrow V(s){-}V(s') \leq \delta/4$}
{\sum_{s'\in A}p_{s'}\cdot V(s) +  \sum_{s'\in B}p_{s'}\cdot \delta/4 \geq \delta/2}
\StepR{implies}{$\sum _{s'\in B }p_{s'} \leq 1$}
{\sum_{s'\in A}p_{s'}\cdot V(s) +  \delta/4 \geq \delta/2}
\StepR{implies}{Arithmetic}
{\sum_{s'\in A}p_{s'}\cdot V(s) \geq \delta/4}
\StepR{implies}{Arithmetic}
{\sum_{s'\in A}p_{s'} \geq \delta/(4\cdot V(s))~.}
\end{Reason}

Observe now that $\sum_{s'\in A}p_{s'}$ is the probability that the super-martingale $V$ is decreased under execution of ${\mathit{Com}}$ by at least $\delta/4$ from a given $s$.  

Finally we conclude that $p(v)= \delta/(4\cdot v) \min 1$ and $d(v)= \delta/4$ together satisfy $p,d$ progress.

\fi 

\newpage
\MakeEndNotes 

\newpage
\MakeSargasso 
\ifNoSargasso\relax\else 
\bigskip
\Ct{\emph{Sargasso} contains stuff that probably won't make it into this paper (or book), and  might even be thrown away eventually. But we keep it  handy in case we change our minds or want to copy-paste bits of it. Putting it \emph{here} during development is really important: it means you don't have to wonder which (other) file it was put in for ``safe keeping". If Sargasso gets to big, is can always be suppressed by putting an \texttt{$\backslash$end\{document\}} just before it. So you don't have to see it every time: but still you know where to look for that interesting thing you wrote a few months ago\ldots}

\C{\include{EndNoteH}}
\B{\include{old-hacky-proof}}

\newpage
{\Cx
\begin{definition}[Characteristic Functionals of $\mathtt{while}$-Loops]
Let $G \subseteq \Sigma$, $C \in \PGCL$, and $f \in \E$. Then we call
\begin{align*}
	\charwp{G}{C}{f}(X) ~{}={}~ \iverson{\neg G}\cdot f + \iverson{G} \cdot \wp{C}{X}
\end{align*}
the characteristic functional of $\wp{\WHILEDO{G}{C}}{f}$.
Note that by the Kleene Fixedpoint Theorem \cite{kleene}, we have
\begin{align*}
	\sup_{n \in \Nats}~\charwpn{G}{C}{f}{n}(\zero) ~{}={}~ \wp{\WHILEDO{G}{C}}{f}~. \tag*{$\triangle$} 
\end{align*}
\end{definition}

\begin{proof}[Proof of \autoref{thm:facts}]
For proving \autoref{thm:facts}.(\ref{thm:facts-1}), consider 
\begin{align*}
	\charwp{G}{C}{\one}(\zero) 	&~{}={}~ \iverson{\neg G}\cdot \one + \iverson{G} \cdot \wp{C}{X} ~{}={}~ \iverson{\neg G}\cdot \iverson{\neg G} + \iverson{G} \cdot \wp{C}{X} ~{}={}~ \charwp{G}{C}{\iverson{\neg G}}(\zero)
\end{align*}
and therefore
\begin{align*}
	\wp{\WHILEDO{G}{C}}{\iverson{\one}}  ~{}={}~ \sup_{n \in \Nats}~\charwpn{G}{C}{\one}{n}(\zero) ~{}={}~ \sup_{n \in \Nats}~\charwpn{G}{C}{\iverson{\neg G}}{n}(\zero) ~{}={}~ \wp{\WHILEDO{G}{C}}{\iverson{\neg G}}~.
\end{align*}
For proving  \autoref{thm:facts}.(\ref{thm:facts-2}), consider the two characteristic functionals
\begin{align*}
	\charwp{0 < f \leq H}{C}{\iverson{f = 0}}(X) &~{}={}~ \iverson{f = 0 \Jr{\land}{\lor} H < f}\cdot \iverson{f = 0} + \iverson{0 < f \leq H} \cdot \wp{C}{X}\\
	&~{}={}~ \iverson{f = 0} + \iverson{0 < f \leq H} \cdot \wp{C}{X}
	\intertext{and}
	\charwp{0 < f}{C}{\iverson{f = 0}}(X) &~{}={}~ \iverson{f = 0} + \iverson{0 < f} \cdot \wp{C}{X}~.
\end{align*}
Since $\iverson{0 < f \leq H} \leq \iverson{0 < f}$, we see that $\charwp{0 < f \leq H}{C}{\iverson{f = 0}}(X) \leq \charwp{0 < f}{C}{\iverson{f = 0}}(X)$ for every $X \in \E$ and thus
\begin{align*}
	\wp{\WHILEDO{0 < f \leq H}{C}}{\iverson{f=0}}  &~{}={}~ \sup_{n \in \Nats}~\charwpn{0 < f \leq H}{C}{\iverson{f = 0}}{n}(\zero) \\
	&~{}\leq{}~ \sup_{n \in \Nats}~\charwpn{0 < f}{C}{\iverson{f = 0}}{n}(\zero) ~{}={}~ \wp{\WHILEDO{0 < f}{C}}{\iverson{f = 0}}~.\tag*{\qedsymbol}
\end{align*}
\renewcommand{\qedsymbol}{}
\end{proof}

{\Cx
\begin{proof}[Alternative proof 1 of \autoref{thm:facts}.(\ref{thm:facts-2})]
We show that 
\begin{align}
	& \wp{\WHILEDO{G\land V{\leq}H}{\mathit{Com}}}{\iverson{\neg G}} \label{e1057-1} \\
	\leq\quad & \wp{\WHILEDO{G}{\mathit{Com}}}{\iverson{\neg G}} \label{e1057-2}
\end{align}
using the general rule for fixed points that $F(g){\leq}g\Implies \mu F{\leq}g$. In this case $f$ is \Eqn{e1057-1} and $F$ its defining functional, with $g,G$ for \Eqn{e1057-2}; and we are showing that $f{\leq}g$.
To apply the general fixed-point rule, we must therefore establish
\begin{align}
	& \wp{(\IF{G\land V{\leq}H}(\COMPOSE{\mathit{Com}}{\WHILEDO{G}{\mathit{Com}})})}{\iverson{\neg G}} \label{e1127-1} \\
	\leq\quad & \wp{\WHILEDO{G}{\mathit{Com}}}{\iverson{\neg G}} \label{e1127-2} \\
	=\quad & \wp{(\IF{G}(\COMPOSE{\mathit{Com}}{\WHILEDO{G}{\mathit{Com}})})}{\iverson{\neg G}}~, \label{e1127-3}
\end{align}
where we have used $g{=}G(g)$ to make rewrite \Eqn{e1127-2} to \Eqn{e1127-3}, whose comparison with \Eqn{e1127-1} is easier.

Now if $V{\leq}H$ in the (initial) state to which these $\wpsymbol$'s are applied, then \Eqn{e1127-1} and \Eqn{e1127-3} are equal, since the\quad$\IF{G\land V{\leq}H}$\quad of the former reduces to the\quad$\IF{G}$\quad of the latter, and they are otherwise textually identical. But if initially $V{>}H$ then \Eqn{e1127-1} is $\iverson{\neg G}$, and we know that \Eqn{e1127-3} is at least $\iverson{\neg G}$ in any state.
\end{proof}
}

{\Cx
\begin{proof}[Alternative proof  2 of \autoref{thm:facts}.(\ref{thm:facts-2})]
Let $A,B$ be any two predicates on the state. We show that 
\begin{align}
	& \wp{\WHILEDO{A\land B}{\mathit{Com}}}{\iverson{\neg A}} \label{e1057B-1} \\
	\leq\quad & \wp{\WHILEDO{A}{\mathit{Com}}}{\iverson{\neg A}} \label{e1057B-2}
\end{align}
using the general rule for fixed points that $F(g){\leq}g\Implies \mu F{\leq}g$. In this case $f$ is \Eqn{e1057B-1} and $F$ its defining functional, with $g,G$ for \Eqn{e1057B-2}; and we are showing that $f{\leq}g$.
To apply the general fixed-point rule, we must therefore establish
\begin{align}
	& \wp{(\IF{A\land B}(\COMPOSE{\mathit{Com}}{\WHILEDO{A}{\mathit{Com}})})}{\iverson{\neg A}} \label{e1127B-1} \\
	\leq\quad & \wp{\WHILEDO{A}{\mathit{Com}}}{\iverson{\neg A}} \label{e1127B-2} \\
	=\quad & \wp{(\IF{A}(\COMPOSE{\mathit{Com}}{\WHILEDO{A}{\mathit{Com}})})}{\iverson{\neg A}}~, \label{e1127B-3}
\end{align}
where we have used $g{=}G(g)$ to make rewrite \Eqn{e1127B-2} to \Eqn{e1127B-3}, whose comparison with \Eqn{e1127B-1} is easier.

Now if $B$ holds in the (initial) state to which these $\wpsymbol$'s are applied, then \Eqn{e1127B-1} and \Eqn{e1127B-3} are equal, since the\quad$\IF{A\land B}$\quad of the former reduces to the\quad$\IF{A}$\quad of the latter, and they are otherwise textually identical. But if initially $\neg B$ then \Eqn{e1127B-1} is exactly $\iverson{\neg A}$, and we know that \Eqn{e1127B-3} is at least $\iverson{\neg A}$ from any state.
\end{proof}

The result might surprise at first, that a stronger loop-guard could induce less- rather than more termination: so let $A$ be ``in the desert'' and $B$ be ``still have water''. Then indeed $f$ terminates more probably than $g$, the wrong way 'round for our claimed inequality, because you might run out of water. But termination alone is not enough: it's loop $g$ that is more likely to get you out of the desert.
}
}

\newpage
{\Cx
\begin{enumerate}
\item\label{ztest} Our argument in this paper is done over semantic, that is $\wpsymbol$-images of $\PGCL$ (syntactic) programs. As is well known, such images for \emph{standard} programs have only finite nondeterminism (given that $\GCL$ itself has only binary demonic choice).
\item Our semantic model \cite{McIver:05a} however uses a semantic space that is characterised directly: the ``output set'' of (sub-)distributions is up-closed, convex-closed and Cauchy closed. All the wp-facts we use here are valid in that model: that is, rather than restrict ourselves to wp-images, instead we are using the (proved) fact that all wp-images of $\PGCL$ are in that more general space.
\item Although it's tempting to think that Cauchy closure (i.e.\ topological/metric) corresponds to countable nondeterminism, it is not so. Consider first
\begin{align*}
	& \ASSIGN{x}{1} \\
	& \WHILE{x\neq0} \\
	& \qquad \PCHOICE{\ASSIGN{x}{0}}{\NF{1}{2}}{\ASSIGN{x}{x + 2}} \\
	& \} \\
	& \NDCHOICE{\ASSIGN{x}{x-1}}{\SKIP}~,
\end{align*}
which has uncountably many demonic choices (over geometric-style discrete distributions)
\footnote{\Cx
Let $b$ be any real number in the unit interval $[0,1]$, and consider its binary expansion $0.b_1b_2\cdots b_n\cdots$. Construct the discrete (but in fact countably infinite) distribution
\begin{align*}
	0{+}b_1 &\quad \AT \NF{1}{2} \\
	2{+}b_2 &\quad \AT \NF{1}{4} \\
	\vdots \\
	n{+}b_n &\quad \AT \NF{1}{2^n} \quad.\\
	\vdots
\end{align*}
\par For every such $b$ the above distribution is a possible result of the program: and they are all different distributions. And there are uncountably many of them; but still the set of them all is Cauchy-closed.
}~
yet is still Cauchy closed. Thus it remains within our scope even though it has uncountable demonically branching.
\par On the other hand, the program\quad``choose $x$ from the natural numbers''\quad has only countably infinite branching (and cannot be written in the language of \autoref{table:wp}). Yet when embedded in the probabilistic model \cite{McIver:05a} the output set of distributions is not closed, and hence it is not within the scope of our result.\par Thus the boundary of our result is not countable vs.\ uncountable branching, but rather closed vs.\ not-closed topologically.
\footnote{\Cx We are also restricted to only discrete distributions (but they can be over a countably infinite set) --- but that is a different issue.}
\item We must have a section along these lines because it is in principle possible that someone will come up with a counter-example to our rule (if indeed there is one) by choosing a transition system which is not Cauchy-closed. Without this section it would not be easy for them to see why our whole paper was not flawed. (Remember, the reason the paper would not be invalidated by such a counter-example is that not-Cauchy-closed programs cannot be written in $\PGCL$ with only finite demonic choice, as we present in \autoref{table:wp}. We must make sure that is understood.)
\item So maybe to think about\ldots\ (but not on the critical path!) Can we give an explanation of the difference between Cauchy-closed and countable demonic choice? It's a nice point, because most people thinking intuitively would pick the second criterion (countable) as the important one; but our maths suggests that it's the first criterion that matters (Cauchy-closure).
\item Also mention the apparent anomalies: (1) the variant is bounded, and it looks like the rule is unsound (example 2/3-1/3 random walker in the arXiv); (2) the rule looks unnecessarily weak (when $d(V)>V$, so the process ``cannot possibly decrease $V$ that much, since it would end up less than zero).
\item Explain why we don't use $\awpsymbol$ everywhere.
\end{enumerate}
}

\newpage
{\Bx
\section{Preliminaries}
\Cf{This section copied here 170613.}
Let $\Vars$ denote the set of program variables and let
\begin{align*}
	\Sigma ~{}={}~ \{\sigma ~|~ \sigma\colon \Vars \To \mathbb{Q}\}
\end{align*}
denote the set of program states.
We call a subset $G \subseteq \Sigma$ of program states a \emph{predicate}.
We use Iverson bracket notation $\iverson{G}$ to denote the indicator function of a predicate $G$\C{, that is 1 on those states where $G$ holds and 0 otherwise}.
\Cf{I know this as ``characteristic'' function: hadn't heard of its being named after Iverson \smiley.}
An \emph{expectation} is a random variable that maps program states to non-negative reals or $\infty$.
\Cf{In our program logic we do not allow expectations to be infinite or even unbounded: there must be a single upper bound for the whole state space $\Sigma$. Is that problematic? \B{I think not. We just have to ensure that in your proof rule (Theorem~\ref{thm:lem-2-4-1}) the invariant is bounded, but I already added that condition in this paper. \Cx I agree --- but we might want to put a footnote somewhere to save readers from having to figure that out for themselves. Perhaps leave this footnote here util we do that.}}
More formally, we have the following definition:
\begin{definition}[Expectations~\textnormal{\cite{DBLP:series/mcs/McIverM05}}]
\label{def:expectations}

	The set of expectations on $\Sigma$, denoted by $\E$, is defined as
	\begin{align*}
		\E ~{}={}~ \left\{f ~\middle|~ f\colon \Sigma \To \Rposinf\right\} ~.
	\end{align*}
	For $f\in \E$, we write $f < \infty$ iff $f$ is everywhere non-infinite, i.e.\ $\forall\, \sigma \in \Sigma\colon~f(\sigma) \neq \infty$.
	We furthermore say that $f \in \E$ is bounded iff there exists an upper bound $b \in \Rpos$, s.t.\
	\Cf{If we run short of space, we might consider running-in small bits of maths like this one.}
	\begin{align*}
		\forall \sigma \in \Sigma\colon \quad f(\sigma) \leq b~.
	\end{align*}
	
	There is a natural complete partial order $\leq$ on $\E$ obtained by pointwise lifting the complete partial order on $\Rposinf$, i.e.\
	\begin{align*}
		f_1 ~{}\leq{}~f_2 \quad\text{iff}\quad \forall \sigma\in\Sigma\colon~~ f_1(\sigma) ~{}\leq{}~ f_2(\sigma) ~. \tag*{$\triangle$}
	\end{align*}
\end{definition}
\Cf{Need to say what $\PGCL$ is before writing a definition about it. Just say ``it's the programming language defined in \autoref{table:wp}''? What has to be clear before reading \Def{def:expectations} is that $\PGCL$ is syntax.}
\begin{definition}[The $\wpsymbol$-Transformer~\textnormal{\cite{DBLP:series/mcs/McIverM05}}]
\label{def:wp}
	The weakest pre-expectation transformer $\wpsymbol\colon \PGCL \To \E \To \E$ is defined by structural induction on all $\PGCL$ programs according to the rules given in \autoref{table:wp}.
	\hfill$\triangle$
\end{definition}
\begin{definition}[Universal Almost-Sure Termination]
	\label{def:uast}
	A program $C \in \PGCL$ is said to terminate universally almost-surely iff
	\begin{align*}
		\wp{C}{\one} ~{}={}~ \one\tag*{$\triangle$}\quad.
	\end{align*}
\end{definition}
\begin{table}[t]
\renewcommand{\arraystretch}{1.5}
\begin{tabular}{@{\hspace{1em}}l@{\hspace{2em}}l@{\hspace{1em}}}
	\hline\hline
	$\boldsymbol{C}$			& $\boldsymbol{\textbf{\textsf{wp}}\,\textbf{.}\,C\,\textbf{.}\, f}$\\
	\hline\hline
	$\SKIP$					& $f$ \\
	$\ASSIGN{x}{e}$			& $f\subst{x}{e}$ \\
	$\ITE{G}{C_1}{C_2}$		& $\iverson{G} \cdot \wp{C_1}{f} + \iverson{\neg G} \cdot \wp{C_2}{f}$ \\
	$\PCHOICE{C_1}{p}{C_2}$	& $p \cdot \wp{C_1}{f} + (1 - p) \cdot \wp{C_2}{f}$ \\
	$\NDCHOICE{C_1}{C_2}$	& $\min\left\{ \wp{C_1}{f},\, \wp{C_2}{f} \right\}$ \\
	$\COMPOSE{C_1}{C_2}$		& $\wp{C_1}{\big(\wp{C_2}{f}\big)}$ \\
	 \makebox[0pt][r]{\Ct{Consider $\lfp X\Spot~xxx$~?}\quad}%
	 $\WHILEDO{G}{C'}$			& $\lfp X\C{\textbf{.}}~\iverson{\neg G} \cdot f + \iverson{G} \cdot \wp{C'}{X}$\\
	\hline
\end{tabular}

\vspace{1ex}
\Ct{Say somewhere that $\cdot$ is multiplication.}
\caption{Rules for the $\wpsymbol$-transformer.}
\label{table:wp}
\end{table}
\begin{theorem}[Basic Properties of $\wpsymbol$~\textnormal{\cite{DBLP:series/mcs/McIverM05}}]
	\label{thm:basic-prop}
	For all $C \in \PGCL$ and $f, f_1, f_2 \in \E$, the following holds:
	\begin{enumerate}
		\item \qquad \label{thm:basic-prop-mon}
			$f_1 ~\leq~ f_2 \quad\text{implies}\quad \wp{C}{f_1} ~{}\leq{}~ \wp{C}{f_2}$ \qquad(monotonicity)\\[-.75em]
		\item \qquad \label{thm:basic-prop-scaling}
			$\forall r \in \Rpos\colon \quad \wp{C}{(r \cdot f)} ~{}\leq{}~ r \cdot \wp{C}{f}$ \qquad(scaling)\\[-.75em]
		\item \qquad \label{thm:basic-prop-feasibility}
			If function $f$ satisfies $f ~{}\leq{}~ b$ for some scalar $b \in \Rpos$, then also $\wp{C}{f} ~{}\leq{}~ b $ \qquad(feasibility)
	\end{enumerate}
\end{theorem}
\begin{theorem}[Technical Facts about $\wpsymbol$]
	\label{thm:facts}
	For all $G \subseteq \Sigma$, $C \in \PGCL$, $f \in \E$, and $H \in \Rpos$, we have:
	\begin{enumerate}
		\item \qquad \label{thm:facts-1}
			$\wp{\WHILEDO{G}{C}}{\one} ~{}={}~ \wp{\WHILEDO{G}{C}}{\iverson{\neg G}}$\\[-.75em]
		\item \qquad \label{thm:facts-2}
			$\wp{\WHILEDO{0 < f \leq H}{C}}{\iverson{f = 0}} ~{}\leq{}~ \wp{\WHILEDO{0 < f}{C}}{\iverson{f = 0}} $ 
	\end{enumerate}
\end{theorem}
\begin{proof}
See Appendix.
\end{proof}
%
%


%
%
\begin{definition}[Probabilistic Invariants~\protect{~\textnormal{\cite[\textnormal{p.\ 39, Definition 2.2.1}]{DBLP:series/mcs/McIverM05}}}]
Let $G$ be a predicate and $C$ be a $\pGCL$ program.
Then $I \in \E$ is a probabilistic invariant of the loop $\WHILEDO{G}{C}$, iff $I$ is \emph{bounded} and
\begin{align*}
	\iverson{G} \cdot I ~\leq~ \wp{C}{I}~. \tag*{$\triangle$}
\end{align*}
\C{In this case we say that $I$ is \emph{preserved} by each iteration of $\WHILEDO{G}{C}$.}
\footnote{\Cx Compare standard Hoare/Dijkstra -style weakest preconditions, where $I$'s being preserved is $G\land I\Implies\wp{C}{I}$ \Cite{ADoP}, and recall that $A{\Implies}B$ just when $\iverson{A}{\leq}\iverson{B}$.}
\end{definition}
\begin{theorem}[Total Correctness for Probabilistic Loops~\protect{~\textnormal{\cite[\textnormal{p.\ 43, Lemma 2.4.1, Case 2.}]{DBLP:series/mcs/McIverM05}}}]
\label{thm:lem-2-4-1}
	Let $\mathit{Term}$ be a predicate such that $\iverson{\mathit{Term}} \leq \wp{\WHILEDO{G}{C}}{\one}$, i.e.\ from any initial state satisfying $\mathit{Term}$ the loop terminates almost-surely,
	\Cf{Maybe remark on the difference between this and \Def{def:uast}.}
	and let $I \in \E$ be a probabilistic invariant of $\WHILEDO{G}{C}$.
	Then 
	\begin{align*}
		\iverson{Term} \cdot I ~\leq~ \wp{\WHILEDO{G}{C}}{(\iverson{\neg G} \cdot I)}~.
	\end{align*}
\end{theorem}
Intuitively, this theorem states that if $I$ is a probabilistic invariant that is preserved in each iteration of the loop body, then also the whole loop ``preserves" $I$ from all states from which the loop terminates almost-surely.
\begin{theorem}[Variant Rule for Loops~\protect{~\textnormal{\cite[\textnormal{p.\ 55, Lemma 2.7.1}]{DBLP:series/mcs/McIverM05}}}]
\label{thm:lem-2-7-1}
	Let $\mathit{Inv}, \mathit{Guard} \subseteq \Sigma$ be predicates,
	let $\mathit{VInt}\colon \mathit{Inv} \To \Ints$ be an integer-valued expression,
	let $\mathit{Low}, \mathit{High} \in \Ints$ be fixed integers,
	let $\varepsilon \in (0,\, 1]$ be a fixed probability, and
	let $\mathit{Com}$ be a $\pGCL$ program. \Cf{Why $\mathit{Com}$ rather thank our earlier $C$?}
	Then
	\begin{enumerate}
		\item\label{i1657-1}
			$\iverson{\mathit{Guard}} \cdot \iverson{Inv} ~\leq~ \iverson{\mathit{Low} \leq \mathit{VInt} < \mathit{High} }$, and
		\item\label{i1657-2}
			$\iverson{\mathit{Inv}}$ being a probabilistic invariant for the loop $\WHILEDO{\mathit{Guard}}{\mathit{Com}}$, and
		\item\label{i1657-3}
			$\forall N \in \Ints\colon \quad \varepsilon \cdot \iverson{\mathit{Guard}} \cdot \iverson{\mathit{Inv}} \cdot \iverson{\mathit{VInt} = N} ~\leq~ \wp{\mathit{Com}}{\iverson{\mathit{VInt} < N}}$
	\end{enumerate}
	all together imply
	\begin{align*}
		\iverson{\mathit{Inv}} ~\leq~ \wp{\WHILEDO{\mathit{Guard}}{\mathit{Com}}}{1}~.
	\end{align*}
\end{theorem}
Intuitively the theorem states that if a standard invariant $\mathit{Inv}$ (i.e.\ $\mathit{Inv}$ is a predicate and a probabilistic invariant) and a bounded integer-valued variant $\mathit{VInt}$ are given and if one iteration of the loop body $\mathit{Com}$
\Cf{\label{n1646}``\ldots is guaranteed with probability at least $\varepsilon$ to decrease the variant by at least 1'' is how this is usually said. Are we saying it this way deliberately? Are they actually the same?}
\C{[}decreases the variant by at least a constant $\varepsilon$ in expectation\C{]}, then the loop $\WHILEDO{\mathit{Guard}}{\mathit{Com}}$ terminates almost-surely from any state satisfying $\mathit{Inv}$.

\section{Proof Rules for Almost-Sure Termination}
\Cf{This section copied here 170613.}

%
%
\begin{definition}[Quasi Variants]
\label{def:quasi-var}
$V \in \E$, with $V < \infty$, is called a \emph{quasi variant for $\WHILEDO{G}{C}$}, iff
\begin{enumerate}
		\item
			$V$ indicates termination, i.e.\ 
			\Cf{More conventional (and just as good) is $V{=}0 \Implies \neg G$, i.e.\ that driving $V$ to 0 ensures loop exit.}
			\begin{align*}
				\iverson{\neg G} ~{}={}~ \iverson{V = 0}~, \quad\text{and}
			\end{align*}
		\item
			$V$ satisfies a super-martingale property, that is
			\begin{align*}
				\iverson{G} \cdot \awp{C}{V} ~{}\leq{}~ V~
			\end{align*}
			which is equivalent to
			\begin{align*}
				\forall r \in \Rpos\colon \quad \iverson{0 < V \leq r} \cdot (r - V) ~\leq~ \wp{C}{(r \ominus V)}~,
			\end{align*}
			where $f_1 \ominus f_2 = \Max{f_1 - f_2}{0}$. \hfill$\triangle$
			\Cf{I think maybe the second definition should be given first, and then explained in English. The connection, definition and proof of its awp formulation should be self-contained, possibly in the appendix. \B{I would argue that the awp version should be used in all examples because it is much nicer in practice.}}
	\end{enumerate}
\end{definition}
\Jf{I guess we'll stick to one definition at some point and then have a lemma asserting that there is an alternative equivalent characterisation. Personally, I feel that the earlier one is easier to grasp at the expense that one has define awp (which however is easy). \Cx Annabelle and I are tending to the opposite view, which we'll explain in more detail. But basically although the advantage of awp is indeed that it's easy to understand, its disadvantage however is that we do not have a published, logic complete with loop rules, that includes both wp and awp. We have it only for wp. The great advantage of using our logic exactly (and its loop rules) is (1) that we get shorter, more high-level proofs and (2) we do not have to build (or make proofs in this paper about) how nondeterminism interacts with probability. Recall that in Holger's paper they did have to do that, and it increased the complexity of their exposition considerably. So we are going to suggest that we don't use awp in the proof, but that we do use Benjamin's nice proof that the super-martingale property can equivalently be expressed using awp. \Jx Fine. \CBar}
\begin{definition}[$p,d$-Variants]
\label{def:pd-var}
	A quasi variant $V$ for $\WHILEDO{G}{C}$ is called a \emph{$p,d$-variant for $\WHILEDO{G}{C}$}, iff there exist two functions $p\colon \Rpos \To [0,\, 1]$ and $d\colon \Rpos \To \Rpos$, such that
\Cf{Both $p,d$ should be strictly positive. \Bx Both $p$ and $d$ can be $0$ on input $0$. It makes everything more defined. \Cx But then they are not antitone. If we never use them at 0, maybe we shouldn't define them there. \B{They are antitone above zero as currently defined. The problem I had with undefined-at-0 is: what is the value of $\iverson{G}\cdot \lambda \sigma. p(V(\sigma))$ if $\sigma \models \neg G$? Strictly speaking $0 \cdot p(0) = 0 \cdot \text{undefined} = \text{undefined}$, which I found undesirable... :-/}}	
	\begin{enumerate}
		\item \label{def:pd-var-1}
			$p$ and $d$ are strictly positive above 0, i.e.
			\begin{align*}
				\forall\, 0 < v\colon \qquad p(v) ~{}>{}~ 0 \quad\text{and}\quad d(v) ~{}>{}~ 0~,
			\end{align*}
		\item \label{def:pd-var-2}
			$p$ and $d$ are antitone above 0, i.e.\
			\begin{align*}
				\forall\, 0 < v \leq v'\colon\qquad p(v') ~{}\leq{}~ p(v) \quad\text{and}\quad d(v') ~{}\leq{}~ d(v)~,
			\end{align*}
		\item \label{def:pd-var-3}
			$V$, $p$, and $d$ satisfy a progress condition, namely
			\begin{align*}
				\iverson{G} \cdot \lambda \sigma \textbf{.}\, p\big(V(\sigma)\big) \Wide{\leq} \lambda \sigma\textbf{.}\, \wp{C}{\iverson{V \leq V(\sigma){-}d\big(V(\sigma)\big)}}~. \tag*{$\triangle$}
			\end{align*}
	\end{enumerate}
\end{definition}

\Jf{
To what extent does the result below enable us to cover a larger set of programs to be handled as Charakov and Sankarananayan could handle in their CAV'13 paper using super-martingales? Would be nice of at least an example could be given that can be handled now, but not with their approach. Or obtain a result that this rules captures a strictly larger set of programs. \Cx{ I believe this kind of comparison is covered in our \textit{arXiv} paper: Annabelle will check. }
\Ax We have looked at this a bit, and have the following results (as of January at least!)

Wrt.\ Chakarov, the example which we can do that they can't is the 1-dimensional random walk which goes up or down with probability 1/2. I believe also that Chakarov's rules do not treat nondeterminism.

Moreover, via a corollary of a result by Blackwell we can show that if the super-martingale is in fact a non-constant  and bounded \C{\emph{exact martingale} (do you mean, Annabelle?)}, then the program \emph{does not} terminate with probability 1.  We have two random-walker type examples to contrast with some work of Chatterjee et al.\ (on refuting martingales). The first one terminates with probability 1, but has unbounded time to termination; we can show termination using our rule.
The second example is a walker which does not terminate at all: we can show this by exhibiting a bounded, non-constant martingale. In both cases the refutation result of Chatterjee et al.\ imply \emph{only} that neither terminates in finite expected time.

There is also some even more recent work by Chatterjee which I have to look at\ldots

I have looked at this paper again --- it's called ``Termination of Nondeterministic Recursive Probabilistic Programs'' and was put in the archive on 11 January 2017. Maybe they will submit it to POPL too?

 It seems to do the following: it proves soundness and completeness for ranking super-martingales for programs that terminate in finite expected time. (Apparently they found some problem with Hermanns et al.'s result on this.) They also look again at almost-sure termination and produce some results for MDP's using what they call ``difference-bounded'' martingales and super-martingales, thus expanding the martingale-type techniques for almost-sure termination. Using these ideas they are able to prove termination of the 1-dimensional random walk; however it needs a result on ``tail probabilities''.
 
In contrast \underline{our rule does not need tail probabilities} --- i.e.\ it is designed to use local reasoning on the program code.  We also explore the question of completeness and have some illuminating results there too.
 }

\begin{theorem}[$p,d$-Variants Witness Universal Almost-Sure Termination]
\label{xthm:pd}
If $C$ terminates universally almost-surely and if $V$ is a $p,d$-variant for $\WHILEDO{G}{C}$, then $\WHILEDO{G}{C}$ terminates universally almost-surely.
\end{theorem}
\begin{proof}
Our proof consists of two parts:
\begin{enumerate}
	\item 
		We fix any strictly positive real number $H \in \Rpos \setminus \{0\}$ and consider the set of states described by the predicate $0 < V \leq H$.
		We then prove that the loop $\WHILEDO{G}{C}$ will almost-surely leave that set, by proving that another, related loop terminates almost-surely.
		For that we will exploit \Thm{thm:lem-2-7-1}.
	
	\item
		We prove that almost-sure escape from the set described by $0 < V \leq H$ for any $H$, as just above, implies \emph{universal almost-sure termination} of the loop $\WHILEDO{G}{C}$. For that we will exploit \Thm{thm:lem-2-4-1}.
\end{enumerate}
Here are the details.
Let $V$ be the $p,d$-variant for $\WHILEDO{G}{C}$.
Fix any arbitrary strictly positive $H \in \Rpos \setminus \{0\}$.
We can express the probability that the execution of $\WHILEDO{G}{C}$ leaves the set described by $0 < V \leq H$ in terms of the termination probability of another loop, namely
\begin{Equation}\label{e1412A}
	\WHILEDO{0 < V \leq H}{C}~,
\end{Equation}
and we now show that probability to be 1.

\paragraph{\textbf{(1) The loop \Eqn{e1412A} escapes $\boldsymbol{\iverson{0 < V \leq H}}$ almost-surely for any $\boldsymbol{H}$:}}

We will now prove that the termination probability of the \emph{above} loop is 1, meaning that the \emph{original} loop escapes the set described by $0 < V \leq H$ almost-surely.
For that, we instantiate the variables of Theorem~\ref{thm:lem-2-7-1} as follows:
\begin{align*}
	\mathit{Inv} ~=~ \true \qquad 
	\mathit{Guard} ~=~ 0 < V \leq H \qquad 
	\mathit{VInt} ~=~ \left\lceil \frac{V}{d(H)} \right\rceil \\
	\mathit{Low} ~=~ 0 \qquad 
	\mathit{High} ~=~ \left\lceil \frac{H}{d(H)} \right\rceil + 1 \qquad 
	\varepsilon ~=~ p(H) \qquad 
	\mathit{Com} ~=~ C
\end{align*}
%
%
%
%
%
%
%
%
Intuitively, $\mathit{VInt}$ can be thought of as a \emph{discretized} version of $V$. 
On the interval $(0,\, H]$ this discretized variant $\mathit{VInt}$ will decrease \Cd{in expectation}\Cf{See \Note{n1646}.}  by at least 1 \Cd{(or drop to 0)}\Cf{If it drops to 0, then it must have decreased by at least 1\ldots?} with probability $\varepsilon = p(H)$.

We now verify that our above choices satisfy the preconditions of Theorem~\ref{thm:lem-2-7-1}:
\begin{enumerate}
	\item 
		First observe that $V \leq H$ implies $\left\lceil\sfrac{V}{a} \right\rceil \leq \left\lceil\sfrac{H}{a} \right\rceil$ for any strictly positive $a$.
		Since $d(H)$ is indeed strictly positive, we therefore have
		\Cf{The calculations just here could just as easily be over predicates, and there would then be less ``Iverson clutter''.}
		\begin{align*}
			& \iverson{V \leq H} ~\leq~ \iverson{\left\lceil \frac{V}{d(H)} \right\rceil \leq \left\lceil \frac{H}{d(H)} \right\rceil } \\
			\implies \quad & \iverson{V \leq H} ~\leq~ \iverson{\left\lceil \frac{V}{d(H)} \right\rceil < \left\lceil \frac{H}{d(H)} \right\rceil + 1 }\\
			\Longleftrightarrow \quad & \iverson{V \leq H} ~\leq~ \iverson{\mathit{VInt} < \mathit{High} } 
			\tag*{$\left( \mathit{VInt} = \left\lceil \frac{V}{d(H)} \right\rceil \text{ and } \mathit{High} = \left\lceil \frac{H}{d(H)} \right\rceil + 1 \right)$}\\
			\Longleftrightarrow \quad & \iverson{0 \leq V \leq H} ~\leq~ \iverson{0 \leq \mathit{VInt} < \mathit{High} } 
			\tag{$V \geq 0$ and $\mathit{VInt} \geq 0$} \\
			\implies \quad & \iverson{0 < V \leq H} ~\leq~ \iverson{0 \leq \mathit{VInt} < \mathit{High} } \\
			\Longleftrightarrow \quad & \iverson{\mathit{Guard}} ~\leq~ \iverson{\mathit{Low} \leq \mathit{VInt} < \mathit{High} } 
			\tag{$\mathit{Guard} = 0 < V \leq H$ and $\mathit{Low} = 0$}\\
			\Longleftrightarrow \quad & \iverson{\mathit{Guard}} \cdot \iverson{\true}~\leq~ \iverson{\mathit{Low} \leq \mathit{VInt} < \mathit{High} } \\
			\Longleftrightarrow \quad & \iverson{\mathit{Guard}} \cdot \iverson{\mathit{Inv}}~\leq~ \iverson{\mathit{Low} \leq \mathit{VInt} < \mathit{High} }~, 
			\tag{$\mathit{Inv} = \true$}
		\end{align*}
		which is an instance of \Thm{thm:lem-2-7-1}\Itm{i1657-1}.

	\item \Cf{Put this first.}
		We assume that the loop body $C$ terminates universally almost-surely. Thus
		\begin{align*}
			&1 ~=~ \wp{\mathit{C}}{1}\\
			\implies \quad & \iverson{\mathit{Guard}} \cdot \iverson{\mathit{Inv}} ~\leq~ \wp{\mathit{C}}{1}\\
			\Longleftrightarrow \quad & \iverson{\mathit{Guard}} \cdot \iverson{\mathit{Inv}} ~\leq~ \wp{\mathit{C}}{\iverson{\true}}\\
			\Longleftrightarrow \quad & \iverson{\mathit{Guard}} \cdot \iverson{\mathit{Inv}} ~\leq~ \wp{\mathit{Com}}{\iverson{\mathit{Inv}}}~, \tag{$\mathit{Com} = C$ and $\mathit{Inv} = \true$}
		\end{align*}
		which is an instance of \Thm{thm:lem-2-7-1}\Itm{i1657-2}.

		The second condition we have to ensure in order for $\iverson{\mathit{Inv}}$ to be a probabilistic invariant is that $\iverson{\mathit{Inv}}$ is bounded, which is trivially true since it is bounded by 1:
		\begin{align*}
			\iverson{\mathit{Inv}} \eeq \iverson{\true} \eeq 1 \lleq 1
		\end{align*}

	\item Starting from our progress condition imposed by $V$ being a $p,d$-variant, we verify the last precondition:
	\Cf{Slightly puzzled by the role of the $\lambda\sigma$'s ---
	for example, why do we not write on the left of the first line $\lambda \sigma \textbf{.}\, \iverson{G(\sigma)} \cdot p(V(\sigma))$?}
		\begin{align*}
			&\iverson{G} \cdot \lambda \sigma \textbf{.}\, p\big(V(\sigma)\big) ~\leq~ \lambda \sigma\textbf{.}\, \wp{C}{\iverson{V \leq V(\sigma) - d\big(V(\sigma)\big)}} \tag{the progress condition}\\
			\Longleftrightarrow \quad &\iverson{0 < V} \cdot \lambda \sigma \textbf{.}\, p\big(V(\sigma)\big) ~\leq~ \lambda \sigma\textbf{.}\, \wp{C}{\iverson{V \leq V(\sigma) - d\big(V(\sigma)\big)}} \tag{$\iverson{V = 0} = \iverson{\neg G}$} \\
			\C{\implies} \quad &\iverson{0 < V \leq H} \cdot \lambda \sigma \textbf{.}\, p\big(V(\sigma)\big) ~\leq~ \lambda \sigma\textbf{.}\, \wp{C}{\iverson{V \leq V(\sigma) - d\big(V(\sigma)\big)}} \tag{$\iverson{0 < V \leq H} \leq \iverson{0 < V}$} \\
			\implies \quad &\iverson{0 < V \leq H}\cdot \lambda \sigma \textbf{.}\, p\big(V(\sigma)\big) ~\leq~ \lambda \sigma\textbf{.}\, \wp{C}{\iverson{\left\lceil \frac{V}{d(H)} \right\rceil \leq \left\lceil \frac{V(\sigma)}{d(H)} - \frac{d\big(V(\sigma)\big)}{d(H)}\right\rceil}} \tag{\Cx monotonicity}\\
			\implies \quad &\iverson{0 < V \leq H}\cdot \lambda \sigma \textbf{.}\, p\big(V(\sigma)\big) ~\leq~ \lambda \sigma\textbf{.}\, \wp{C}{\iverson{\left\lceil \frac{V}{d(H)} \right\rceil \leq \left\lceil \frac{V(\sigma)}{d(H)}\right\rceil - 1}} \tag{$V \leq H$\C{(?)} implies $\lambda \sigma\textbf{.}\, d\big(V(\sigma)\big) \geq d(H)$} \\
			\Longleftrightarrow \quad &\iverson{0 < V \leq H}\cdot \lambda \sigma \textbf{.}\, p\big(V(\sigma)\big) ~\leq~ \lambda \sigma\textbf{.}\, \wp{C}{\iverson{\mathit{VInt} \leq \mathit{VInt}(\sigma) - 1}} 
			\tag*{$\left( \mathit{VInt} = \left\lceil \frac{V}{d(H)} \right\rceil \right)$}\\
			\C{\Longleftrightarrow} \quad &\iverson{0 < V \leq H}\cdot \lambda \sigma \textbf{.}\, p\big(V(\sigma)\big) ~\leq~ \lambda \sigma\textbf{.}\, \wp{C}{\iverson{\mathit{VInt} < \mathit{VInt}(\sigma)}} \\
			\implies \quad &\iverson{0 < V \leq H}\cdot p(H) ~\leq~ \lambda \sigma\textbf{.}\, \wp{C}{\iverson{\mathit{VInt} < \mathit{VInt}(\sigma)}} \tag{$V \leq H$\C{(?)} implies $\lambda \sigma\textbf{.}\, p\big(V(\sigma)\big) \geq p(H)$} \\
			\implies \quad &\forall N \in \Ints\colon \quad \iverson{0 < V \leq H}\cdot p(H) \cdot \iverson{\mathit{VInt} = N}~\leq~  \iverson{\mathit{VInt} = N} \cdot \lambda \sigma\textbf{.}\, \wp{C}{\iverson{\mathit{VInt} < \mathit{VInt}(\sigma)}} \\
			\Longleftrightarrow\textrm{\Cx?} \quad &\forall N \in \Ints\colon \quad \iverson{0 < V \leq H}\cdot p(H) \cdot \iverson{\mathit{VInt} = N}~\leq~  \iverson{\mathit{VInt} = N} \cdot \wp{C}{\iverson{\mathit{VInt} < N}} \tag{\Cx choose $\sigma$?}\\
			\implies \quad &\forall N \in \Ints\colon \quad \iverson{0 < V \leq H}\cdot p(H) \cdot \iverson{\mathit{VInt} = N}~\leq~ \wp{C}{\iverson{\mathit{VInt} < N}} \\
			\Longleftrightarrow \quad &\forall N \in \Ints\colon \quad p(H) \cdot \iverson{0 < V \leq H} \cdot 1 \cdot \iverson{\mathit{VInt} = N}~\leq~ \wp{C}{\iverson{\mathit{VInt} < N}} \\
			\Longleftrightarrow \quad &\forall N \in \Ints\colon \quad \varepsilon \cdot \iverson{\mathit{Guard}} \cdot \iverson{\mathit{Inv}} \cdot \iverson{\mathit{VInt} = N}~\leq~ \wp{\mathit{Com}}{\iverson{\mathit{VInt} < N}} \tag{$p(H) = \varepsilon$, $\mathit{Guard} = 0 < V \leq H$, $\mathit{Inv} = \true$, and $\mathit{Com} = C$}
		\end{align*}
\end{enumerate}
We have thus shown that all preconditions of Theorem~\ref{thm:lem-2-7-1} are satisfied which in conclusion gives us
\begin{align*}
	&\iverson{\mathit{Inv}} ~\leq~ \wp{\WHILEDO{\mathit{Guard}}{\mathit{Com}}}{1} \tag{consequence of Theorem~\ref{thm:lem-2-7-1}}\\
	\Longleftrightarrow \quad &\iverson{\true} ~\leq~ \wp{\WHILEDO{0 < V \leq H}{C}}{1} \tag{$\mathit{Inv} = \true$, $\mathit{Guard} = 0 < V \leq H$, and $\mathit{Com} = C$} \\
	\Longleftrightarrow \quad &\iverson{\true} ~=~ \wp{\WHILEDO{0 < V \leq H}{C}}{1}~.
\end{align*}
\paragraph{\textbf{(2) The loop terminates almost-surely:}}


We now use \autoref{thm:lem-2-4-1} to show that almost-sure escape from $0 < V \leq H$ for arbitrary $H$ implies escape from $0 < V$ in general, i.e.\ with no upper bound $H$ at all.
For the probabilistic invariant $I$ in \autoref{thm:lem-2-4-1} we choose $I = H \ominus V$.
The fact that $H \ominus V$ is a probabilistic invariant follows from the fact that $I$ is bounded by $H$ and from the super-martingale property of $V$: for any non-negative real $r$ we have
\begin{align*}
	&
	\iverson{0 < V \leq r} \cdot (r - V) ~\leq~ \wp{C}{(r \ominus V)} \tag{$V$ is a super-martingale} \\
	~\implies\quad & \iverson{0 < V \leq H} \cdot (H - V) ~\leq~ \wp{C}{(H \ominus V)} \tag{instantiate $r$ with our particular choice $H$} \\
	~\Longleftrightarrow\quad & \iverson{0 < V \leq H} \cdot (H \ominus V) ~\leq~ \wp{C}{(H \ominus V)} \\
	~\Longleftrightarrow\quad & H \ominus V \text{ is a probabilistic invariant of } \WHILEDO{0 < V \leq H}{C}~.
\end{align*}
Thus $\mathit{Term} = \true$, $G = (0 < V \leq H)$, $C$, and $I = H \ominus V$ satisfy the preconditions of \autoref{thm:lem-2-4-1}, and we conclude
\begin{align*}
	&\iverson{\true} \cdot H \ominus V ~\leq~ \wp{\WHILEDO{0 < V \leq H}{C}}{\big(\iverson{\neg (0 < V \leq H)} \cdot H \ominus V\big)} \tag{by \autoref{thm:lem-2-4-1}} \\
	\Longleftrightarrow\quad & H \ominus V ~\leq~ \wp{\WHILEDO{0 < V \leq H}{C}}{\big(\iverson{0 = V ~\vee~ H < V} \cdot H \ominus V\big)} \\
	\Longleftrightarrow\quad & H \ominus V ~\leq~ \wp{\WHILEDO{0 < V \leq H}{C}}{\big(\iverson{0 = V } \cdot H\big)} \\
	\Longleftrightarrow\quad & \one \ominus \NF{V}{H} ~\leq~ \wp{\WHILEDO{0 < V \leq H}{C}}{\big(\iverson{0 = V }\big)}~. \tag{$\dagger$; by \autoref{thm:basic-prop}.(\ref{thm:basic-prop-scaling})}
\end{align*}
Now finally we can reason about the loop $\WHILEDO{G}{C}$ in which we were originally interested. Consider the following:
\begin{align*}
			&\wp{\WHILEDO{G}{C}}{\one}\\
	{}={}~	&\wp{\WHILEDO{G}{C}}{\iverson{\neg G}} \tag{by \autoref{thm:facts}.(\ref{thm:facts-1})}\\
	{}={}~	&\wp{\WHILEDO{0 < V}{C}}{\iverson{V = 0}} \tag{$V$ indicates termination}\\
	\geq{}~	&\wp{\WHILEDO{0 < V \leq H}{C}}{\iverson{V = 0}} \tag{by \autoref{thm:facts}.(\ref{thm:facts-2})}\\
	{}={}~	&\one \ominus \NF{V}{H} \tag{by Equation $\dagger$ above}
\end{align*}

Overall, we have established $\wp{\WHILEDO{G}{C}}{\one} \geq \one \ominus \frac{V}{H}$ for any $H \in \Rpos$, in particular for arbitrarily large $H$.
If we now fix any initial state $\sigma$, we can take the supremum over $H$ and get
\begin{align*}
	(\wp{\WHILEDO{G}{C}}{\one}) (\sigma) ~{}\geq{}~ \underbrace{\sup_{H \in \Rpos} 1 - \NF{V(\sigma)}{H}}_{{}=\one} ~.
\end{align*}

Finally, since $\wp{C'}{\one} \leq \one$ for any $C'$ (see \autoref{thm:basic-prop}.(\ref{thm:basic-prop-feasibility})),
in particular for $C' = \WHILEDO{G}{C}$, and $\sup_{H \in \Rpos}~ 1 - \NF{V(\sigma)}{H} = 1$, we obtain $(\wp{\WHILEDO{G}{C}}{\one}) (\sigma) = 1$ for all states $\sigma$ and hence
\begin{align*}
	\wp{\WHILEDO{G}{C}}{\one} ~{}={}~ \one~,
\end{align*}
i.e.\ $\WHILEDO{G}{C}$ terminates universally almost-surely.
\end{proof}
}

\fi 

\end{document}